\documentclass[a4paper,10pt,table]{article}
\usepackage{amssymb,amsmath,latexsym,color,amsthm,authblk,mathpazo,palatino,bbding,setspace,datetime,breakcites,hyphenat,microtype,diagbox,tikz,xcolor,tkz-graph,subcaption,graphicx,titlesec}
\usepackage[round,authoryear]{natbib}
\usepackage{hyperref,theoremref}
\usepackage[titletoc,title]{appendix}
\usepackage{mathtools}
\usepackage{bm}
\usepackage{accents}
\usepackage{esvect}
\usepackage{mathtools}
\usepackage{bm}
\usepackage{esvect}
\usepackage{setspace}
\usepackage{microtype}
\usepackage[normalem]{ulem}

\usepackage{subcaption}

\usepackage[shortlabels]{enumitem}
\usepackage{float}
\usepackage{amssymb}
\restylefloat{table}

\definecolor{armygreen}{rgb}{0.29, .8, 0.13}
\definecolor{auburn}{rgb}{0.43,0.21, 0.1}
\definecolor{burgundy}{rgb}{0.5,0.0, 0.13}
\definecolor{medium red}{rgb}{.490,.298,.337}
\definecolor{dark red}{rgb}{.235,.141,.161}
\definecolor{col}{RGB}{10,218,200}
\definecolor{col2}{RGB}{200,50,255}

\hypersetup{
	colorlinks = true,
	linkcolor = {burgundy},
	citecolor = {burgundy}, 		
	linkbordercolor = {white},
}

\usepackage{hyperref}
\hypersetup{
    colorlinks=true,
    linkcolor=blue,
    filecolor=magenta,      
    urlcolor=blue,
    pdftitle={Strategy-Proof Probabilistic Social Choice Correspondence},
    pdfpagemode=FullScreen,
}

\let\OLDthebibliography\thebibliography
\renewcommand\thebibliography[1]{
	\OLDthebibliography{#1}
	\setlength{\parskip}{0pt}
	\setlength{\itemsep}{0pt plus 0.1ex}
}

\DeclareFontFamily{U}{mathx}{\hyphenchar\font45}
\DeclareFontShape{U}{mathx}{m}{n}{<-> mathx10}{}
\DeclareSymbolFont{mathx}{U}{mathx}{m}{n}
\DeclareMathAccent{\widebar}{0}{mathx}{"73}

\titleformat{\section}[block]{\normalfont\scshape\large\filcenter}{\thesection .}{1em}{}
\titleformat{\subsection}{\normalfont\scshape\large}{\thesubsection}{1em}{}
\titleformat{\subsubsection}{\normalfont\scshape\large}{\thesubsubsection}{1em}{}

\newtheorem{theorem}{Theorem}
\newtheorem{proposition}{Proposition}
\newtheorem{claim}{Claim}
\newtheorem{lemma}{Lemma}

\newtheorem{assumption}{Assumption}
\newtheorem*{theorem*}{Theorem}

\theoremstyle{definition}
\newtheorem{definition}{Definition}
\newtheorem{example}{Example}
\theoremstyle{remark}
\newtheorem{remark}{\textsc{Remark}}

\newcommand{\du}{\mathcal{D}_U}
\newcommand{\de}{\mathcal{D}_E}
\newcommand{\hv}{\hat{\varphi}}

\newdateformat{monthyeardate}{\monthname[\THEMONTH], \THEYEAR}

\usepackage{fancyhdr}
\title{\textsc{Strategy-Proof Probabilistic Social Choice Correspondences under Conditional Expected Utility}\thanks{\noindent Soumyarup Sadhukhan gratefully acknowledges financial support from the Indian Institute of Technology Kanpur under the Faculty Research Initiation Grant (Project No. IITK/MATH/2021295). }}

\author[1]{Madhuparna Karmokar\footnote{Contact: madhuparnakarmokar@yahoo.in}}
\author[2]{Ujjwal Kumar\footnote{Contact: ujjwal.kumar@uni-bonn.de}}
\author[3]{Soumyarup Sadhukhan\footnote{Contact: soumyarup.sadhukhan@gmail.com}}

\affil[1]{\small Indian Institute of Management Calcutta}
\affil[2]{\small Hausdorff Center for Mathematics and Institute for Microeconomics, University of Bonn}
\affil[3]{\small Indian Institute of Technology Kanpur}

\date{}

\begin{document}
\maketitle
	\begin{abstract}\singlespacing	
		\noindent 
       
        We study unanimous and strategy-proof probabilistic social choice correspondences (PSCCs), where the selected set of alternatives is interpreted as an interim outcome, and agents evaluate sets using conditional expected utility. We analyze two preference domains introduced by \cite{barbera2001strategy}: the conditionally expected utility consistent (CEUC) domain and the conditionally expected utility consistent with equal probabilities (CEUCEP) domain. Our results characterize all unanimous and strategy-proof PSCCs on these domains and identify cases when randomization enlarges the class of admissible rules. On the CEUC domain, every unanimous and strategy-proof PSCC is a random dictatorship, showing that randomization over sets yields no additional flexibility. In contrast, the CEUCEP domain admits a richer family of unanimous and strategy-proof PSCCs. For at most three agents, these rules are precisely the random bi-dictatorial rules, which are convex combinations of bi-dictatorial rules introduced in \cite{feldman1980strongly}. For four or more agents, the characterization depends on the number of alternatives. When there are exactly three alternatives, the class expands to the larger  family of coalition-weighted rules. Thus, randomization enlarges the class of strategy-proof correspondences in the three-alternative case, producing rules that are not convex combinations of deterministic strategy-proof correspondences. However, for four or more alternatives, the class of unanimous and strategy-proof probabilistic correspondences again collapse to random bi-dictatorships.
		\end{abstract}

\vspace{4mm}
		

\newpage
\section{Introduction}
The classical social choice function framework studies rules that select a single alternative from a set of feasible alternatives. Although this is appropriate in many settings, many real-life decisions involve sets of alternatives or randomized outcomes. For instance, a committee may select a shortlist of candidates, a department may identify a pool of acceptable hires, a jury may nominate finalists, or a public body may approve a set of proposals for further consideration. In such environments, the immediate collective outcome is not necessarily a final alternative, but a non-empty subset of alternatives. Moreover, the rule that selects such a subset may itself be probabilistic, either because randomization is used to preserve ex-ante fairness or because institutional procedures explicitly involve lotteries. These considerations motivate the study of probabilistic social choice correspondences (PSCCs), which assign lotteries over non-empty subsets of alternatives.

The interpretation of the selected set depends on the application. One plausible interpretation is that the selected subset consists of mutually compatible alternatives that are jointly treated as the final outcome, as in committee formation, team selection, hiring a cohort of faculty members, electing multiple representatives for a legislative body, or selecting a comprehensive package of public policies (see \cite{barbera1991voting},
\cite{ju2003characterization}, \cite{ozyurt2008strategy}). Alternatively, the selected subset may consist of incompatible alternatives that have passed an initial screening, from which a final alternative will be chosen at a later stage, often through a second-stage procedure or tie-breaking rule (see \cite{kelly1977strategy}, \cite{barbera2001strategy}, \cite{benoit2002strategic}, \cite{ching2002multi}, \cite{ingalagavi2023class}, \cite{rodriguez2025strategy}). In this paper, we adopt the latter interpretation, viewing the selected subset of alternatives as an interim outcome. This perspective is particularly relevant in institutional settings where a committee identifies a pool of qualified candidates, but the final selection is subject to a later tie-breaking rule or a randomized final stage.

When selected sets are interim outcomes, voters evaluate them by considering not only which alternatives they contain, but also how likely each contained alternative is to be finally chosen.  One way to model this is to allow each voter to have a subjective assessment of alternatives. An assessment captures a voter’s belief about the relative chances of different alternatives being selected at the final stage. These beliefs may arise from institutional features of the second-stage procedure, from expectations about tie-breaking, or from the voter’s subjective view about which alternatives are more likely to survive the final selection process. Given such an assessment, a voter evaluates a set by conditioning her beliefs on the event that the final choice must come from that set. She then compares sets according to their conditional expected utilities. We study two domains of such preferences over sets of alternatives introduced in \cite{barbera2001strategy}: the conditionally expected utility consistent (CEUC) domain and the conditionally expected utility consistent with equal probabilities (CEUCEP) domain. The CEUC domain allows these assessments to be subjective and unequal across alternatives, thereby capturing situations in which voters may hold heterogeneous beliefs about the second-stage selection process. The CEUCEP domain corresponds to the special case in which all alternatives in a selected set are treated as equally likely.

Given these domains, we consider two desirable properties of PSCCs, namely unanimity and strategy-proofness. Unanimity requires that if all voters have the same top-ranked set (subset of alternatives), then that set must be selected with probability one. This is a minimal efficiency requirement, since when there is complete consensus about the best set, the PSCC should respect it.  Strategy-proofness requires truthful reporting to be a dominant strategy. Since the rule is probabilistic, an agent compares the lottery obtained by truthful reporting with the lottery obtained after a misreport. We use first-order stochastic dominance with respect to the agent’s true preference over sets. Thus, a rule is strategy-proof if no agent can misreport her preference and obtain a lottery that assigns more probability to better sets, according to her true preference. This requirement is natural in our setting because voters’ preferences over sets take into account their expected utilities from the eventual final outcome. A strategy-proof rule, therefore, makes truthful reporting a dominant strategy even when the rule randomizes over selected subsets.

Unanimity and strategy-proofness are central requirements in classical social choice theory. In the standard social choice function framework, the Gibbard-Satterthwaite (\cite{gibbard1973manipulation,satterthwaite1975strategy}) theorem shows that, on the unrestricted domain, any strategy-proof rule satisfying a suitable range condition must be dictatorial, that is, the outcome is always determined by the top-ranked alternative of a fixed agent. \cite{gibbard1977manipulation} extends this insight to probabilistic social choice functions and shows that, when randomization over alternatives is allowed, unanimity and strategy-proofness lead to random dictatorship. Similar results arise for deterministic social choice correspondences (DSCCs) as well. In particular, \cite{barbera2001strategy} show that on the CEUC domain, every unanimous and strategy-proof deterministic rule is dictatorial, while on the CEUCEP domain, the DSCCs are bi-dictatorial. A bi-dictatorial rule always selects the top-ranked alternatives of a pair of agents (not necessarily distinct) as the outcome.

We study the probabilistic version of the deterministic model of
\cite{barbera2001strategy}. The central question is whether randomization over set-valued outcomes only produces lotteries over deterministic strategy-proof correspondences, or whether it gives rise to new strategy-proof rules that have no deterministic analogue. Our results show that both possibilities arise.

Our first result concerns the CEUC domain. We show that every unanimous and strategy-proof PSCC on the CEUC domain is random dictatorial (Theorem \ref{du_n}). Thus, although the rule is allowed to randomize over arbitrary non-empty subsets of alternatives, all probability is assigned to singleton
sets that are top-ranked by agents. In this domain, allowing lotteries over sets does not generate any structure beyond the usual random-dictatorial one.

The CEUCEP domain exhibits a richer structure. We begin by establishing a key structural result which shows that every unanimous and strategy-proof PSCC on this domain is tops-only and has support only on singleton and doubleton subsets of the union of the agents' top-ranked sets (Theorem \ref{de_tops-only}).\footnote{Informally, this refers to the collection of the top-ranked sets across all agents.} We then characterize the entire class of unanimous and strategy-proof PSCCs for up to three agents, showing that they are precisely the random bi-dictatorial rules (Theorem \ref{de_2}). These rules assign fixed probabilities to pairs of agents (not necessarily distinct) and, for each realized pair, select the union of the agents' top-ranked  sets. Equivalently, they are exactly the convex combinations of the bi-dictatorial rules introduced by
\cite{feldman1980strongly}.\footnote{\cite{feldman1980strongly} originally introduced these rules under the name \emph{duumvirate}. The terminology \emph{bi-dictatorial} was later adopted by \cite{barbera2001strategy}.}

For four or more agents, the characterization depends on the number of alternatives. When there are exactly three alternatives, unanimous and strategy-proof PSCCs are characterized by coalition-weighted rules (Theorem \ref{de_m=3_4}). A coalition-weighted rule is generated by a coalition weight function, which determines the probability assigned to an alternative as a function of the coalition of agents who rank it at the top. This class contains random bi-dictatorial rules, but is generally larger.
Particularly, Example \ref{eg_2} shows that, for three alternatives and four agents, there are unanimous and strategy-proof PSCCs that are not random bi-dictatorial. Hence, these rules cannot be written as convex combinations of unanimous and strategy-proof DSCCs.  By contrast, when there are at least four alternatives, this additional freedom disappears. We show that every unanimous and strategy-proof PSCC on the CEUCEP domain is random bi-dictatorial (Theorem \ref{de_n}). Hence, for at least four alternatives, every admissible probabilistic rule is a lottery over deterministic bi-dictatorial rules. For clarity, we summarize our characterization results in Table \ref{tab:contributions_ceucep}.

\begin{table}[htbp]
\centering
\caption{Characterization of Unanimous and Strategy-Proof PSCCs }
\label{tab:contributions_ceucep}
\small
\begin{tabular}{|lcclc|}
\hline
\textbf{Domain} & \textbf{Number of Agents ($n$)} & \textbf{Number of Alternatives ($m$)} & \textbf{Characterization Result} & \textbf{Theorem} \\ [1ex] \hline 
CEUC & $n \ge 2$ &$m \ge 3$ & Random Dictatorial Rules & \ref{du_n} \\ [1ex]
CEUCEP & $n \le 3$& $m \ge 3$ & Random Bi-dictatorial Rules & \ref{de_2} \\ [1ex]
CEUCEP & $n \ge 4$ & $m = 3$ & Coalition-weighted Rules (Strictly larger class) & \ref{de_m=3_4} \\ [1ex]
CEUCEP & $n \ge 4$ & $m \ge 4$ & Random Bi-dictatorial Rules & \ref{de_n} \\ \hline
\end{tabular}
\end{table}
Finally, we relate the three-alternative, and at least four agents case to the deterministic results of \cite{barbera2001strategy}. We focus particularly on this configuration because, in all other cases, the admissible probabilistic rules are convex combinations of strategy-proof deterministic rules. Although coalition-weighted rules may be strictly more general than random bi-dictatorial rules, their deterministic members are exactly bi-dictatorial rules (Proposition \ref{prop_deterministic_cw}). This shows that the deterministic version of our result coincides with that of \cite{barbera2001strategy}, but randomization can enlarge the class of unanimous and strategy-proof PSCCs beyond the convex hull of unanimous and strategy-proof DSCCs in the three-alternative case. We also provide an equivalent formulation of coalition-weighted rules for more than three alternatives, and show that these are exactly the random bi-dictatorial rules (Theorem \ref{thm_mgeq4}). Thus, yielding an alternative characterization of unanimous and strategy-proof PSCCs for at least four alternatives.

\subsection{Related Literature}

The seminal results of \cite{gibbard1973manipulation} and \cite{satterthwaite1975strategy} show that any strategy-proof and onto social choice function must be dictatorial. A central criticism of this framework is its restriction to single-valued outcomes. This has led to the study of deterministic social choice correspondences (DSCCs), where outcomes are sets of alternatives.
   An important issue in this literature is how agents compare subsets of alternatives.  \cite{gardenfors1976manipulation} and \cite{kelly1977strategy} provide foundational approaches to extending preferences over alternatives to preferences over sets.  \cite{gardenfors1976manipulation} introduced a ``sure-thing'' principle for comparing sets of alternatives. Roughly, this principle says that adding an alternative that is worse than every element of a set should make the set worse, while adding an alternative that is better than every element of a set should make the set better. In contrast, \cite{kelly1977strategy} studies a more conservative extension, now known as the Kelly extension,  under which one set is strictly preferred to another only if every element of the former is strictly preferred to every element of the latter. These approaches led to a series of impossibility and characterization results for strategy-proof correspondences. Adding to this strand of impossibility results, \cite{ching2002multi} introduce a new notion of strategy-proofness for DSCCs and show that, under their criterion, any such correspondence must be either dictatorial or constant.

A major advancement in the restricted-domain analysis of strategy-proof social choice correspondences is due to \cite{barbera2001strategy}. Rather than deriving preferences over sets from ordinal preferences over alternatives through a general extension rule, they model preferences over sets using conditional expected utility. They introduce the conditionally expected utility consistent (CEUC) domain and its equal-probability counterpart, the conditionally expected utility consistent with equal probabilities (CEUCEP) domain. Their characterization shows that unanimity and strategy-proofness force dictatorship on the CEUC domain, while on the CEUCEP domain the admissible deterministic rules are dictatorial or bi-dictatorial. An earlier related result by \cite{feldman1980strongly} establishes a bi-dictatorship theorem for strategy-proof decision schemes based on even-chance lotteries. These results form the deterministic baseline for our analysis of probabilistic social choice correspondences.

Another strand studies preferences over sets and lotteries, particularly in environments where outcomes may involve randomization or ties. \cite{benoit2002strategic} considers a framework in which agents’ preferences are defined either over sets of alternatives or over lotteries on alternatives, and shows that, under mild conditions, no rule can simultaneously satisfy strategy-proofness and a stronger unanimity requirement called near- unanimity.  Relatedly, \cite{duggan2000strategic} generalize the Gibbard–Satterthwaite theorem to non-resolute environments, allowing for ties and heterogeneous beliefs about their resolution, and show that social choice correspondences remain vulnerable to manipulation under very weak conditions. Closely related, \cite{brandt2022strategyproof} study social choice correspondences when both preferences and outcomes may contain ties and a final alternative is selected via tie-breaking rules such as lotteries or deterministic refinements, and show that all anonymous and Pareto-optimal correspondences are manipulable under such procedures. In a related vein, \cite{ingalagavi2023class} study strategy-proof deterministic social choice correspondences on the CEUC and CEUCEP domains under additional structure, such as single-peaked utility functions. Closely related to our setting, \cite{rodriguez2025strategy} studies deterministic social choice correspondences for conditional expected utility maximizers, where agents start from a common assessment over alternatives and use it to evaluate each selected set. From a different perspective, \cite{brandt2023characterizing} characterize the top cycle correspondence using strategy-proofness-type axioms.

Turning to probabilistic social choice, most of the literature focuses on lotteries over single alternatives. As noted above, \cite{gibbard1977manipulation} provides a foundational characterization in this setting. However, comparatively little is known about PSCCs, where randomization occurs over sets of alternatives rather than individual alternatives. An early contribution is \cite{nandeibam2002structure}, who studies probabilistic social choice correspondences as rules assigning lotteries over subsets of feasible alternatives and analyzes the distribution of coalitional influence. More closely related in spirit is the literature on randomized committee formation, where randomization also occurs over set-valued outcomes. For instance, \cite{roy2019formation} and \cite{roy2023committee} study randomized voting rules over feasible committees on separable domains and obtain random dictatorship characterizations under strategy-proofness and onto conditions. These papers differ from ours in their interpretation of sets: committees are final outcomes composed of mutually compatible alternatives, whereas we interpret selected sets as interim outcomes that will later be resolved into a single alternative.

Our paper, therefore, connects two strands of the literature: deterministic strategy-proof correspondences under conditional expected utility and probabilistic social choice. We show that the relationship between deterministic and probabilistic rules is domain-dependent. On $\mathcal D_U$, the probabilistic characterization collapses to random dictatorship. On $\mathcal D_E$, the generic structure is random bi-dictatorship, the probabilistic counterpart of deterministic bi-dictatorship. However, when there are three alternatives and at least four agents, randomization generates strategy-proof PSCCs that are not mixtures of deterministic strategy-proof correspondences. This identifies a new role for randomization in set-valued social choice.

\subsection{Organization of the paper}

 The rest of the paper is organized as follows. In Section \ref{sec:model}, we introduce the formal framework and notation. In Section \ref{sec:du} and  Section \ref{sec:de}, we study the CEUC and CEUCEP domains, respectively. Finally, in Section \ref{sec:discussions},  we relate our findings to the deterministic results established by \cite{barbera2001strategy} and provide an extension of coalition-weighted rules when there are at least four alternatives. All proofs omitted from the main text are relegated to the Appendix and supplementary material.

\section{Preliminaries}\label{sec:model}

Consider a society of $n$ agents where $n \geq 2$. We denote the set of agents by $N=\{1,2,\ldots,n\}$. Let $A$ be a finite set of alternatives, and unless otherwise specified, we assume that $|A|=m \geq 3$.  Let $\mathcal{A}$ denote the set of all non-empty subsets of $A$. To maintain consistency throughout the paper, we distinguish between alternatives and subsets of alternatives. Specifically, we let lowercase letters denote individual alternatives, writing $a, b, c, x, y, z, w \in A$, while uppercase letters denote non-empty subsets of alternatives, writing $X, Y, Z, W \in \mathcal{A}$. A (weak) preference over $\mathcal{A}$ is a weak order over $\mathcal{A}$ (complete and transitive binary relation), and the set of all preferences over $\mathcal{A}$ is denoted by $\mathbb{W}(\mathcal{A})$.\footnote{In a similar manner, we can define the preferences over $A$. A (weak) preference over $A$ is a weak order over $A$.} For a preference $R$, we denote by $P$, the strict part of $R$, i.e., $XPY$ implies $XRY$ and $\neg YRX$. Moreover, for a preference $R$ and a set $X$, we write $I(X,R)$ to denote all the sets that are indifferent to $X$ according to $R$, i.e., $I(X,R)=\{Y\in \mathcal{A}\mid XRY \text{ and } YRX\}$. A preference $R$ is a strict preference if $|I(X,R)|=1$ for all $X\in \mathcal{A}$. We often write a preference $R$ as $R\equiv XYZ\cdots$, where $X,Y,Z\in \mathcal{A}$, meaning according to $R$, $X$ is the top-ranked set, $Y$ is the second-ranked set, and $Z$ is the third-ranked set. 

In this paper, we talk about two kinds of preferences over $\mathcal{A}$ which come from a utility function and an assessment. A utility function $v$ is a mapping  $v: A \rightarrow \mathbb{R}$. We make the following assumption on the utility functions that we consider in this paper. 
 \begin{assumption}\label{ass_1}
      For all distinct alternatives $x, y \in A$, $v(x) \neq v(y)$. 
 \end{assumption}
 An assessment $\lambda$ is a function $\lambda: A \rightarrow (0,1]$ such that $\sum_{a \in A} \lambda(a) = 1$ . We now define the first kind of preferences that we work with in this paper.
\begin{definition}
 An ordering $R$ over $\mathcal{A}$ is conditionally expected utility consistent (CEUC) if there exists a utility function $v$ and an assessment $\lambda$
such that $\forall  X,Y \in \mathcal{A}$,
  $$XRY\Leftrightarrow\sum_{x \in X}v(x)\frac{\lambda(x)}{\sum_{y \in X} \lambda(y)} \geqslant \sum_{x \in Y}v(x) \frac{\lambda(x)}{\sum_{y \in Y} \lambda(y)}$$

  Let $\du$ be the set of all CEUC orderings over $\mathcal{A}$.
\end{definition}
Note that the above definition is well-defined as $\lambda$ takes a value strictly greater than $0$. Also, given $v$ and $\lambda$, we sometimes extend $v$ over the elements of $\mathcal{A}$ (also called the expected utility) as $$v^{\lambda}(X)=\sum_{x \in X}v(x)\frac{\lambda(x)}{\sum_{y \in X} \lambda(y)}.$$ Next, we define the second kind of preferences that is a special case of CEUC orderings.
\begin{definition}
 An ordering $R$ over $\mathcal{A}$ is conditionally expected utility consistent with equal probabilities (CEUCEP) if there exists a utility function $v$ such that $\forall  X,Y \in \mathcal{A}$,
  $$XRY\Leftrightarrow \dfrac{\sum_{x \in X}v(x)}{|X|}  \geqslant  \dfrac{\sum_{x \in Y}v(x)}{|Y|}.$$ 

Let $\de$ be the set of all CEUCEP orderings.
\end{definition}
It is straightforward to see that $R$ is a CEUCEP ordering corresponding to a valuation function $v$ if it is the CEUC ordering corresponding to $v$ and the uniform assessment, i.e., $\eta(x)=\tfrac{1}{m}$ for all $x\in A$. Below we provide an illustration of a CEUC ordering and a CEUCEP ordering. 
\begin{example}
    Let $A=\{a,b,c\}$. Consider an assessment $\lambda$ and a valuation $v$ where $\lambda(a) = 0.5$, $\lambda(b) =0.3$, $\lambda(c)=0.2$, and  
    $v(a)=1,\; v(b)=2,\; v(c)=3$. Then,
    $v^\lambda(\{a,b\}) = \frac{11}{8}$, 
     $v^\lambda(\{b, c\}) = \frac{12}{5}$, 
      $v^\lambda(\{c, a\}) = \frac{11}{7}$, and
       $v^\lambda(\{a, b, c\}) = \frac{17}{10}$. Thus the corresponding CEUC ordering is
$$R=\{c\}\{b,c\}\{b\}\{a,b,c\}\{a,c\}\{a,b\}\{a\}.$$
%     \begin{center}
%      \begin{tabular}{c c}
%     & $\underline{R}$ \\
%  & $\{c\}$ \\
%  & $\{c, b\}$ \\
%  & $\{b\}$ \\
%  & $\{a, b, c\}$ \\
%  & $\{c, a\}$ \\
%  & $\{b, a\}$ \\
%  & $\{a\}$ 
% \end{tabular}
%   \end{center}

To see the CEUCEP ordering with the same utility function $v$, we first calculate the expected utilities of the sets with the uniform assessment $\upsilon$. These are
    $v^{\eta}(\{a,b\}) = \frac{3}{2}$, 
     $v^{\eta}(\{b, c\}) = \frac{5}{2}$, 
      $v^{\eta}(\{c, a\}) = 2$, and 
       $v^{\eta}(\{a, b, c\})=2$.
       Therefore, the ordering is 
       $$Q=\{c\}\{b,c\}\big[\{b\}\{a,b,c\}\{a,c\}\big]\{a,b\}\{a\}.$$
 %       \begin{center}
 % \begin{tabular}{c c}
 %  & $\underline{Q}$ \\
 % & $\{c\}$ \\
 % & $\{c, b\}$ \\
 % & $\{b\}$, $\{a, b, c\}$, $\{c, a\}$  \\
 % & $\{b, a\}$ \\
 % & $\{a\}$ 
 % \end{tabular}
 % \end{center}
    where $\big[\{b\}\{a,b,c\}\{a,c\}\big]$ indicates that the three sets $\{b\}$, $\{a, b, c\}$, and $\{c, a\}$ are indifferent in the ordering. Note that for an ordering (either CEUC or CEUCEP), the best and worst sets are always unique singletons, meaning any indifference must occur between these two extremes (see Remark \ref{rem_2}). \hfill $\square$
\end{example}
\begin{remark}
Clearly $\mathcal{D}_E\subset\mathcal{D}_U\subset \mathbb{W}(\mathcal{A})$.
\end{remark}

\begin{remark}\label{rem_2}
    Although the assumption that no two distinct elements have the same utility rules out indifference between
sets of cardinality 1, it is obvious that there can be $X, Y \in \mathcal{A}$ such that $XR_iY$ and $YR_iX$. Yet, the best and the worst elements of any preferences in $\mathcal{D}_U$
and $\mathcal{D}_E$ will be unique and singleton sets. 
\end{remark}
Henceforth, throughout, $\mathcal{D}$ denotes either $\de$ or $\du$, and we restrict attention to preferences in $\mathcal{D}$. In view of Remark \ref{rem_2}, we denote the best element (top-ranked set) of a preference $R\in \mathcal{D}$ by $\tau(R)$. Moreover, for a set of agents $S\subseteq N$ and a profile $R_N\in \mathcal{D}^n$, we use $\tau(R_S)$ to denote the collection of top-ranked sets of the agents in $S$ at the profile $R_N$. Formally, $\tau(R_S)=\{\tau(R_i)\mid i\in S\}$. 

A \textbf{probabilistic social choice correspondence} (PSCC) $\varphi$ on $\mathcal{D}^n$ is a mapping
$\varphi: \mathcal{D}^n \rightarrow \Delta \mathcal{A}$ 
where $\Delta \mathcal{A}$ denotes the set of probability distributions over $\mathcal{A}$.\footnote{More formally, $\Delta \mathcal{A}=\{\lambda : \mathcal{A}\to [0,1] \text{ such that } \sum_{X\in \mathcal{A}}\lambda(X)=1\}$.} For a PSCC $\varphi$, a set $X \in \mathcal{A}$, and a profile $R_N$, we denote by $\varphi_X(R_N)$ the probability assigned to $X$ at $R_N$. For $\mathcal{B}\subseteq \mathcal{A}$  and $R_N\in \mathcal{D}^n$, we write $\varphi_{\mathcal{B}}(R_N)$ to denote the total probability assigned to the sets in $\mathcal{B}$ at $R_N$ by $\varphi$, i.e., $\varphi_{\mathcal{B}}(R_N)=\sum_{X\in \mathcal{B}}\varphi_X(R_N)$.
A PSCC is a \textbf{deterministic social choice correspondence} (DSCC) if it selects a degenerate probability distribution at every preference profile. More formally, a PSCC $\varphi:\mathcal{D}^n \to \Delta \mathcal{A}$ is a DSCC if $\varphi_X(R_N)\in \{0,1\}$ for all $X \in \mathcal{A}$ and all $R_N\in \mathcal{D}^n$. Sometimes, for simplicity, we write a DSCC, $\varphi$, as a function from $\mathcal{D}^n$ to $\mathcal{A}$ where for all $R_N\in \mathcal{D}^n$, $\varphi(R_N)=X$ if $\varphi_X(R_N)=1$ for some $X\in \mathcal{A}$.

 A PSCC $\varphi$ is \textbf{unanimous} if for any profile $R_N=(R_1,\ldots,R_n)\in \mathcal{D}^n$ with $\tau(R_1)=\cdots=\tau(R_n)=\{a\}$ for some $a\in A$, we have $\varphi_{\{a\}}(R_N)=1$. Two profiles $R_N,R_N'\in \mathcal{D}^n$ are \textit{tops-equivalent} if $\tau(R_i)=\tau(R_i')$ for all $i\in N$. A PSCC is called \textbf{tops only} if its outcomes do not change over tops-equivalent profiles. In other words, the outcome of such a PSCC depends only on the top-ranked sets in a preference profile. Mathematically, $\varphi(R_N)=\varphi(R_N')$ for all tops-equivalent profiles $R_N$ and $R'_N$. Next, we define the notion of \textit{strategy-proofness} of a PSCC. The idea is standard in the literature and was first introduced in \cite{gibbard1977manipulation}. For this, we define the upper contour set of $X\in \mathcal{A}$ at a preference $R$ as the set of elements in $\mathcal{A}$ that are preferred over $X$ according to $R$. More formally, $U(X,R)=\{Y\in \mathcal{A}\mid YRX\}$

\begin{definition}
	A PSCC $\varphi: \mathcal{D}^n \to \Delta \mathcal{A}$ is \textbf{strategy-proof} if for all $i \in N$, all $R_i, R'_i \in  \mathcal{D}$, and all $R_{-i} \in \mathcal{D}^{n-1}$, $\varphi(R_i,R_{-i})$ first order stochastically dominates $\varphi(R'_i,R_{-i})$ according to $R_i$, that is,
    \begin{equation}\label{eq_spf}
        \varphi_{U(X,R_i)}(R_i,R_{-i}) \geq \varphi_{U(X,R_i)}(R'_i,R_{-i}) \mbox{ for all } X\in \mathcal{A}.
    \end{equation}
    \end{definition}

The above definition treats the elements of $\mathcal A$ as the outcomes of the rule and applies the standard first-order stochastic dominance (FOSD) extension of an agent's preference over these outcomes. We now clarify why this is the
appropriate formulation in our setting.

There are two sources of uncertainty in our model. The first comes as the rule itself is probabilistic and assigns a lottery over elements of $\mathcal A$. The second one is due to the fact that a selected set may later be resolved into a single alternative, via a tie-breaking rule, or more generally, the second-stage procedure, which need not be known to the agents. The second source of uncertainty is already incorporated in the agents' preferences over sets, as when an agent compares two sets, she takes into account her belief about which alternatives are more likely to be selected from them. The assessment in the definition of the preferences represents the agent's belief about this second-stage resolution. Thus, we define strategy-proofness here, considering only the first source of uncertainty, using FOSD. 

In an alternate formulation, one may define strategy-proofness in terms of the compound lottery based on the outcome of the PSCC and the assessment of the agents. We argue that to do so in a meaningful way, it gets back to our formulation using FOSD. To see this, if an agent has a utility function $v_i$ and an assessment $\lambda_i$, then any lottery over sets induces a compound lottery over alternatives. For $p\in\Delta \mathcal A$, the probability induced on an alternative $x$ is  $$\sum_{X \in\{Y\in\mathcal A\mid \,x\in Y\}} p_X \frac{\lambda_i(x)}{\sum_{y\in X}\lambda_i(y)}.$$  The expected utility of this induced lottery is $$\begin{aligned} \sum_{x\in A} v_i(x) \sum_{X \in\{Y\in\mathcal A\mid \,x\in Y\}} p_X \frac{\lambda_i(x)}{\sum_{y\in X}\lambda_i(y)} &= \sum_{X\in\mathcal A} p_X \left( \sum_{x\in X} v_i(x) \frac{\lambda_i(x)}{\sum_{y\in X}\lambda_i(y)} \right) \\ &= \sum_{X\in\mathcal A}p_Xv_i^{\lambda_i}(X). \end{aligned}$$  Thus, for a fixed pair $(v_i,\lambda_i)$, evaluating the compound lottery over alternatives is the same as evaluating the lottery over sets using the induced utility $v_i^{\lambda_i}$. The mechanism, however, is defined on ordinal preferences over $\mathcal A$, not on the cardinal pairs $(v_i,\lambda_i)$ that generate them. This distinction is important. In $\mathcal D_U$, assessments are subjective and may differ across agents, so the same lottery over sets need not induce a common lottery over alternatives. Moreover, even for a fixed agent, the same ordering over $\mathcal A$ may be generated by more than one utility-assessment pair. An incentive condition based on one such representation would therefore depend on cardinal information that is not part of the report. For this reason, a feasible way is to define strategy-proofness directly on lotteries over $\mathcal A$, using FOSD with respect to the reported ordering over sets. This keeps the incentive condition ordinal and representation-independent. If $p\in\Delta\mathcal A$ first-order stochastically dominates $q\in\Delta \mathcal A$ according to $R_i$, then $$ \sum_{X\in\mathcal A}p_Xu_i(X) \geq \sum_{X\in\mathcal A}q_Xu_i(X)$$ for every utility representation $u_i:\mathcal A\to\mathbb R$ of $R_i$. In particular, the inequality holds for every conditional expected-utility representation $(v_i,\lambda_i)$ that generates $R_i$. Hence, FOSD strategy-proofness rules out profitable deviations under the compound-lottery interpretation for every utility function and assessment consistent with the reported preference over sets. In $\mathcal D_E$, the assessment is fixed to be uniform, but the rule is still defined on orderings over $\mathcal A$ rather than on cardinal utilities; hence we use the same ordinal formulation throughout.

 Next, we define an equivalent formulation of strategy-proofness that will be used repeatedly in the proofs. Similar to the upper contour set, we may also define the strict upper contour set of $X\in \mathcal{A}$ at a preference $R$, denoted by $U^{-}(X,R)$. The strict upper contour set of $X$ at $R$ denote the set of elements in $\mathcal{A}$ that are \textit{strictly} preferred to $X$ according to $R$, i.e.,  $U^{-}(X,R)=\{Y\in \mathcal{A}\mid YPX\}$. It is straightforward to see that (\ref{eq_spf}) is equivalent to $$\varphi_{U^{-}(X,R_i)}(R_i,R_{-i}) \geq \varphi_{U^{-}(X,R_i)}(R'_i,R_{-i}) \mbox{ for all } X\in \mathcal{A}.$$

\section{Characterization of Unanimous and Strategy-proof PSCCs on $\du$}\label{sec:du}
In this section, we explore the structure of unanimous and strategy-proof PSCCs on $\mathcal{D}_U$. We start with the definition of random dictatorial rules. In Theorem \ref{du_n}, we show that every PSCC on the $\du$ domain is unanimous and strategy-proof if and only if it is a random dictatorial rule.

\begin{definition}
    A PSCC $\varphi: \mathcal{D}^n \to \Delta \mathcal{A}$ is \textbf{random dictatorial} if there exist coefficients $\epsilon_1,\ldots,\epsilon_n$ with $\epsilon_i\geq 0$ and $\sum_{i=1}^n\epsilon_i=1$ such that for all $R_N\in \mathcal{D}^n$,
    $$\varphi_X(R_N)=\sum_{i\in \{j\in N\mid \tau(R_j)=X\}}\epsilon_i.$$
\end{definition}

In the proof of Theorem \ref{du_n}, we utilize the seminal result by \cite{gibbard1977manipulation} regarding unanimous and strategy-proof probabilistic social choice functions (PSCFs) on the unrestricted domain. Let $\mathcal{P}$ denote the set of all strict preferences over $A$. For any $R \in \du$, we define $R|_{A}\in \mathcal{P}$ (with a slight abuse of notation) as follows: for distinct $a, b \in A$, $a R|_{A} b$ if and only if $\{a\} R \{b\}$.\footnote{Note that $R|_{A}$ is a strict preference as by  Assumption 1, $v(x)\neq v(y)$ for all $x,y \in A$.}
Extending this to the entire domain, we define $\du|_{A} = \{R|_{A} \mid R \in \du\}$. It is straightforward to observe that $\du|_A = \mathcal{P}$. A PSCF $\phi$ is a mapping $\phi: \mathcal{P}^n \to \Delta A$. The properties of unanimity and strategy-proofness for PSCFs are defined analogously to those for PSCCs, as is the definition of a random dictatorship. \cite{gibbard1977manipulation} established the following characterization.
\begin{theorem*}[\cite{gibbard1977manipulation}]
    Assume $m\geq 3$. Every PSCF $\phi:\mathcal{P}^n \rightarrow \Delta A$ is unanimous and strategy-proof if and only if $\phi$ is random dictatorial.
\end{theorem*}

 We now extend the random-dictatorship conclusion of  \cite{gibbard1977manipulation} from PSCFs to PSCCs on the domain $\du$. The key difference is that a PSCC may assign probability to arbitrary nonempty subsets of alternatives. We show that under unanimity and strategy-proofness on $\du$ the rule assigns probability only to singleton top-ranked sets and therefore reduces to an ordinary PSCF.
 
\begin{theorem}\label{du_n}
    Let $n\geq 2$ and $\varphi:\du^n\to \Delta \mathcal{A}$ be a PSCC. Then $\varphi$ is unanimous and strategy-proof if and only if $\varphi$ is random dictatorial.
\end{theorem}
\begin{proof} Random dictatorial rules are unanimous and strategy-proof on any domain. So, we only prove the necessary part here. We prove the theorem by induction  on the number of agents. For brevity, we present the base case $n=2$ here and relegate the induction step to Appendix \ref{ap_dictatorial}. 

\noindent \textbf{Base case: $n=2$.} Let $\varphi:\du^2\to \Delta \mathcal{A}$ be a PSCC satisfying unanimity and strategy-proofness. We show that $\varphi$ is a random dictatorial. We start with a lemma that shows that $\varphi$ assigns probabilities only to the top-ranked sets of the agents at any preference profile.

  \begin{lemma}\label{lem_1}
    For all $(R_1,R_2)\in \du^2$, $$\varphi_X(R_1,R_2)=0 \text{ for all } X\notin \{\tau(R_1),\tau(R_2)\}.$$ 
\end{lemma}

\begin{proof}
    Consider the preference profile $(R_1,R_2)\in \mathcal{D}_U^2$. Suppose $\tau(R_1)=\{a\}$ and $\tau(R_2)=\{b\}$ for some $a,b\in A$. When $a=b$, $\varphi_{\{a\}}(R_1,R_2)=1$ follows from unanimity of $\varphi$. Hence, $\varphi_X(R_1,R_2)=0 \text{ for all } X\notin \{\{a\}\}$.  We proceed to show that $\varphi_X(R_1,R_2)=0 \text{ for all } X\notin \{\{a\},\{b\}\}$ when $a\neq b$. We begin by establishing a sequence of claims, which will then be used to prove the lemma. For this purpose, we introduce the following preferences in $\mathcal{D}_U$ for all distinct $x,y,z \in A$. 
    \begin{itemize}
   \item  $R^x\equiv \{x\}\cdots$,
   \item $\bar{R}^{xy}\equiv\{x\}\{x,y\}\{y\}\cdots$,
        \item $\tilde{R}^{xy}\equiv\{x\}\{x,y\}\cdots$,
        \item $\hat{R}^{xyz}\equiv\{x\}\{x,y\}\cdots\{y,z\}\{y\}$, and 
        \item $\breve{R}^{xy}\equiv\{x\}\cdots \{x,y\}\{y\}$.
    \end{itemize}
    The existence of these preferences is established in Appendix \ref{appen_2}.
    \begin{claim}\label{claim dictator 1} Let $x,y \in A$ be distinct. Then
       $\varphi_X(\bar{R}^{xy},\bar{R}^{yx})=0$ for all $X\in \mathcal{A}\setminus \{\{x\},\{x,y\},\{y\}\}$ and  $\varphi(\bar{R}^{xy},\bar{R}^{yx})=\varphi(\bar{R}^{xy},{R}^y)=\varphi({R}^x,\bar{R}^{yx})$. 
    \end{claim}
    \textbf{Proof of the claim:} By unanimity of $\varphi$, $\varphi_{\{x\}}(\bar{R}^{xy},\bar{R}^{xy})=1$. This, together with strategy-proofness of $\varphi$, implies $\varphi_{\{\{x\},\{x,y\},\{y\}\}}(\bar{R}^{xy},\bar{R}^{yx})=\varphi_{\{\{x\},\{x,y\},\{y\}\}}(\bar{R}^{xy},\bar{R}^{xy})=1$. Hence, $\varphi_X(\bar{R}^{xy},\bar{R}^{yx})=0$ for all $X\in \mathcal{A}\setminus \{\{x\},\{x,y\},\{y\}\}$. 
   
   Since $ \{\{x\},\{x,y\},\{y\}\}$ is an upper contour set in $\bar{R}^{xy}$, $\varphi_{\{\{x\},\{x,y\},\{y\}\}}(\bar{R}^{xy},\bar{R}^{yx})=\varphi_{\{\{x\},\{x,y\},\{y\}\}}(\bar{R}^{xy},{R}^y)$, as otherwise agent $1$ can manipulate at $(\bar{R}^{xy},{R}^{yx})$ via the preference $R^y$. This means that $\varphi_X(\bar{R}^{xy},{R}^y)=0$ for $X\in \mathcal{A}\setminus \{\{x\},\{x,y\},\{y\}\}$. Moreover, since the relative ordering of the sets $\{y\},\{x,y\}$ and $\{x\}$ is the same at $\bar{R}^{yx}$ and $R^y$, by strategy-proofness, $\varphi_{X}(\bar{R}^{xy},\bar{R}^{yx})=\varphi_{X}(\bar{R}^{xy},{R}^y)$ for $X\in\{\{x\},\{x,y\},\{y\}\}$. Hence, $\varphi(\bar{R}^{xy},\bar{R}^y)=\varphi(\bar{R}^{xy},{R}^y)$. Using similar logic, we can show that  $\varphi(\bar{R}^{xy},\bar{R}^{yx})=\varphi({R}^x,\bar{R}^{yx})$. \hfill $\square$

 \begin{claim}\label{claim dictator 2}
Let $x,y \in A$ be distinct. Then $\varphi(\bar{R}^{xy},\bar{R}^{yx})=\varphi(\tilde{R}^{xy},{R}^y)=\varphi({R}^x,\tilde{R}^{yx})$.
    \end{claim}
\noindent\textbf{Proof of the claim:} By strategy-proofness of $\varphi$, $\varphi_{\{x\}}(\tilde{R}^{xy},R^y)=\varphi_{\{x\}}(\bar{R}^{xy},R^y)$ and $\varphi_{\{x,y\}}(\tilde{R}^{xy},R^y)=\varphi_{\{x,y\}}(\bar{R}^{xy},R^y)$. Next we argue that $\varphi_{\{y\}}(\tilde{R}^{xy},R^y)=\varphi_{\{y\}}(\bar{R}^{xy},R^y)$. By Claim \ref{claim dictator 1}, $\varphi(\bar{R}^{xy},\bar{R}^{yx})=\varphi(\bar{R}^{xy},{R}^y)$. Hence, $\varphi_X(\bar{R}^{xy},{R}^y)=0$ for all $X\in \mathcal{A}\setminus\{\{x\},\{x,y\},\{y\}\}$. Thus, $\varphi_{\{y\}}(\tilde{R}^{xy},R^y)\leq \varphi_{\{y\}}(\bar{R}^{xy},R^y)$. Note that again by Claim \ref{claim dictator 1}, $\varphi_{\{y\}}(\bar{R}^{xy},\bar{R}^{yx})=\varphi_{\{y\}}(\tilde{R}^{xy},\bar{R}^{yx})=\varphi_{\{y\}}(\bar{R}^{xy},R^y)$. If  $\varphi_{\{y\}}(\tilde{R}^{xy},R^y)<\varphi_{\{y\}}(\bar{R}^{xy},R^y)$, then agent $2$ will manipulate at the profile $(\tilde{R}^{xy},R^y)$ via the preference $\bar{R}^{yx}$. This means that $\varphi(\tilde{R}^{xy},R^y)=\varphi(\bar{R}^{xy},R^{yx})$. Combining this with Claim \ref{claim dictator 1}, we have $\varphi(\tilde{R}^{xy},R^y)=\varphi(\bar{R}^{xy},\bar{R}^{yx})$. Using similar arguments it follows that  $\varphi({R}^x,\tilde{R}^{yx})=\varphi(\bar{R}^{xy},\bar{R}^{yx})$. \hfill $\square$
   
   \begin{claim}\label{claim dictator 3}
      Let $x,y \in A$ be distinct. Then $\varphi_{\{x,y\}}(\bar{R}^{xy},\bar{R}^{yx})=0$.
   \end{claim}

  \noindent\textbf{Proof of the claim:} By Claims \ref{claim dictator 1} and \ref{claim dictator 2}, $\varphi_{\{x,y\}}(\tilde{R}^{xy},\hat{R}^{zxy})=0$. By strategy-proofness, it follows that $\varphi_{\{x,y\}}(\tilde{R}^{xy},\breve{R}^{yx})=0$. Since by Claim \ref{claim dictator 2}, $\varphi(\bar{R}^{xy},\bar{R}^{yx})=\varphi(\tilde{R}^{xy},\breve{R}^{yx})$, we have $\varphi_{\{x,y\}}(\bar{R}^{xy},\bar{R}^{yx})=\varphi_{\{x,y\}}(\tilde{R}^{xy},\breve{R}^{yx})=0$.
   \hfill $\square$

 By Claims \ref{claim dictator 1} and \ref{claim dictator 3}, it follows that $\varphi_X(\bar{R}^{ab},\bar{R}^{ba})=0$ for all $X\in \mathcal{A}\setminus\{\{a\},\{b\}\}$. It follows that  $\varphi_{\{\{a\},\{b\}\}}(\bar{R}^{ab},\bar{R}^{ba})=1$. By strategy-proofness and Claim \ref{claim dictator 1}, we have $\varphi_{\{a\}}(R_1,R_2)=\varphi_{\{a\}}(\bar{R}^{ab},R_2)=\varphi_{\{a\}}(\bar{R}^{ab},\bar{R}^{ba})$ and $\varphi_{\{b\}}(R_1,R_2)=\varphi_{\{b\}}({R}_1,\bar{R}^{ba})=\varphi_{\{b\}}(\bar{R}^{ba},\bar{R}^{ab})$. Combining this with the fact that  $\varphi_{\{\{a\},\{b\}\}}(\bar{R}^{ab},\bar{R}^{ba})=1$,  we obtain that $ \varphi_X(R_1,R_2)=0$ for all $X\in \mathcal{A}\setminus \{\{a\},\{b\}\}$. This completes the proof of the lemma.
\end{proof}

By Lemma \ref{lem_1} and strategy-proofness, $\varphi$ is tops-only and assigns positive probability only to singleton top-ranked sets. Hence we can associate with $\varphi$ a PSCF $\hv:\mathcal P^2\to\Delta A $ defined by $\hv_a(Q_1,Q_2)=\varphi_{\{a\}}(R_1,R_2),$  where $Q_i=R_i|_A$ for each $i\in\{1,2\}$. This is well-defined because $\du|_A=\mathcal{P}$ and $\varphi$ is tops-only. We claim that $\hv$ is unanimous and strategy-proof. Unanimity follows immediately from the unanimity of $\varphi$. To see strategy-proofness, consider a preference $(Q_1,Q_2)\in\mathcal{P}^2$, an agent $j \in \{1,2\}$, a preference $Q_j' \in \mathcal{P}$, and an outcome $x\in A$. Assume without loss of generality that $j=1$. We show that $\hv_{U(x,Q_1)}(Q_1,Q_2)\geq \hv_{U(x,Q_1)}(Q_1',Q_2)$. Let $(R_1,R_2) \in \du^2$ be a preference profile such that $Q_i=R_i|_A$ for all $i\in \{1,2\}$, and let $R_1' \in \du$ be such that $Q_1'=R_1'|_A$. By the construction of $\hv$, \begin{align*}
     \hv_{U(x,Q_1)}(Q_1,Q_2)&=\varphi_{U(x,Q_1)}(R_1,R_2)=\varphi_{U(x,R_1)}(R_1,R_2) \mbox{ and }\\ \hv_{U(x,Q_1)}(Q_1',Q_{2})&=\varphi_{U(x,Q_1)}(R_1',R_{2})=\varphi_{U(x,R_1)}(R_1',R_{2}).
 \end{align*} Since $\varphi$ is strategy-proof, $\varphi_{U(x,R_1)}(R_1,R_2)\geq \varphi_{U(x,R_1)}(R_1',R_{2})$. Hence, $\hv_{U(x,Q_1)}(Q_1,Q_2)\geq \hv_{U(x,Q_1)}(Q_1',Q_{2}).$  Therefore, $\hv$ is strategy-proof. By \cite{gibbard1977manipulation}, $\hv$ is random dictatorial. Hence there exist $(\epsilon_1,\epsilon_2)\in [0,1]^2$ with $\epsilon_1+\epsilon_2=1$ such that for all $(Q_1,Q_2)\in \mathcal{P}^2$ and all $a\in A$, $$\hv_a(Q_1,Q_2)=\sum_{\{i\mid \tau(Q_i)=a\}}\epsilon_i.$$
Therefore, by the definition of $\hv$, for all $(R_1,R_2)\in\du^2$ and all $X\in \mathcal{A}$, $$\varphi_X(R_1,R_2)=\sum_{\{i\mid \tau(R_i)=X\}}\epsilon_i,$$ 
implying $\varphi$ is random dictatorial.  This completes the proof of the base case.
\end{proof}

\section{Characterization of Unanimous and Strategy-proof PSCCs on \texorpdfstring{$\de$}{D\_E}}\label{sec:de}
In this section, we explore the structure of unanimous and strategy-proof PSCCs on the $\de$ domain. We first define a new class of PSCCs that we refer to as random bi-dictatorial rules. These rules are a probabilistic generalization of the bi-dictatorial rules introduced in \cite{barbera2001strategy}. A DSCC $f:\de^n \to \mathcal{A}$ is bi-dictatorial if there exist agents $i,j\in N$ (possibly identical) such that for every $R_N\in \de$, $f(R_N)=\tau(R_i)\cup\tau(R_j)$. Next, we proceed to formally define a random bi-dictatorial rule. To do so, we denote by $N^{(2)}$ the set of ordered pairs from $N$ where the first element is less than the second. More formally, $N^{(2)}=\{(i,j)\in N^2\mid i\leq j\}$.

\begin{definition}
    A PSCC $\varphi: \de^n\to \Delta \mathcal{A}$ is a 
    random bi-dictatorial rule if there exist $\beta_{ij}\geq 0$ for all $(i,j)\in N^{(2)}$ with  $$\sum_{(i,j)\in N^{(2)}}\beta_{ij}=1$$ such that for all $R_N\in \de^n$ and all $X\in \mathcal{A}$ 
    $$\varphi_X(R_N)=\sum_{\{(i,j)\in N^{(2)} \mid \tau(R_i)\cup \tau(R_j)=X\}}\beta_{ij}.$$

    We call $\{\beta_{ij}\}_{(i,j)\in N^{(2)}}$ as the parameters of the PSCC $\varphi$.
\end{definition}

  The following example provides an illustration of a random bi-dictatorial rule.
  
      \begin{example}
Let $N=\{1,2,3\}$ and $A=\{a,b,c\}$. Here $N^{(2)}=\{(1,1),(1,2),(1,3),(2,2),(2,3),(3,3)\}$. Consider the random bi-dictatorial rule $\varphi$ with parameters
$\beta_{11}=0.2$, $\beta_{12}=0.2$, $\beta_{13}=0$, 
$\beta_{22}=0.1$, $\beta_{23}=0.4$, and $\beta_{33}=0.1$. Suppose that at a profile $R_N\in \de^3$, we have 
$\tau(R_1)=\{a\}$, $\tau(R_2)=\{a\}$, and $\tau(R_3)=\{b\}$. Then
\begin{align*}
    \varphi_{\{a\}}(R_N)&=\beta_{11}+\beta_{22}+\beta_{12}=0.5,\\
    \varphi_{\{b\}}(R_N)&=\beta_{33}=0.1, \text{and}\\
    \varphi_{\{a,b\}}(R_N)&=\beta_{13}+\beta_{23}=0.4,
\end{align*}
 while all other subsets of $A$ receive probability zero. \hfill $\square$
\end{example}
   \begin{remark}
     A random bi-dictatorial rule $\varphi: \de^n\to \Delta \mathcal{A}$ is characterized by $\frac{n(n+1)}{2}$ parameters. When the parameters corresponding to distinct agent pairs $(i, j)$ are set to zero, the rule collapses to a standard random dictatorship. Thus, random bi-dictatorial rules serve as a natural generalization of random dictatorial rules. Moreover, these rules satisfy the tops-only property, and as every preference in the domain $\de$ has a singleton set as top-ranked, a random bi-dictatorial rule assigns positive probability exclusively to the singletons and doubletons composed of the agents' top-ranked alternatives. Formally, if $\varphi$ is a random bi-dictatorial rule, for all $R_N\in \de^n$, $\varphi_X(R_N)>0$ if $X\subseteq \tau(R_N)$ and $|X|\leq 2$.
    \hfill $\square$
    \end{remark}

    % Moreover, for two profiles $R_N$ and $(R_i',R_{-i})$ with $\tau(R_i)\neq \tau(R'_i)$, the outcomes of $\varphi$ at $R_N$ and $(R_i',R_{-i})$ remains the same for all sets other than the sets of the form $\tau(R_i)\cup \{w\}$ and $\tau(R'_i)\cup \{w\}$ where $w\in A$, i.e.,
    % \begin{equation}
    %     \varphi_{X}(R_N)=\varphi_X(R_N') \text{ for all }X\notin \{\tau(R_i)\cup \{w\},\tau(R'_i)\cup \{w\}\} \text{ where }w\in A.
    % \end{equation}
Next we state an important property of the random bi-dictatorial rules that shows every random bi-dictatorial rule is a convex combination of bi-dictatorial DSCCs introduced in \cite{barbera2001strategy}. Let $\{\phi^1,\ldots, \phi^r\}$ be a set of PSCCs. We say a PSCC $\varphi:\mathcal{D}^n \to \Delta \mathcal{A}$ is a convex combination of a set of PSCCs $\{\phi^1,\ldots, \phi^r\}$ if there exists a set of non-negative numbers $\{\alpha_1,\ldots, \alpha_r\}$ with $\sum_{l=1}^r\alpha_l=1$ such that for all $R_N\in \de^n$ and all $X\in \mathcal{A}$, $$\varphi_X(R_N)=\sum_{l=1}^r\alpha_l\phi^l_X(R_N).$$ 
\begin{lemma}\label{lem_uk_1}
    Every random bi-dictatorial rule is a convex combination of bi-dictatorial rules.
\end{lemma}
\begin{proof}
    Let $\varphi$ be a random bi-dictatorial rule with parameters
$\{\beta_{ij}\}_{(i,j)\in N^{(2)}}$. Further, let $f^{ij}$ for $(i,j)\in N^{(2)}$ be the bi-dictatorial rule corresponding to $(i,j)$. That is, for all $R_N\in \de^n$, $f^{ij}(R_N)=\tau(R_i)\cup \tau(R_j)$. From the definition of a random bi-dictatorial rule, we have for all $R_N\in \de^n$ and all $X\in \mathcal{A}$,
\begin{align*}
    \varphi_X(R_N)&=\sum_{\{(i,j)\in N^{(2)} \mid \tau(R_i)\cup \tau(R_j)=X\}}\beta_{ij} \\
    &=\sum_{(i,j)\in N^{(2)}}\beta_{ij} 1_{\{\tau(R_i)\cup \tau(R_j)=X\}}\\
    &=\sum_{(i,j)\in N^{(2)}}\beta_{ij} f^{ij}_X(R_N),
\end{align*}
where the last equality follows as $1_{\{\tau(R_i)\cup \tau(R_j)=X\}}=f^{ij}_X(R_N)$. This, together with $\sum_{(i,j)\in N^{(2)}}\beta_{ij}=1$, shows that $\varphi$ is a convex combination of the bi-dictatorial rules.
\end{proof}

\begin{remark}\label{rem_star}
   In view of Lemma \ref{lem_uk_1} and the fact that bi-dictatorial rules are strategy-proof on the $\de$ domain (Theorem 3.2 of \cite{barbera2001strategy}), we can conclude that random bi-dictatorial rules are strategy-proof on the $\de$ domain for any number of agents. Moreover, it follows from the definition of a random bi-dictatorial rule that random bi-dictatorial rules are unanimous. \hfill $\square$ 
\end{remark}

We now turn to our results in this section. We begin with a crucial theorem establishing that any unanimous and strategy-proof PSCC satisfies the tops-only property. Furthermore, at any given preference profile, the PSCC assigns positive probability exclusively to subsets that have a cardinality of at most two, provided they are subsets of the union of the agents' top-ranked sets.

\begin{theorem} \label{de_tops-only}
Let $\varphi: \de^n \to \Delta(\mathcal{A})$ be a unanimous and strategy-proof PSCC. Then $\varphi$ is tops-only, and for all profiles $R_N \in \mathcal{D}_E^n$, $\varphi_X(R_N) > 0$ implies $X \subseteq \tau(R_N)$ with $|X| \leq 2$.
\end{theorem}
Due to its length and technical nature, the proof of this theorem is relegated to the supplementary material (Appendix \ref{appen_supple}). In view of Theorem \ref{de_tops-only}, if one agent changes her preference, the outcome may change only if the top-ranked set is changed. We now investigate how the change happens when an agent changes her preference to a preference with a different top-ranked set. To do so, we begin with a lemma that asserts the existence of a particular kind of preference relation within this domain. Formally, the lemma establishes that for any pair of distinct alternatives $a, b \in A$, there exists a strict preference $\bar{R} \in \mathcal{D}_E$ under which the singleton $\{a\}$, the doubleton $\{a, b\}$, and the singleton $\{b\}$ are ranked as the first, second, and third most-preferred sets, respectively. Furthermore, for every remaining alternative $x \in A \setminus \{a, b\}$, the doubleton sets $\{a, x\}$ and $\{b, x\}$ are ranked consecutively in $\bar{R}$. This specific preference structure will play a crucial role in the proof of the main theorem of this section. The formal proof of the lemma can be found in the supplementary material (Appendix \ref{appen_supple}).

\begin{lemma}\label{de_lem_5}
 For distinct $a,b\in A$, there exists a strict preference $\bar{R}\in \de$ such that
    \begin{enumerate}
        \item [(i)] $\bar{R}\equiv \{a\}\{a,b\}\{b\}\cdots$ and
        \item [(ii)] for any $x\in A\setminus\{a,b\}$, 
        \ $\{a,x\}$ and $\{b,x\}$ are adjacent in $\bar{R}$, i.e., there is no $W\in \mathcal{A}$ such that $ \{a,x\}\bar{P}W\bar{P}\{b,x\}$.
    \end{enumerate}
\end{lemma}

In the following lemma, we demonstrate the change in outcome when an agent changes her top-ranked set.

\begin{lemma}\label{de_lem_6}
	Consider $R_N\in\de^n$ and $R_i'\in \de$ such that $\tau(R_i)=\{p\}$ and $\tau(R_i')=\{r\}$ for some $p\neq r$. Then,
	\begin{enumerate}
				\item [(i)] $\varphi_{\{\{p\},\{p,r\},\{r\}\}}(R_N)=\varphi_{\{\{p\},\{p,r\},\{r\}\}}(R_i',R_{-i})$,
                \item [(ii)]$\varphi_X(R_N)=\varphi_X(R_i',R_{-i})$ for all $X\notin \{\{p,w\},\{r,w\}\mid w\in \tau(R_{-i})\cup \{p,r\}\}$,
				 and
				\item [(iii)] $\varphi_{\{\{p,w\},\{r,w\}\}}(R_N)=\varphi_{\{\{p,w\},\{r,w\}\}}(R_i',R_{-i})$ for all $w\in \tau(R_{-i})\setminus \{p,r\}$.
			\end{enumerate}
		
	\end{lemma}

\begin{proof}  Note that by Theorem \ref{de_tops-only}, $\varphi$ is tops-only, and 
 \begin{equation}\label{eq_de_1}
\begin{aligned}
\varphi_X(R_N)>0 
&\implies X\subseteq \tau(R_N)\ \text{and}\ |X|\leq 2, \text{ and}\\
\varphi_X(R_i',R_{-i})>0 
&\implies X\subseteq \tau(R_i',R_{-i})\ \text{and}\ |X|\leq 2.
\end{aligned}
\end{equation}
 Now, in view of tops-onlyness of $\varphi$, we may assume the following structure of $R_i$, 
    \begin{enumerate}
    \item [(a)]  $R_i$ is a strict preference,
        \item [(b)] $R_i\equiv \{p\}\{p,r\}\{r\}\cdots$ and
        \item [(c)] for any $x\in A\setminus\{p,r\}$, 
        \ $\{p,x\}$ and $\{r,x\}$ are adjacent in $R$.
    \end{enumerate}
The existence of such a preference is shown in Lemma \ref{de_lem_5}. Let $v$ be the utility function corresponding to $R_i$. Again, by tops-onlyness, we may assume that $v'$ (the utility function corresponding to $R_i'$) is of the following form
\[v'(x) = 
  \begin{cases} 
   v(x) & \text{if } x\notin \{p,r\}, \\
  v(r) & \text{if } x= p,\\
  v(p) & \text{if } x=r.
  \end{cases}
\]
This means $R_i'$, will have the same structure as in $R_i$ except for sets involving either $p$ or $r$ (not both). Formally,
\begin{enumerate}
    \item [(a)] $R_i'$ is a strict preference,
        \item [(b)] $R_i'\equiv \{r\}\{p,r\}\{p\}\cdots$ and
        \item [(c)] for any $x\in A\setminus\{p,r\}$, 
        \ $\{r,x\}$ and $\{p,x\}$ are adjacent in $R_i'$.
    \end{enumerate}

    Let $\mathcal{A}_{1,2}$ be the set of all elements of $\mathcal{A}$ with cardinality one or two. Note that by the above construction of $R'_i$ from $R_i$, the relative position of a set $Z\in \mathcal{A}_{1,2}$, not involving $p$ and $r$, remains the same in both preferences with respect to all other sets in $\mathcal{A}_{1,2}$. Thus, for any $Z\in \mathcal{A}_{1,2}$  
\begin{equation}\label{de_eq_0.1}
U(Z,R_i)\cap\mathcal{A}_{1,2}=U(Z,R'_j)\cap{\mathcal{A}_{1,2}}.
\end{equation}
 Moreover, for $w\in A\setminus \{p,r\}$, the relative position of $\{p,w\}$ (or $\{r,w\}$) in $\mathcal{A}_{1,2}$ change with respect to only $\{r,w\}$ (or $\{p,w\}$). Thus, for any $w\in A\setminus \{p,r\}$
 \begin{equation}\label{de_eq_0.2}
   U(\{r,w\},R_i)\cap{\mathcal{A}_{1,2}}=U(\{p,w\},R'_i)\cap{\mathcal{A}_{1,2}} \text{ and } U^{-}(\{p,w\},R_i)\cap{\mathcal{A}_{1,2}}=U^-(\{r,w\},R'_i)\cap{\mathcal{A}_{1,2}}.
\end{equation}

  We now complete the proof of the lemma. For (a), note that this directly follows by applying strategy-proofness as $R_i\equiv \{p\}\{p,r\}\{r\}\cdots$  and 
    $R_i'\equiv \{r\}\{p,r\}\{p\}\cdots$. For (b), assume $X\notin \{\{p,w\},\{r,w\}\mid w\in \tau(R_{-i})\cup \{p,r\}\}$. If $X$ is such $|X|\geq 3$, by (\ref{eq_de_1}), $\varphi_X(R_N)=0=\varphi_X(R_i',R_{-i})$.

     So, assume that $|X|\leq 2$. By (\ref{eq_de_1}), as  probabilities are allotted only to sets in $\mathcal{A}_{1,2}$, by strategy-proofness and (\ref{de_eq_0.1}), we have $$\varphi_{U(X,R_i)}(R_N)=\varphi_{U(X,R_i)}(R_i',R_{-i}).$$ Moreover, as $I(X,R_i)=I(X,R_i')=X$, (\ref{de_eq_0.1}) implies $U^-(X,\bar{R}_i)\cap {\mathcal{A}_{1,2}}=U^-(X,R_i')\cap {\mathcal{A}_{1,2}}$. Thus, again by strategy-proofness, $$\varphi_{U^-(X,R_i)}(R_N)=\varphi_{U^-(X,R_i)}(R_i',R_{-i}).$$ Combining these two above equations together, we have $\varphi_{X}(R_N)= \varphi_{X}(R'_i,R_{-i})$.

 Finally, we show (c). Consider $w\in A\setminus \{p,r\}$. By (\ref{eq_de_1}), as  probabilities are allotted only to sets in $\mathcal{A}_{1,2}$, strategy-proofness and (\ref{de_eq_0.2}) together imply that $$\varphi_{U(\{r,w\},R_i)}(R_N)=\varphi_{U(\{r,w\},R_i)}(R_i',R_{-i}) \text{ and } \varphi_{U^-(\{x,w\},R_i)}(R_N)=\varphi_{U^-(\{x,w\},R_i)}(R_i',R_{-i}).$$ As $U(\{r,w\},R_i)=U^-(\{p,w\},R_i) \cup \{\{p,w\},\{r,w\}\}$, the two equations above imply that $$\varphi_{\{\{p,w\},\{r,w\}\}}(R_N)=\varphi_{\{\{p,w\},\{r,w\}\}}(R_i',R_{-i}).$$ This completes the proof of the lemma.  \end{proof}

With the underlying structure established by Theorem \ref{de_tops-only} and Lemma \ref{de_lem_6}, we proceed to characterize the entire class of unanimous and strategy-proof PSCCs on the $\de$ domain. We structure our results according to the number of agents. This organization is motivated by the fact that our findings for $n \geq 4$ agents further depend on the number of alternatives. 

\subsection{The case \texorpdfstring{$n\leq 3$}{n <= 3}}

Our next result characterizes the unanimous and strategy-proof PSCCs as random bi-dictatorial rules when there are at most three agents.

\begin{theorem}\label{de_2}
Let $n\leq 3$ and $\varphi:\de^n\to \Delta \mathcal{A}$ be a PSCC. Then $\varphi$ is unanimous and strategy-proof if and only if $\varphi$ is a random bi-dictatorial rule.
\end{theorem}

\begin{proof} 
The if part of the theorem follows from Remark \ref{rem_star}. For the only-if part, we provide the proof for $n=3$. For $n=2$, the arguments are similar and much less complicated. To show that $\varphi$ is a random bi-dictatorial rule, we first define the set of parameters, $\{\alpha_{pq}\}_{(p,q)\in N^{(2)}}$. For simplifying the notations, in the rest of the proof, we will be denoting a profile in terms of its top-ranked sets of agents. That is, we write  $(\bar{R}_1,\bar{R}_2,\bar{R}_3)\equiv(\{p\},\{q\},\{r\})$ where $\tau(\bar{R}_1)=\{p\}$, $\tau(\bar{R}_2)=\{q\}$, and $\tau(\bar{R}_3)=\{r\}$. Note that as $\varphi$ is tops-only, this identification is sufficient for our purpose. Consider three distinct alternatives $a$, $b$, and $c$ (this is possible as $|A|\geq 3$) and a preference profile $(\{a\},\{b\},\{c\})$ in $\de^3$. Then, define for distinct agents $i,j$
$$\alpha_{ii}\coloneqq\varphi_{\{a\}}(\{a\},\{b\},\{c\})$$ and $$\alpha_{ij}\coloneqq\varphi_{\{a,b\}}(\{a\},\{b\},\{c\}).$$
Next, we show that $\varphi$ is the random bi-dictatorial rule with respect to the parameters $\{\alpha_{pq}\}_{(p,q)\in N^{(2)}}$. We start with showing that $\sum_{(i,j)\in N^{(2)}}\alpha_{ij}=1$. Consider the preference profile $(\{a\},\{b\},\{c\})\in \de^3$. Therefore, by the definition of $\{\alpha_{pq}\}_{(p,q)\in N^{(2)}}$, we have 
\begin{align*}
\alpha_{11}&=\varphi_{\{a\}}(\{a\},\{b\},\{c\}), &
\alpha_{22}&=\varphi_{\{b\}}(\{a\},\{b\},\{c\}), &
\alpha_{33}&=\varphi_{\{c\}}(\{a\},\{b\},\{c\}), \\
\alpha_{12}&=\varphi_{\{a,b\}}(\{a\},\{b\},\{c\}), &
\alpha_{13}&=\varphi_{\{a,c\}}(\{a\},\{b\},\{c\}), &
\alpha_{23}&=\varphi_{\{b,c\}}(\{a\},\{b\},\{c\}).
\end{align*}
As by Theorem \ref{de_tops-only}, $\varphi_X(\{a\},\{b\},\{c\})=0$ for all $X\notin \{\{a\},\{b\},\{c\},\{a,b\},\{b,c\},\{a,c\}\}$, hence,
$\sum_{(i,j)\in N^{(2)}}\alpha_{ij}=\sum_{X\in \mathcal{A}}\varphi_X(\{a\},\{b\},\{c\})=1$.

\vspace{2mm}

We now show that for all $(\{x_1\},\{x_2\},\{x_3\})\in \de^3$ and all $X\in \mathcal{A}$,
\begin{equation}\label{de_feq}
\varphi_X(\{x_1\},\{x_2\},\{x_3\})=\sum_{\{(i,j)\in N^{(2)} \mid \{x_i\}\cup \{x_j\}=X\}}\alpha_{ij}.
\end{equation}
Note that if $(\{x_1\},\{x_2\},\{x_3\})$ is a unanimous profile, then in view of unanimity, (\ref{de_feq}) holds. So, we need to consider two possibilities, either $|\{\{x_1\},\{x_2\},\{x_3\}\}|=3$ or $|\{\{x_1\},\{x_2\},\{x_3\}\}|=2$.  We first consider $|\{\{x_1\},\{x_2\},\{x_3\}\}|=3$. Fix distinct $x,y,z\in A$ and consider the profile $(\{x\},\{y\},\{z\})\in \de^3$. Recall that by the definition of $\{\alpha_{pq}\}_{(p,q)\in N^{(2)}}$, (\ref{de_feq}) holds for the profile $(\{a\},\{b\},\{c\})$. We start with a claim which shows that $\varphi(\{x\},\{y\},\{z\})$ satisfies (\ref{de_feq}) for all $X\notin \{\{x,y\},\{y,z\},\{x,z\}\}$. We prove it for any three distinct $p,q,r$ so that it can be used for the remaining proof. 
\begin{claim}\label{de_cl_1}
    For any distinct $\{p\},\{q\},\{r\}$, $\varphi(\{p\},\{q\},\{r\})$ satisfies (\ref{de_feq}) for all $X\notin \{\{p,q\},\{q,r\},\{p,r\}\}$.
\end{claim}
\textbf{Proof of the claim:} In view of Theorem \ref{de_tops-only}, it is equivalent to show that $\varphi(\{p\},\{q\},\{r\})$ satisfies (\ref{de_feq}) for all $X\in \{\{p\},\{q\},\{r\}\}$. Further, as $\varphi(\{a\},\{b\},\{c\})$ satisfies (\ref{de_feq}) for all $X\in \mathcal{A}$, without loss of generality, this is equivalent to show that $\varphi_{\{p\}}(\{p\},\{q\},\{r\})=\varphi_{\{a\}}(\{a\},\{b\},\{c\})$. We distinguish two cases. First, assume that $p\notin \{b,c\}$, and consider two profiles $(\{a\},\{b\},\{c\})$ and $(\{p\},\{b\},\{c\})$. If $p=a$ then we have $\varphi_{\{p\}}(\{p\},\{b\},\{c\})=\varphi_{\{a\}}(\{a\},\{b\},\{c\})$, else by Lemma \ref{de_lem_6} (i), we have $\varphi_{\{\{a\},\{a,p\},\{p\}\}}(\{p\},\{b\},\{c\})=\varphi_{\{\{a\},\{a,p\},\{p\}\}}(\{a\},\{b\},\{c\})$. Moreover as $p\notin \{b,c\}$ and $a,b,c$ are distinct, Theorem \ref{de_tops-only} yields $\varphi_{\{\{a\},\{a,p\}\}}(\{p\},\{b\},\{c\})=\varphi_{\{\{a,p\},\{p\}\}}(\{a\},\{b\},\{c\})=0$. Hence, $\varphi_{\{p\}}(\{p\},\{b\},\{c\})=\varphi_{\{a\}}(\{a\},\{b\},\{c\})$. Now, as $p\notin \{b,c\}$ and $p,q,r$ are distinct, by Lemma \ref{de_lem_6} (ii), we have $\varphi_{\{p\}}(\{p\},\{b\},\{c\})=\varphi_{\{p\}}(\{p\},\{q\},\{c\})=\varphi_{\{p\}}(\{p\},\{q\},\{r\})$. Combining this with the previous equation gives us $\varphi_{\{p\}}(\{p\},\{q\},\{r\})=\varphi_{\{a\}}(\{a\},\{b\},\{c\})$. 

Now, assume that $p\in \{b,c\}$ and, without loss of generality, assume that $p=b$. Note that this means $b\notin \{q,r\}$. We will show that $\varphi_{\{a\}}(\{a\}\{b\},\{c\})=\varphi_{\{b\}}(\{b\}\{q\},\{r\})$. As $a,b,c$ are distinct, we have
\begin{align*}
  \varphi_{\{a\}}(\{a\},\{b\},\{c\})&=  \varphi_{\{a\}}(\{a\},\{c\},\{c\}) \hspace{20mm}(\text{applying Lemma \ref{de_lem_6} (ii)})\\
  &=\varphi_{\{\{a\},\{a,b\},\{b\}\}}(\{a\},\{c\},\{c\})\hspace{5mm}(\text{as $b\notin \{a,c\}$, by Theorem \ref{de_tops-only}, $\varphi_{\{\{b\},\{b,c\}\}}(\{a\},\{c\},\{c\})=0$)}\\
  &=\varphi_{\{\{a\},\{a,b\},\{b\}\}}(\{b\},\{c\},\{c\})\hspace{5mm}(\text{applying Lemma \ref{de_lem_6} (i)})\\
  &=\varphi_{\{b\}}(\{b\},\{c\},\{c\})\hspace{20mm}(\text{as $a\notin \{b,c\}$, by Theorem \ref{de_tops-only}, $\varphi_{\{\{a\},\{a,b\}\}}(\{b\},\{c\},\{c\})=0$)}\\
  &=\varphi_{\{b\}}(\{b\},\{c\},\{r\})\hspace{20mm}(\text{applying Lemma \ref{de_lem_6} (ii) as $b\notin \{c,r\}$})\\
  &=\varphi_{\{b\}}(\{b\},\{q\},\{r\})\hspace{20mm}(\text{applying Lemma \ref{de_lem_6} (ii) as $b\notin \{q,r\}$}).
\end{align*}
This completes the proof of the claim.
\hfill $\square$

Next we show that $\varphi(\{x\},\{y\},\{z\})$ satisfies (\ref{de_feq}) for all $X\in \{\{x,y\},\{y,z\},\{x,z\}\}$ as well. Without loss of generality, it is enough to show that $\varphi_{\{x,y\}}(\{x\},\{y\},\{z\})=\varphi_{\{a,b\}}(\{a\},\{b\},\{c\})$. Similar to the proof of Claim \ref{de_cl_1}, we consider a few cases.

\vspace{2mm}
\textbf{Case 1:} $\{a,b\}=\{x,y\}$ \\
If $a=x$ and $b=z$, we have $\varphi_{\{a,b\}}(\{a\},\{b\},\{c\})=\varphi_{\{x,y\}}(\{x\},\{y\},\{c\})$. Moreover, as $c$ is distinct from $a,b$, it follows that $c\notin \{x,y\}$. This, together with the fact that $z$ is distinct from $x,y$ and Lemma \ref{de_lem_6} (ii), implies $\varphi_{\{x,y\}}(\{x\},\{y\},\{c\})=\varphi_{\{x,y\}}(\{x\},\{y\},\{z\})$. Thus, $\varphi_{\{x,y\}}(\{x\},\{y\},\{z\})=\varphi_{\{a,b\}}(\{a\},\{b\},\{c\})$. Now, assume that $a=y$ and $b=x$. We have
\begin{align*}
    \varphi_{\{\{a\},\{a,b\},\{b\}\}}(\{a\},\{b\},\{c\})&=\varphi_{\{\{a\},\{a,b\},\{b\}\}}(\{b\},\{b\},\{c\})\hspace{20mm}(\text{applying Lemma \ref{de_lem_6} (i))}\\
    &=\varphi_{\{\{a\},\{a,b\},\{b\}\}}(\{b\},\{a\},\{c\})\hspace{20mm}(\text{applying Lemma \ref{de_lem_6} (i))}.
\end{align*}
However, as by Claim \ref{de_cl_1}, $\varphi_{\{a\}}(\{a\},\{b\},\{c\})=\varphi_{\{b\}}(\{b\},\{a\},\{c\})$ and $\varphi_{\{b\}}(\{a\},\{b\},\{c\})=\varphi_{\{a\}}(\{b\},\{a\},\{c\})$, from the above equation, we have $\varphi_{\{a,b\}}(\{a\},\{b\},\{c\})=\varphi_{\{a,b\}}(\{b\},\{a\},\{c\})$. Further, as $c\notin \{a,b\}$ and $z\notin \{x,y\}=\{a,b\}$, Lemma \ref{de_lem_6} (ii) implies $\varphi_{\{a,b\}}(\{b\},\{a\},\{c\})=\varphi_{\{a,b\}}(\{b\},\{a\},\{z\})$. Combining all these, we have $\varphi_{\{a,b\}}(\{a\},\{b\},\{c\})=\varphi_{\{a,b\}}(\{b\},\{a\},\{z\})$. This completes the proof for Case 1. 

\vspace{2mm}
\textbf{Case 2:} $\{a,b\}\neq \{x,y\}$ \\
Without loss of generality, we assume that $x\notin \{a,b\}$. We further distinguish two cases.

\vspace{2mm}
\textbf{Case 2.1:} $c\notin \{x,y\}$ \\
By Lemma \ref{de_lem_6} (iii), $\varphi_{\{\{a,b\},\{x,b\}\}}(\{a\},\{b\},\{c\})=\varphi_{\{\{a,b\},\{x,b\}\}}(\{x\},\{b\},\{c\})$. However, as $x\notin \{a,b,c\}$ and $a\notin \{x,b,c\}$, by Theorem \ref{de_tops-only}, we have $\varphi_{\{x,b\}}(\{a\},\{b\},\{c\})=\varphi_{\{a,b\}}(\{x\},\{b\},\{c\})=0$. Hence, $\varphi_{\{a,b\}}(\{a\},\{b\},\{c\})=\varphi_{\{x,b\}}(\{x\},\{b\},\{c\})$. Now, if $b=y$, we have $\varphi_{\{a,b\}}(\{a\},\{b\},\{c\})=\varphi_{\{x,y\}}(\{x\},\{y\},\{c\})$, else, by Lemma \ref{de_lem_6} (iii), $\varphi_{\{\{x,b\},\{x,y\}\}}(\{x\},\{b\},\{c\})=\varphi_{\{\{x,b\},\{x,y\}\}}(\{x\},\{y\},\{c\})$. This, together with $y\notin \{x,b,c\}$ and $b\notin \{x,y,c\}$, implies  $\varphi_{\{x,y\}}(\{x\},\{b\},\{c\})=\varphi_{\{x,b\}}(\{x\},\{y\},\{c\})=0$. Thus, we have $\varphi_{\{x,b\}}(\{x\},\{b\},\{c\})=\varphi_{\{x,y\}}(\{x\},\{y\},\{c\})$. Finally, as $c\notin \{x,y\}$ and $z\notin \{x,y\}$, Lemma \ref{de_tops-only} (b) implies $\varphi_{\{x,y\}}(\{x\},\{y\},\{c\})=\varphi_{\{x,y\}}(\{x\},\{y\},\{z\})$. Thus, combining all these we have, $\varphi_{\{a,b\}}(\{a\},\{b\},\{c\})=\varphi_{\{x,y\}}(\{x\},\{y\},\{z\})$. This completes the proof for Case 2.1. 

\vspace{2mm}
\textbf{Case 2.2:} $c\in \{x,y\}$ \\
WLG assume that $c=x$. We consider two subcases under this.

\vspace{2mm}
\textbf{Case 2.2.1:} $a\neq y$ \\
As $c=x$, we have to show that $\varphi_{\{a,b\}}(\{a\},\{b\},\{c\})=\varphi_{\{c,y\}}(\{c\},\{y\},\{z\})$.  Note that as $a,b,c$ and $x,y,z$ are distinct, combining these with $a\neq y$ and $c=x$ yield $a\notin \{c,y\}$ and $z\notin \{c,y\}$. Therefore, 
\begin{align*}
     &\varphi_{\{\{a\},\{a,b\},\{b\}\}}(\{a\},\{b\},\{c\})\\&=\varphi_{\{\{a\},\{a,b\},\{b\}\}}(\{b\},\{b\},\{c\})\hspace{20mm}(\text{applying Lemma \ref{de_lem_6} (i))}\\
    &=\varphi_{\{b\}}(\{b\},\{b\},\{c\})\hspace{5mm}(\text{as $a\notin \{b,c\}$, by Theorem \ref{de_tops-only}, $\varphi_{\{\{a\},\{a,b\}\}}(\{b\},\{b\},\{c\})=0$)}\\
    &=\varphi_{\{b\}}(\{b\},\{b\},\{a\})\hspace{5mm}(\text{applying Lemma \ref{de_lem_6} (ii) and the fact that $b\notin \{a,c\}$)}\\
     &=\varphi_{\{\{c\},\{b,c\},\{b\}\}}(\{b\},\{b\},\{a\})\hspace{5mm}(\text{as $c\notin \{b,a\}$, by Theorem \ref{de_tops-only}, $\varphi_{\{\{c\},\{c,b\}\}}(\{b\},\{b\},\{a\})=0$)}\\
    &=\varphi_{\{\{c\},\{b,c\},\{b\}\}}(\{c\},\{b\},\{a\})\hspace{20mm}(\text{applying Lemma \ref{de_lem_6} (i))}\\
    &=\varphi_{\{\{c\},\{b,c\},\{b\}\}}(\{c\},\{c\},\{a\})\hspace{20mm}(\text{applying Lemma \ref{de_lem_6} (i))}\\
    &=\varphi_{\{c\}}(\{c\},\{c\},\{a\})\hspace{5mm}(\text{as $b\notin \{a,c\}$, by Theorem \ref{de_tops-only}, $\varphi_{\{\{b\},\{b,c\}\}}(\{c\},\{c\},\{a\})=0$)}\\
     &=\varphi_{\{\{c\},\{c,y\},\{y\}\}}(\{c\},\{c\},\{a\})\hspace{5mm}(\text{as $y\notin \{a,c\}$, by Theorem \ref{de_tops-only}, $\varphi_{\{\{y\},\{c,y\}\}}(\{c\},\{c\},\{a\})=0$)}\\
       &=\varphi_{\{\{c\},\{c,y\},\{y\}\}}(\{c\},\{y\},\{a\})\hspace{20mm}(\text{applying Lemma \ref{de_lem_6} (i))}\\
       &=\varphi_{\{\{c\},\{c,y\},\{y\}\}}(\{c\},\{y\},\{z\})\hspace{5mm}(\text{applying Lemma \ref{de_lem_6} (ii) as $a\notin \{c,y\}$ and $z\notin \{c,y\}$)}.
\end{align*}
However, as by Claim \ref{de_cl_1}, $\varphi_{\{a\}}(\{a\},\{b\},\{c\})=\varphi_{\{c\}}(\{c\},\{y\},\{z\})$ and $\varphi_{\{b\}}(\{a\},\{b\},\{c\})=\varphi_{\{y\}}(\{c\},\{y\},\{z\})$, from the above equation we have, $\varphi_{\{a,b\}}(\{a\},\{b\},\{c\})=\varphi_{\{c,y\}}(\{c\},\{y\},\{z\})$. This completes the proof for Case 2.2.1.

\vspace{2mm}
\textbf{Case 2.2.2:} $a=y$ \\
As $c=x$ and $a=y$, we have to show that $\varphi_{\{a,b\}}(\{a\},\{b\},\{c\})=\varphi_{\{c,a\}}(\{c\},\{a\},\{z\})$.  Note that as $x,y,z$ are distinct, $a=y$, and $c=x$, combining all these, we have $a\notin \{c,z\}$ and $c\notin \{a,z\}$. Thus, 
\begin{align*}
     \varphi_{\{\{a\},\{a,b\},\{b\}\}}(\{a\},\{b\},\{c\})&=\varphi_{\{\{a\},\{a,b\},\{b\}\}}(\{a\},\{a\},\{c\})\hspace{5mm}(\text{applying Lemma \ref{de_lem_6} (i))}\\
    &=\varphi_{\{a\}}(\{a\},\{a\},\{c\})\hspace{20mm}(\text{as $b\notin \{a,c\}$, by Theorem \ref{de_tops-only}, $\varphi_{\{\{b\},\{a,b\}\}}(\{a\},\{a\},\{c\})=0$)}\\
    &=\varphi_{\{a\}}(\{a\},\{a\},\{z\})\hspace{20mm}(\text{applying Lemma \ref{de_lem_6} (ii) as  $a\notin \{c,z\}$)}\\
     &=\varphi_{\{\{a\},\{a,c\},\{c\}\}}(\{a\},\{a\},\{z\})\hspace{5mm}(\text{as $c\notin \{a,z\}$, by Theorem \ref{de_tops-only}, $\varphi_{\{\{c\},\{a,c\}\}}(\{a\},\{a\},\{z\})=0$)}\\
    &=\varphi_{\{\{a\},\{a,c\},\{c\}\}}(\{c\},\{a\},\{z\})\hspace{5mm}(\text{applying Lemma \ref{de_lem_6} (i))}.
\end{align*}
Again, as by Claim \ref{de_cl_1}, $\varphi_{\{a\}}(\{a\},\{b\},\{c\})=\varphi_{\{c\}}(\{c\},\{a\},\{z\})$ and $\varphi_{\{b\}}(\{a\},\{b\},\{c\})=\varphi_{\{a\}}(\{c\},\{a\},\{z\})$, the above equation yields $\varphi_{\{a,b\}}(\{a\},\{b\},\{c\})=\varphi_{\{c,a\}}(\{c\},\{a\},\{z\})$. This completes the proof for Case 2.2.2.

As these cases are exhaustive, we have shown that $\varphi(\{x\},\{y\},\{z\})$ satisfies (\ref{de_feq}) for all $X\in \{\{x,y\},\{y,z\},\{x,z\}\}$. This, together with Claim \ref{de_cl_1}, implies $\varphi(\{x\},\{y\},\{z\})$ satisfies (\ref{de_feq}) for all $X\in \mathcal{A}$.

We now consider profiles $(\{x_1\},\{x_2\},\{x_3\})$ such that $|\{\{x_1\},\{x_2\},\{x_3\}\}|=2$. Without loss of generality, assume that $x_1=x_2\neq x_3$. For simplifying the notations, we write $x_1=x_2=x$ and $x_3=z$. Consider $y\notin \{x,z\}$, and by what we have already shown, $\varphi(\{x\},\{y\},\{z\})$ satisfies (\ref{de_feq}) for all $X\in \mathcal{A}$. Now, applying Lemma \ref{de_lem_6}, we have
\begin{equation}\label{de_eq_1}
    \varphi_{\{\{x\},\{x,y\},\{y\}\}}(\{x\},\{y\},\{z\})=\varphi_{\{\{x\},\{x,y\},\{y\}\}}(\{x\},\{x\},\{z\}),
\end{equation}
\begin{equation}\label{de_eq_2}
  \varphi_{\{z\}}(\{x\},\{y\},\{z\})=\varphi_{\{y\}}(\{x\},\{x\},\{z\}),   \text{ and}
\end{equation}
\begin{equation}\label{de_eq_3}
    \varphi_{\{\{x,z\},\{y,z\}\}}(\{x\},\{y\},\{z\})=\varphi_{\{\{x,z\},\{y,z\}\}}(\{x\},\{x\},\{z\}).
\end{equation}
Moreover, as $\varphi(\{x\},\{y\},\{z\})$ satisfies (\ref{de_feq}) for all $X\in \mathcal{A}$, (\ref{de_eq_1}), (\ref{de_eq_2}), and (\ref{de_eq_3}) yield
\begin{align}
   & \varphi_{\{\{x\},\{x,y\},\{y\}\}}(\{x\},\{x\},\{z\})=\alpha_{11}+\alpha_{12}+\alpha_{22},\nonumber\\
   & \varphi_{\{z\}}(\{x\},\{x\},\{z\})=\alpha_{33}\nonumber, \text{ and }\\
   & \varphi_{\{\{x,z\},\{y,z\}\}}(\{x\},\{x\},\{z\})=\alpha_{13}+\alpha_{23}.\label{de_eq_4}
\end{align}
However, as $y\notin \{x,z\}$, by Theorem \ref{de_tops-only}, $\varphi_{\{\{x,y\},\{y\},\{y,z\}\}}(\{x\},\{x\},\{z\})=0$. Thus, (\ref{de_eq_4}) becomes
\begin{align*}
   & \varphi_{\{x\}}(\{x\},\{x\},\{z\})=\alpha_{11}+\alpha_{12}+\alpha_{22},\nonumber\\
   & \varphi_{\{z\}}(\{x\},\{x\},\{z\})=\alpha_{33}\nonumber, \text{ and }\\
   & \varphi_{\{x,z\}}(\{x\},\{x\},\{z\})=\alpha_{13}+\alpha_{23}.
\end{align*}
As $\varphi_{X}(\{x\},\{x\},\{z\})=0$ for all $X\notin \{\{x\},\{x,z\},\{z\}\}$, this means $\varphi(\{x\},\{y\},\{z\})$ satisfies (\ref{de_feq}) for all $X\in \mathcal{A}$. Hence, the proof of the theorem is complete. \end{proof}

\subsection{The general case \texorpdfstring{$n\geq 4$}{n >= 4}}\label{sec:de_n}

 In this section, we present our main results for settings with $n \geq 4$ agents on the $\mathcal{D}_E$ domain.  Because the characterization varies significantly depending on the number of alternatives, we organize our results into two subsections: the first addresses the $m = 3$ case, and the second covers $m \geq 4$.

\subsubsection{The case when \texorpdfstring{$m=3$}{m = 3}}\label{de_m_3}
In this section, we assume $m=3$ and explore the structure of unanimous and strategy-proof rules on the $\de$ domain. Recall that Theorem \ref{de_2} establishes that for $n=3$, a PSCC is unanimous and strategy-proof if and only if it is a random bi-dictatorial rule. Next, we demonstrate that this characterization no longer holds when $m=3$ and $n \geq 4$. The following example illustrates the existence of a unanimous and strategy-proof PSCC on the $\de$ domain for $n=4$ and $m=3$ that is not a random bi-dictatorial rule.

\begin{example}\label{eg_2}
   Let $A = \{a, b, c\}$ and $N = \{1, 2, 3, 4\}$. We construct a PSCC ${\varphi}$ that satisfies unanimity and strategy-proofness but is not random bi-dictatorial. For simplicity, we additionally require ${\varphi}$ to satisfy anonymity and neutrality.\footnote{Anonymity requires that the outcome of the PSCC be invariant to permutations of agents: for any profile $R_N$ and any permutation $\sigma$ of the set of agents $N$, $\varphi(R_N^\sigma) = \varphi(R_N)$, where $R_N^\sigma=(R_1^\sigma,\ldots,R_n^\sigma)$ denotes the profile obtained by relabeling the preferences of agents according to $\sigma$. Neutrality requires that the outcome be invariant to permutations of alternatives: for any profile $R_N$ and any permutation $\pi$ of the set of alternatives $A$, $\varphi(\pi(R_N)) = \pi(\varphi(R_N))$, where $\pi(R_N)=(\pi(R_1),\ldots,\pi(R_n))$ is the profile obtained by replacing each alternative $x \in A$ with $\pi(x)$ in every agent's preference relation.} Note that by Theorem \ref{de_tops-only}, $\varphi$ is tops-only. Therefore, whenever needed, we identify each preference profile by the top-ranked sets of the agents. Furthermore, since $\varphi$ satisfies anonymity and neutrality, the specific identities of the agents and the alternatives are irrelevant; only the count of the top-ranked sets matters. Consequently, together with unanimity, it suffices to specify the outcome for the following three profiles: $(\{a\}, \{a\}, \{b\}, \{c\})$, $(\{a\}, \{a\}, \{a\}, \{c\})$, and $(\{a\}, \{a\}, \{c\}, \{c\})$. The corresponding outcomes are detailed in the tables below.

\begin{table}[htbp]
\centering

\begin{subtable}[t]{0.3\textwidth}
\centering
\begin{tabular}{||c c||}
\hline
Set & Outcome \\
\hline\hline
$\{a\}$ & 0.27 \\
$\{b\}$ & 0.09 \\
$\{c\}$ & 0.09 \\
$\{a,b\}$ & 0.23 \\
$\{a,c\}$ & 0.23 \\
$\{b,c\}$ & 0.09 \\
$\{a,b,c\}$ & 0 \\
\hline
\end{tabular}
\caption{Outcome at $(\{a\},\{a\},\{b\},\{c\})$} \label{de_tab_1}
\end{subtable}
\hfill
\begin{subtable}[t]{0.3\textwidth}
\centering
\begin{tabular}{||c c||}
\hline
Set & Outcome \\
\hline\hline
$\{a\}$ & 0.59 \\
$\{b\}$ & 0 \\
$\{c\}$ & 0.09 \\
$\{a,b\}$ & 0 \\
$\{a,c\}$ & 0.32 \\
$\{b,c\}$ & 0 \\
$\{a,b,c\}$ & 0 \\
\hline
\end{tabular}
\caption{Outcome at $(\{a\},\{a\},\{a\},\{c\})$} \label{de_tab_2}
\end{subtable}
\hfill
\begin{subtable}[t]{0.3\textwidth}
\centering
\begin{tabular}{||c c||}
\hline
Set & Outcome \\
\hline\hline
$\{a\}$ & 0.27 \\
$\{b\}$ & 0 \\
$\{c\}$ & 0.27 \\
$\{a,b\}$ & 0 \\
$\{a,c\}$ & 0.46 \\
$\{b,c\}$ & 0 \\
$\{a,b,c\}$ & 0 \\
\hline
\end{tabular}
\caption{Outcome at $(\{a\},\{a\},\{c\},\{c\})$} \label{de_tab_3}
\end{subtable}

\caption{Outcome of $\varphi$ at different profiles} \label{de_tab}
\end{table}
We first show that $\varphi$ satisfies strategy-proofness. As $\varphi$ is tops-only, an agent cannot change the outcome without changing the top-ranked set of her preference. Moreover, by neutrality, it's enough to consider the change $\{a\}\to \{b\}$, and by anonymity, it's enough to argue for agent 1 only. Consider $(R_1,R_2,R_3,R_4)\in \de^4$ and $R_1'\in \de$ with $\tau(R_1)=\{a\}$ and $\tau(R_1')=\{b\}$. From the outcomes presented in Table \ref{de_tab}, it can be easily seen that for such a change, $\varphi$ satisfies the following
\begin{enumerate}
    \item [(i)] $\varphi_X(R_1,R_2,R_3,R_4)=\varphi_X(R'_1,R_2,R_3,R_4)$ for all $X\in \{\{c\},\{a,b,c\}\}$,
    \item [(ii)] $\varphi_{\{\{a\},\{a,b\},\{b\}\}}(R_1,R_2,R_3,R_4)=\varphi_{\{\{a\},\{a,b\},\{b\}\}}(R'_1,R_2,R_3,R_4)$, and 
    \item [(iii)] $\varphi_{\{\{a,c\},\{b,c\}\}}(R_1,R_2,R_3,R_4)=\varphi_{\{\{a,c\},\{b,c\}\}}(R'_1,R_2,R_3,R_4)$.
\end{enumerate}
 Moreover, it can be seen (from Table \ref{de_tab}) that $\varphi_{\{a\}}(R_1,R_2,R_3,R_4)>\varphi_{\{a\}}(R'_1,R_2,R_3,R_4)$, $\varphi_{\{\{a\},\{a,b\}\}}(R_1,R_2,R_3,R_4)>\varphi_{\{\{a\},\{a,b\}\}}(R'_1,R_2,R_3,R_4)$, and $\varphi_{\{a,c\}}(R_1,R_2,R_3,R_4)>\varphi_{\{a,c\}}(R'_1,R_2,R_3,R_4)$. Now as $\tau(R_1)=\{a\}$, it must hold that $\{a\}R_1\{a,b\}R_1\{b\}$ and $\{a,c\}R_1\{b,c\}$. Combining these observations, we have that agent 1 cannot manipulate by changing her top-ranked set from $\{a\}$ to $\{b\}$. 

 We now move to show that $\varphi$ is not a random bi-dictatorial rule. Suppose it is a random bi-dictatorial rule with respect to a set of parameters $\{\beta_{ij}\}_{(i,j)\in N^{(2)}}$. From Table \ref{de_tab_1}, we have $\beta_{13}+\beta_{23}=0.23$. Consider two more profiles, $(\{b\},\{a\},\{c\},\{a\})$ and  $(\{a\},\{b\},\{c\},\{a\})$. By anonymity and neutrality, the outcomes at these two profiles will be the same as the outcome in Table \ref{de_tab_1}. This means $\beta_{13}=\varphi_{\{b,c\}}(\{b\},\{a\},\{c\},\{a\})=0.09$ and $\beta_{23}=\varphi_{\{b,c\}}(\{a\},\{b\},\{c\},\{a\})=0.09$. But this contradicts that $\beta_{13}+\beta_{23}=0.23$. This completes the verification that $\varphi$ is not a random bi-dictatorial rule. \hfill $\square$ \end{example}

 We now proceed to characterize the class of unanimous and strategy-proof PSCCs for settings with $m=3$ alternatives. The following remark restates Lemma \ref{de_lem_6} for the $m=3$ case and introduces additional constraints on the outcome when an agent deviates to a preference with a different top-ranked set. These constraints will be crucial for the characterization result that follows.

   \begin{remark}\label{de_m=3_2}
	Recall that Lemma \ref{de_lem_6} describes how the outcome changes when an agent switches to a preference with a different top-ranked set. In the special case of three alternatives, exploiting the full force of strategy-proofness yields additional restrictions beyond those stated in Lemma \ref{de_lem_6}. The conditions below provide a reformulation of Lemma \ref{de_lem_6} for $m=3$ together with several further implications of strategy-proofness. Consider a profile $R_N \in \de^n$ and a deviating preference $R_i' \in \de$ such that $\tau(R_i) = \{p\}$ and $\tau(R_i') = \{r\}$ for distinct alternatives $p, r \in A$. Let $w \in A \setminus \{p, r\}$ denote the remaining third alternative. Then,
	\begin{enumerate}
	\item [(i)]  $\varphi_{\{\{p\},\{p,r\},\{r\}\}}(R_N)=\varphi_{\{\{p\},\{p,r\},\{r\}\}}(R_i',R_{-i})$ with \begin{enumerate}
	    \item [(a)] $\varphi_{\{p\}}(R_N)\geq\varphi_{\{p\}}(R_i',R_{-i})$, and
        \item [(b)] $\varphi_{\{\{p\},\{p,r\}\}}(R_N)\geq \varphi_{\{\{p\},\{p,r\}\}}(R_i',R_{-i})$,
	\end{enumerate}
                \item [(ii)]$\varphi_X(R_N)=\varphi_X(R_i',R_{-i})$ for all $X\in \{\{w\},\{a,b,c\}\}$,
				 and
				\item [(iii)] $\varphi_{\{\{p,w\},\{r,w\}\}}(R_N)=\varphi_{\{\{p,w\},\{r,w\}\}}(R_i',R_{-i})$ with
                 $\varphi_{\{p,w\}}(R_N)\geq\varphi_{\{r,w\}}(R_i',R_{-i})$. \hfill $\square$
			\end{enumerate} 
	\end{remark}
% \begin{proof}
%     In view of the tops-onlyness of $\varphi$ (from Lemma \ref{lem_m=3_1}), we may take 
%     \begin{align*}
%      & R_i\equiv \{p\}\{p,r\}\{r\}\{p,r,w\}\{p,w\}\{r,w\}\{w\}, \text{ and } \\
%      &R_i'\equiv \{r\}\{p,r\}\{p\}\{p,r,w\}\{r,w\}\{p,w\}\{w\}.
%     \end{align*}
%     It is easy to verify the existence of such preferences in the $\de$ domain. All the statements of (i), (ii), and (iii) now follow from the strategy-proofness of $\varphi$.
% \end{proof}
As the unanimous and strategy-proof PSCCs are tops-only, it's easier to represent the outcomes at different profiles as the functions of coalitions supporting different alternatives. Accordingly, for each profile $R_N\in \mathcal{D}_E^n$ and alternative $x\in A$, let $N_x(R_N):=\{i\in N:\tau(R_i)=\{x\}\}$ denote the coalition of agents whose top alternative is $x$.
% For a unanimous and strategy-proof PSCC $\varphi$, by Proposition \ref{prop_m=3_1}, we have $\varphi$ is tops-only. This, together with the above notation, implies that the profiles may be identified as three tuples, i.e., for $R_N\in \de^n$, we may write $R_N\equiv (N_a(R_N),N_b(R_N),N_c(R_N))$.
In view of Remark \ref{de_m=3_2} and the above notation, one can immediately see that the following observations hold for any two profiles $R_N,R'_N\in \de^n$:
\begin{align}\label{de_m=3_eq_0}
& \text{(i) } N_x(R_N)=\emptyset \implies \varphi_{\{x,z\}}(R_N)=0 \text{ for all } x,z\in A, \nonumber \\ 
& \text{(ii) }  \varphi_{\{a,b,c\}}(R_N)=0 \text{ and } \nonumber \\
& \text{(iii) } N_x(R_N)=N_x(R'_N) \implies \varphi_{\{x\}}(R_N)=\varphi_{\{x\}}(R'_N) \text{ for all } x\in A.
\end{align}
In what follows, we prove a lemma that shows (\ref{de_m=3_eq_0})-(iii) holds more generally in the sense that for two alternatives $x,y\in A$, the probabilities remain the same if the supports for $x$ and $y$ do not change across two profiles.  
\begin{lemma}\label{de_m=3_3}
    Let $x,y\in A$. Suppose $R_N$ and $R_N'$ are such that $N_x(R_N)=N_y(R'_N)$. Then $\varphi_{\{x\}}(R_N)=\varphi_{\{y\}}(R'_N)$.
\end{lemma}
\begin{proof}
    If $x=y$, the proof follows directly from (\ref{de_m=3_eq_0})-(iii). So, we assume that $x\neq y$. Let $w\in A\setminus \{x,y\}$. Again, using (\ref{de_m=3_eq_0})-(iii), we may further assume that $N_w(R_N)=N\setminus N_x(R_N)$ and $N_w(R'_N)=N\setminus N_y(R'_N)$. Change the preferences of the agents in $N_x(R_N)$, one by one, to a preference with $\{y\}$ at the top. Suppose the profile we obtain after these changes is $\bar{R}_N$.
     At every such change, by Remark \ref{de_m=3_2}-(i), the total probability of $\{x\}$, $\{x,y\}$, and $\{y\}$ will remain the same, implying $$\varphi_{\{\{x\}, \{x,y\}, \{y\}\}}(R_N)=\varphi_{\{\{x\}, \{x,y\}, \{y\}\}}(\bar{R}_N).$$
     However, as $N_y(R_N)=N_x(\bar{R}_N)=\emptyset$, by (\ref{de_m=3_eq_0})-(i), $\varphi_{\{\{x,y\},\{y\}\}}(R_N)=\varphi_{\{\{x,y\},\{x\}\}}(\bar{R}_N)=0$. This, together with the previous equation, implies that $\varphi_{\{x\}}(R_N)=\varphi_{\{y\}}(\bar{R}_N)$. It only remains to show that $\varphi_{\{y\}}(\bar{R}_N)=\varphi_{\{y\}}(R'_N)$. This follows as by the construction of $\bar{R}_N$, $N_y(\bar{R}_N)=N_x(R_N)$ and by the assumption of the lemma, $N_y(R'_N)=N_x(R_N)$, implying $N_y(\bar{R}_N)=N_y(R'_N)$. \end{proof}

We now define a new class of PSCCs, called coalition-weighted rules. These rules are described in terms of a set function $v:2^N\to[0,1]$, which we call a coalition weight function. The role of $v$ is to record how much probability is assigned to an alternative as a function of the coalition of agents who support it. More precisely, for a unanimous and strategy-proof PSCC, Lemma \ref{de_m=3_3} implies that the probability assigned to a singleton $\{x\}$ depends only on the coalition of agents whose top-ranked set is $\{x\}$. The function $v$ captures exactly this dependence. The four conditions imposed on $v$ have natural interpretations. Normalization corresponds to unanimity: the empty coalition has weight zero, while the grand coalition has weight one. Monotonicity says that the weight of a coalition cannot decrease when additional agents join it; this restriction is induced by strategy-proofness. Supermodularity guarantees that the probability assigned to a doubleton set is non-negative. Finally, partition balance ensures that the probabilities assigned to all singleton and doubleton sets sum to one at every profile. Thus, the conditions on $v$ are exactly those needed to obtain a well-defined PSCC. We now give the formal definition of a coalition weight function.

\begin{definition}\label{defn_tbc}
A function $v:2^N\to[0,1]$ is a \emph{coalition weight function} if it satisfies the following conditions.
\begin{enumerate}
    \item [(i)] \textbf{Normalization:} $v(\emptyset)=0,\quad v(N)=1.$

    \item [(ii)] \textbf{Monotonicity:} for all $S,T\subseteq N$,
    $S\subseteq T \implies v(S)\leq v(T).$

    \item [(iii)] \textbf{Supermodularity:} for all $S,T\subseteq N$, $v(S\cup T)+v(S\cap T)\geq v(S)+v(T).$

    \item [(iv)] \textbf{Partition Balance:} for every three element partition $(S,T,U)$ of $N$,
    $$v(S\cup T)+v(S\cup U)+v(T\cup U)=1+v(S)+v(T)+v(U).$$
\end{enumerate}
\end{definition}

We are now ready to formally define coalition-weighted rules.

\begin{definition}
    A PSCC $\varphi:\de^n \to \Delta \mathcal{A}$ is called a coalition-weighted rule if there exists a coalition weight function $v$ such that for every profile $R_N\in \de^n$ and every $x\in A$,
$$\varphi_{\{x\}}(R_N)=v(N_x(R_N)),$$
for every distinct $x,y\in A$,
$$\varphi_{\{x,y\}}(R_N)=v(N_x(R_N)\cup N_y(R_N))-v(N_x(R_N))-v(N_y(R_N)), \text{ and }$$
$$\varphi_{\{a,b,c\}}(R_N)=0.$$
\end{definition}
We first show that the rule is well-defined. As $v:2^N\to[0,1]$, $\varphi_x(R_N)\geq 0$ for all $x\in A$ and $R_N\in\de^n$. Further, as $N_x(R_N)\cap N_y(R_N)=\emptyset$ for all $x,y\in A$ and all $R_N\in \de^n$, by supermodularity of $v$, $v(N_x(R_N)\cup N_y(R_N))\geq v(N_x(R_N))+v(N_y(R_N))$, implying $\varphi_{\{x,y\}}(R_N)\geq 0$. Finally, for all profiles $R_N\in \de$, as $\varphi_{\{a,b,c\}}(R_N)=0$, we have
\begin{align*}
    &\varphi_{a}(R_N)+\varphi_{b}(R_N)+\varphi_{c}(R_N)+\varphi_{\{a,b\}}(R_N)+\varphi_{\{b,c\}}(R_N)+\varphi_{\{a,c\}}(R_N)\\
    =&v(N_a(R_N))+v(N_b(R_N))+v(N_c(R_N))+[v(N_a(R_N)\cup N_b(R_N))-v(N_a(R_N))-v(N_b(R_N))]\\
    &+ [v(N_b(R_N)\cup N_c(R_N))-v(N_b(R_N))-v(N_c(R_N))]+[v(N_a(R_N)\cup N_c(R_N))-v(N_a(R_N))-v(N_c(R_N))]\\
    =&v(N_a(R_N)\cup N_b(R_N))+v(N_b(R_N)\cup N_c(R_N))+v(N_a(R_N)\cup N_c(R_N))-v(N_a(R_N))-v(N_b(R_N))-v(N_c(R_N))\\
    =&1 \hspace{20mm} (\text{by the partition balance of $v$}).
\end{align*}
Thus, every coalition-weighted rule is a valid PSCC. A natural question that arises is how coalition-weighted rules relate to random bi-dictatorial rules. To answer this question, we first demonstrate that the PSCC $\varphi$ introduced in Example \ref{eg_2} is, in fact, a coalition-weighted rule. Since $\varphi$ has already been shown not to be a random bi-dictatorial rule, this example suggests that coalition-weighted rules constitute a broader class.

\begin{example}
    Recall that in Example \ref{eg_2}, the PSCC $\varphi$ is anonymous and neutral. Thus, if $\varphi$ is a coalition-weighted rule, the corresponding coalition weight function $v$ will depend only on coalition size. Let $v:\{0,1,2,3,4\}\to [0,1]$ be such that 
    \begin{align*}
v(0)&=0, &
v(1)&=0.09, &
v(2)&=0.27, \\
v(3)&=0.59, &
v(4)&=1. &
&
\end{align*}
The reader may verify that the above $v$ function satisfies the four conditions in Definition \ref{defn_tbc}. To make the presentation simpler, we denote a profile by the top-ranked sets of the agents. We now compute the outcome of the PSCC $\varphi^v$ corresponding to the above $v$ at the different profiles shown in Table \ref{de_tab}. At the profile $(\{a\},\{a\},\{b\},\{c\})$, we have
\[
|N_a(\{a\},\{a\},\{b\},\{c\})|=2,\qquad |N_b(\{a\},\{a\},\{b\},\{c\})|=1,\qquad |N_c(\{a\},\{a\},\{b\},\{c\})|=1.
\]
By the definition of a coalition-weighted rule, we have
\[
\varphi^v_{\{a\}}(\{a\},\{a\},\{b\},\{c\})=v(2)=0.27,
\quad
\varphi^v_{\{b\}}(\{a\},\{a\},\{b\},\{c\})=v(1)=0.09,
\quad
\varphi^v_{\{c\}}(\{a\},\{a\},\{b\},\{c\})=v(1)=0.09,
\]
\[
\varphi^v_{\{a,b\}}(\{a\},\{a\},\{b\},\{c\})=v(3)-v(2)-v(1)=0.23,
\qquad
\varphi^v_{\{a,c\}}(\{a\},\{a\},\{b\},\{c\})=v(3)-v(2)-v(1)=0.23,
\]
\[
\varphi^v_{\{b,c\}}(\{a\},\{a\},\{b\},\{c\})=v(2)-v(1)-v(1)=0.09.
\]
Thus the outcome of $\varphi^v$ matches with the outcome of $\varphi$ at the profile $(\{a\},\{a\},\{b\},\{c\})$, presented in Table \ref{de_tab_1} of Example \ref{eg_2}. Next consider the profile $(\{a\},\{a\},\{a\},\{c\})$. Using similar calculation above, we have 
\[
\varphi^v_{\{a\}}=v(3)=0.59,
\qquad
\varphi^v_{\{c\}}=v(1)=0.09,
\]
\[
\varphi^v_{\{a,c\}}=v(4)-v(3)-v(1)=0.32.
\]
This reproduces the outcome of $\varphi$ in Table \ref{de_tab_2}. We leave it to the reader to verify that the same holds for the outcome presented in Table \ref{de_tab_3}. \hfill $\square$ \end{example}

The next theorem characterizes the unanimous and strategy-proof PSCCs on the $\de$ domain for $m=3$ as coalition-weighted rules. Together with Remark \ref{rem_star}, which establishes that every random bi-dictatorial rule is unanimous and strategy-proof on the $\de$ domain, the theorem implies that every random bi-dictatorial rule is a coalition-weighted rule. Combined with Example \ref{eg_2}, which provides a coalition-weighted rule that is not random bi-dictatorial, we obtain that the class of random bi-dictatorial rules is a strict subclass of the class of coalition-weighted rules.

\begin{theorem}\label{de_m=3_4}
    Let $m=3$ and $n\geq 4$. Then a PSCC $\varphi: \de^n \to \Delta \mathcal{A}$ is unanimous and strategy-proof if and only if it is a coalition-weighted rule.
\end{theorem}
\begin{proof}
    (\textit{Only-if part}) Let $\varphi$ be a PSCC satisfying unanimity and strategy-proofness. We will construct a coalition weight function $v$ such that $\varphi$ is a coalition-weighted rule w.r.t. $v$. The construction is as follows: for any $S\subseteq N$
    \begin{align*}
        v(S)=\varphi_{\{a\}}(R_N) \text{ where $R_N\in \de^n$ is such that $N_a(R_N)=S$.}
    \end{align*}
    This is well-defined as $\varphi$ is tops-only and by (\ref{de_m=3_eq_0})-(iii), $\varphi_{\{a\}}(\widebar{R}_N)=\varphi_{\{a\}}(\widehat{R}_N)$ if $N_a(\widebar{R}_N)=N_a(\widehat{R}_N)$. We now show that $\varphi$ is the coalition-weighted rule w.r.t. $v$, i.e., the following holds.
  \begin{enumerate}
      \item [(A)] for $x\in A$ and $R_N\in \de^n$, $v(N_x(R_N))=\varphi_{\{x\}}(R_N)$, and
      \item [(B)] for distinct $x,y\in A$ and $R_N\in \de^n$, $v(N_x(R_N)\cup N_y(R_N))-v(N_x(R_N))-v(N_y(R_N))=\varphi_{\{x,y\}}(R_N)$. 
  \end{enumerate}
  For (A), fix $R_N\in \de^n$ and $x\in A$. Note that by Lemma \ref{de_m=3_3}, $\varphi_{\{x\}}(R_N)=\varphi_{\{a\}}(\bar{R}_N)$ where $\bar{R}_N\in \de^n$ is such that $N_x(R_N)=N_a(\bar{R}_N)$. Thus, by definition, $v(N_x(R_N))=v(N_a(\bar{R}_N))=\varphi_{\{a\}}(\bar{R}_N)=\varphi_{\{x\}}(R_N)$. This shows (A). For (B), fix $R_N\in \de^n$ and distinct $x,y\in A$. Consider another profile $\bar{R}_N\in \de^n$ such that $N_x(\bar{R}_N)=N_x(R_N)\cup N_y(R_N)$, $N_y(\bar{R}_N)=\emptyset$, and $N_w(\bar{R}_N)=N_w(R_N)$ where $w\in A\setminus \{x,y\}$. By changing the preferences of agents in $N_x(R_N)$ to preferences with $\{y\}$ at the top and applying Remark \ref{de_m=3_2}-(i) at each step, we have
  \begin{equation}\label{de_m=3_eq_1}
    \varphi_{\{\{x\},\{x,y\},\{y\}\}}(\bar{R}_N)=\varphi_{\{\{x\},\{x,y\},\{y\}\}}(R_N).
  \end{equation}
  Now by (A), we have
\begin{align*}
    &\varphi_{\{x\}}(R_N)=v(N_x(R_N)) \text{ and } \varphi_{\{y\}}(R_N)=v(N_y(R_N))\\
    & \varphi_{\{x\}}(\bar{R}_N)=v(N_x(R_N)\cup N_y(R_N)).
\end{align*}
Moreover, by Remark \ref{de_m=3_2}, as $N_y(\bar{R}_N)=\emptyset$, $\varphi_{\{y\}}(\bar{R}_N)=0$ and $\varphi_{\{x,y\}}(\bar{R}_N)=0$. Combining these observations in (\ref{de_m=3_eq_1}), we have 
\begin{align*}
    &v(N_x(R_N)\cup N_y(R_N))=v(N_x(R_N))+v(N_y(R_N))+\varphi_{\{x,y\}}(R_N)\\
    \implies & \varphi_{\{x,y\}}(R_N)=v(N_x(R_N)\cup N_y(R_N))-v(N_x(R_N))-v(N_y(R_N)).
\end{align*}
Thus, (B) is shown. 
    
We now show that $v$ is a valid coalition weight function, i.e., $v$ satisfies the conditions in Definition \ref{defn_tbc}. Normalization follows as  $N_a(R_N)=\emptyset$ implies $\varphi_{\{a\}}(R_N)=0$, (by (\ref{de_m=3_eq_0})-(i)), and $N_a(R_N)=N$ implies $\varphi_{\{a\}}(R_N)=1$ (by unanimity). Monotonicity follows from repeated application of Remark \ref{de_m=3_2}-(i)-(a) and \ref{de_m=3_2}-(ii). We now prove the supermodularity of $v$. We first show that marginal increments are increasing: for every $C\subseteq D\subseteq N\setminus\{i\}$,
\begin{equation}\label{de_m=3_eq_2}
    v(D\cup\{i\})-v(D)\geq v(C\cup\{i\})-v(C).
\end{equation}
Fix $C\subseteq D\subseteq N\setminus\{i\}$. Let $x,y,z$ be the three alternatives. Consider a
profile $R_N$ such that
$$N_x(R_N)=C\cup\{i\},\qquad N_z(R_N)=D\setminus C,\qquad
N_y(R_N)=N\setminus(D\cup\{i\}).$$
Let $R'_N$ be the profile obtained from $R_N$ by changing agent $i$'s top from $x$ to $y$.
Then
$$N_x(R'_N)=C,\qquad N_z(R'_N)=D\setminus C,\qquad
N_y(R'_N)=N\setminus D.$$
By Remark \ref{de_m=3_2}-(iii), $\varphi_{\{x,z\}}(R_N)\geq \varphi_{\{x,z\}}(R'_N)$. Using (B), we have
$\varphi_{\{x,z\}}(R_N)= v(D\cup\{i\})-v(C\cup\{i\})-v(D\setminus C)$, and $\varphi_{\{x,z\}}(R'_N) = v(D)-v(C)-v(D\setminus C)$. Therefore,
\begin{align*}
&v(D\cup\{i\})-v(C\cup\{i\})-v(D\setminus C)
\geq v(D)-v(C)-v(D\setminus C) \\
\implies & v(D\cup\{i\})-v(D) \geq v(C\cup\{i\})-v(C).
\end{align*}
This proves (\ref{de_m=3_eq_2}). We now derive supermodularity. Let $S,T\subseteq N$ and $C=S\cap T$. Further, let $T\setminus S=\{i_1,\ldots,i_k\}$.
Define $B_r=\{i_1,\ldots,i_r\}$ for $r=0,1,\ldots,k$, with $B_0=\emptyset$. Applying (\ref{de_m=3_eq_2}) with
$C\cup B_{r-1}\subseteq S\cup B_{r-1}$ and $i=i_r$, we get
$$ v(S\cup B_r)-v(S\cup B_{r-1}) \geq v(C\cup B_r)-v(C\cup B_{r-1})$$ for every $r=1,\ldots,k$. Summing these inequalities for $r=1,\ldots,k$ yields
\begin{align*}
   & v(S\cup T)-v(S) \geq v(T)-v(S\cap T)\\
   \implies & v(S\cup T)+v(S\cap T)\geq v(S)+v(T).
\end{align*}
Thus, $v$ is supermodular.

Finally, we show that $v$ satisfies the partition balance. Fix a partition of $N$, say $(S,T,U)$ and let $R_N\in \de^n$ be such that $(N_a(R_N),N_b(R_N),N_c(R_N))=(S,T,U)$. Since probabilities sum to one and $\varphi_{\{a,b,c\}}(R_N)=0$ (from (\ref{de_m=3_eq_0})-(ii)), we have
\begin{align*}
    &\varphi_{\{a\}}(R_N)+\varphi_{\{b\}}(R_N)+\varphi_{\{c\}}(R_N)+\varphi_{\{a,b\}}(R_N)+\varphi_{\{b,c\}}(R_N)+\varphi_{\{a,c\}}(R_N)=1\\
    \implies & v(S\cup T)+v(S\cup U)+v(T\cup U)=
    1+v(S)+v(T)+v(U). 
\end{align*}

\noindent (\textit{If part}) Consider a coalition-weighted rule $\varphi$ w.r.t. a coalition weight function $v$. Unanimity of $\varphi$ follows directly from the normalization property of $v$, i.e., $v(N)=1$. To see strategy-proofness, consider two profiles $R_N$ and $R_N'$ where only agent $i$ changes her preference from $R_i$ to $R_i'$. As $\varphi$ is tops-only, the outcome will not change unless $R_i$ and $R_i'$ have different top-ranked sets. Suppose $\tau(R_i)=\{x\}$ and $\tau(R_i')=\{y\}$ for some distinct $x,y\in A$. Let $w$ be the third alternative. This means
\begin{align*}
N_x(R'_N)&=N_x(R_N)\setminus \{i\}, &
N_y(R'_N)&=N_y(R_N)\cup \{i\}, \text{ and } &
N_w(R'_N)&=N_w(R_N). 
\end{align*}
As $R_i\in \de$ with $\tau(R_i)=\{x\}$, we must have $\{x\}P_i\{x,y\}P_i\{y\}$ and $\{x,w\}P_i\{y,w\}$. For other sets in $\mathcal{A}$, as $N_w(R_N)=N_w(R_N')$, we have $v(N_w(R_N))=v(N_w(R_N'))$, implying $\varphi_{\{w\}}(R_N)=\varphi_{\{w\}}(R'_N)$. Moreover, by the definition of a coalition-weighted rule, $\varphi_{\{a,b,c\}}(R_N)=\varphi_{\{a,b,c\}}(R'_N)=0$. Therefore, it is enough to show the following to claim that agent $i$ cannot manipulate at $R_N$ via $R_i'$.
\begin{itemize}
    \item $\varphi_{\{\{x\},\{x,y\},\{y\}\}}(R_N)=\varphi_{\{\{x\},\{x,y\},\{y\}\}}(R'_N)$ with 
    \begin{itemize}
        \item $\varphi_{\{x\}}(R_N)\geq\varphi_{\{x\}}(R'_N)$ and
        \item $\varphi_{\{\{x\},\{x,y\}\}}(R_N)\geq\varphi_{\{\{x\},\{x,y\}\}}(R'_N)$, and
    \end{itemize}
    
    \item $\varphi_{\{x,w\}}(R_N)\geq\varphi_{\{x,w\}}(R'_N)$ and $\varphi_{\{\{x,w\},\{y,w\}\}}(R_N)=\varphi_{\{\{x,w\},\{y,w\}\}}(R'_N)$.
\end{itemize}
 First, by the monotonicity of $v$, we have $v(N_x(R_N))\geq v(N_x(R'_N))$, implying $\varphi_{\{x\}}(R_N)\geq\varphi_{\{x\}}(R'_N)$. Further, as $N_x(R_N)\cup N_y(R_N)=N_x(R_N')\cup N_x(R_N')$, we have $\varphi_{\{\{x\},\{x,y\},\{y\}\}}(R_N)=\varphi_{\{\{x\},\{x,y\},\{y\}\}}(R'_N)$. This, together with $\varphi_{\{y\}}(R'_N)\geq\varphi_{\{y\}}(R_N)$ (as $v(N_y(R'_N))\geq v(N_y(R_N))$), implies that $$\varphi_{\{\{x\},\{x,y\}\}}(R_N)\geq\varphi_{\{\{x\},\{x,y\}\}}(R'_N).$$
We now show that $\varphi_{\{x,w\}}(R_N)\geq\varphi_{\{x,w\}}(R'_N)$ and $\varphi_{\{\{x,w\},\{y,w\}\}}(R_N)=\varphi_{\{\{x,w\},\{y,w\}\}}(R'_N)$. To see $\varphi_{\{x,w\}}(R_N)\geq\varphi_{\{x,w\}}(R'_N)$, consider $C=N_x(R_N)$ and $D=N_x(R'_N)\cup N_w(R'_N)$. This means $C\cup D=N_x(R_N)\cup N_w(R_N)$ (as $N_w(R'_N)=N_w(R_N)$) and $C\cap D=N_x(R'_N)$. Applying supermodularity of $v$ for $C$ and $D$, we have
    \begin{align*}
    & v(C\cup D)+v(C\cap D)\geq v(C)+v(D)\\
    \implies & v(N_x(R_N)\cup N_w(R_N))+v(N_x(R'_N))\geq v(N_x(R_N))+v(N_x(R'_N)\cup N_w(R'_N))\\
    \implies & v(N_x(R_N)\cup N_w(R_N))-v(N_x(R_N))\geq v(N_x(R'_N)\cup N_w(R'_N))-v(N_x(R'_N))
\end{align*}
As $N_w(R'_N)=N_w(R_N)$, by subtracting $v(N_w(R_N))$ and  $v(N_w(R'_N))$ from both sides of the above inequality, we have 
\begin{align*}
    & v(N_x(R_N)\cup N_w(R_N))-v(N_x(R_N))-v(N_w(R_N))\geq v(N_x(R'_N)\cup N_w(R'_N))-v(N_x(R'_N))-v(N_w(R'_N))\\
    \implies & \varphi_{\{x,w\}}(R_N)\geq\varphi_{\{x,w\}}(R'_N).
\end{align*}
For $\varphi_{\{\{x,w\},\{y,w\}\}}(R_N)=\varphi_{\{\{x,w\},\{y,w\}\}}(R'_N)$, note that 
\begin{align*}
     \varphi_{\{\{x,w\},\{y,w\}\}}(R_N)=&v(N_x(R_N)\cup N_w(R_N))-v(N_x(R_N))-v(N_w(R_N))\\
     &+v(N_y(R_N)\cup N_w(R_N))-v(N_y(R_N))-v(N_w(R_N)) \\
     =&1-v(N_x(R_N)\cup N_y(R_N))-v(N_w(R_N)) \hspace{10mm}(\text{by the partition balance of $v$})\\
      =&1-v(N_x(R'_N)\cup N_y(R'_N))-v(N_w(R_N)) \hspace{10mm}(\text{as $N_x(R_N)\cup N_y(R_N)=N_x(R'_N)\cup N_y(R'_N)$})\\
      =&1-v(N_x(R'_N)\cup N_y(R'_N))-v(N_w(R'_N)) \hspace{10mm}(\text{as $N_w(R_N)=N_w(R'_N)$})\\
      =&\varphi_{\{\{x,w\},\{y,w\}\}}(R'_N).
\end{align*}
This completes the sufficiency part and the proof of the theorem.
\end{proof}

\subsubsection{The case when \texorpdfstring{$m\geq4$}{m >= 4}}\label{de_mgeq3}
In this section, we assume $m \geq 4$ and delve into the structure of unanimous and strategy-proof rules on the $\de$ domain. The following theorem establishes that the class of unanimous and strategy-proof PSCCs, unlike in the case $m=3$, collapses back to the family of random bi-dictatorial rules when $m\geq 4$.

\begin{theorem} \label{de_n}
Let $m\geq 4$ and $n\geq 4$. Then a PSCC $\varphi:\de^n\to \Delta \mathcal{A}$ is unanimous and strategy-proof if and only if $\varphi$ is a random bi-dictatorial rule.
\end{theorem}
\begin{proof}
The if part of the theorem follows from Remark \ref{rem_star}. We prove the only-if part of the theorem. Recall that by Theorem \ref{de_tops-only}, $\varphi$ is tops-only and for any profile $R_N\in \de^n$, $\varphi_X(R_N)>0$ implies $X\in \tau(R_N)$ with $|X|\leq 2$. We use an induction on the number of agents $n$ to prove the theorem.  Assume that the theorem holds for all sets with $k<n$ agents. For $i,j\in N$ with $i<j$, let $N^{ij} = N\setminus {\{j\}}$ and define the PSCC $g^{ij} :
\de^{n-1} \to \Delta A$ for the set of agents in
$N^{ij}$ as follows: for all $R_{N^{ij}}=(R_i,R_{-\{i,j\}}) \in \de^{n- 1}$,
\begin{equation*}
g^{ij}(R_i,R_{-\{i,j\}}) = \varphi(R_i,R_i,R_{-\{i,j\}}).
\end{equation*}
\begin{claim}\label{cl_1}
    For all distinct $i,j\in N$, $g^{ij}$ is a random bi-dictatorial rule.
\end{claim}
\noindent \textbf{Proof of the claim:} Fix distinct $i,j\in N$. In view of the induction hypothesis, it is enough to show that $g^{ij}$ is unanimous and strategy-proof. The proof of unanimity and strategy-proofness of $g^{ij}$ follows from the exact same arguments used in Claim \ref{du_cl_1}. We omit the arguments here. \hfill $\square$

We now make use of Lemma \ref{de_lem_6} to identify the parameters of the $\varphi$ to show that $\varphi$ is a random bi-dictatorial rule. We do that in the following two lemmas.

\begin{lemma}\label{lem_fc_1}
    Let $R_N$ and $R_N'$ be two preference profiles such that $\tau(R_i)=\{a\}$ and $\tau(R_i')=\{x\}$ for some $i\in N$, and $\tau(R_j)\neq\{a\}$ and $\tau(R_j')\neq \{x\}$ for all $j\neq i$. Then $\varphi_{\{a\}}(R_N)=\varphi_{\{x\}}(R_N')$.
\end{lemma}
\begin{proof}
    Consider $b\in A\setminus \{a,x\}$. As $\tau(R_j)\neq\{a\}$ for all $j\neq i$, in view of Lemma \ref{de_lem_6}, $\varphi_{\{a\}}(R_N)=\varphi_{\{a\}}(R_i,\bar{R}_{-i})$ where $\tau(\bar{R}_j)=\{b\}$ for all $j\neq i$. Similarly, as $\tau(R'_j)\neq\{x\}$ for all $j\neq i$, in view of Lemma \ref{de_lem_6}, $\varphi_{\{x\}}(R'_N)=\varphi_{\{x\}}(R'_i,\bar{R}_{-i})$. We show that $\varphi_{\{a\}}(R_i,\bar{R}_{-i})= \varphi_{\{x\}}(R'_i,\bar{R}_{-i})$. Now, again applying Lemma \ref{de_lem_6} for profiles $(R_i,\bar{R}_i)$ and $(R_i',\bar{R}_i)$, we have 
    \begin{equation}\label{fq_0}
        \varphi_{\{\{a\},\{a,x\},\{x\}\}}(R_i,\bar{R}_{-i})= \varphi_{\{\{a\},\{a,x\},\{x\}\}}(R'_i,\bar{R}_{-i}).
    \end{equation}
    However, as $a\notin \tau(R'_i,\bar{R}_{-i})$ and $x\notin \tau(R_i,\bar{R}_{-i})$, by Theorem \ref{de_tops-only}, $\varphi_{\{\{a,x\},\{x\}\}}(R_i,\bar{R}_{-i})=\varphi_{\{\{a\},\{a,x\}\}}(R'_i,\bar{R}_{-i})=0$. Thus, (\ref{fq_0}) implies $\varphi_{\{a\}}(R_i,\bar{R}_{-i})= \varphi_{\{x\}}(R'_i,\bar{R}_{-i})$, completing the proof of the lemma. 
\end{proof}

\begin{lemma}\label{lem_fc_2}
    Let $R_N$ and $R_N'$ be two preference profiles such that for  distinct $i,j\in N$, $(\tau(R_i),\tau(R_j))=(\{a\},\{b\})$ and $(\tau(R'_i),\tau(R'_j))=(\{x\},\{y\})$, and for all $k\notin \{i,j\}$, $\tau(R_k)\notin \{\{a\},\{b\}\}$ and $\tau(R'_k)\notin \{\{x\},\{y\}\}$. Then $\varphi_{\{a,b\}}(R_N)=\varphi_{\{x,y\}}(R_N')$.
\end{lemma}
\begin{proof}
By Lemma \ref{de_lem_6}, $\varphi_{\{\{a\},\{a,b\},\{b\}\}}(R_N)=\varphi_{\{\{a\},\{a,b\},\{b\}\}}(R_j, R_j,R_{-\{i,j\}})$. Moreover, as for all $k\notin \{i,j\}$, $\tau(R_k)\notin \{\{a\},\{b\}\}$, again by Theorem \ref{de_tops-only}, $\varphi_{\{\{a\},\{a,b\}\}}(R_j, R_j,R_{-\{i,j\}})=0$, implying $\varphi_{\{\{a\},\{a,b\},\{b\}\}}(R_N)=\varphi_{\{b\}}(R_j, R_j,R_{-\{i,j\}})$. Similarly, $\varphi_{\{\{x\},\{x,y\},\{y\}\}}(R'_N)=\varphi_{\{y\}}(R_j', R_j',R'_{-\{i,j\}})$. Now as $g^{ij}$ is a random bi-dictatorial rule and for all $k\notin \{i,j\}$, $\tau(R_k)\notin \{\{a\},\{b\}\}$ and $\tau(R'_k)\notin \{\{x\},\{y\}\}$, we have 
\begin{align*}
    \varphi_{\{b\}}(R_j, R_j,R_{-\{i,j\}})&= g_{\{b\}}(R_j,R_{-\{i,j\}})\\
    &=g_{\{y\}}(R'_j,R'_{-\{i,j\}})\\
    &=\varphi_{\{y\}}(R_j', R_j',R'_{-\{i,j\}}).
\end{align*}
Thus, we have $\varphi_{\{\{a\},\{a,b\},\{b\}\}}(R_N)=\varphi_{\{\{x\},\{x,y\},\{y\}\}}(R'_N)$. However, as for all $k\notin \{i,j\}$, $\tau(R_k)\notin \{\{a\},\{b\}\}$ and $\tau(R'_k)\notin \{\{x\},\{y\}\}$, by Lemma \ref{lem_fc_1}, $\varphi_{\{a\}}(R_N)=\varphi_{\{x\}}(R'_N)$ and $\varphi_{\{y\}}(R_N)=\varphi_{\{b\}}(R'_N)$, this implies  $\varphi_{\{a,b\}}(R_N)=\varphi_{\{x,y\}}(R'_N)$. \end{proof}

We are now set to define a set of parameters $\{\alpha_{pq}\}_{(p,q)\in N^{(2)}}$, taking values in $[0,1]$, and show that $\varphi$ is random bi-dictatorial rule with respect to the parameters $\{\alpha_{pq}\}_{(p,q)\in N^{(2)}}$. Consider two distinct agents $i,j\in N$ and two distinct alternatives $a,b\in A$. Further, consider a preference profile $R_N$ such that $(\tau(R_i),\tau(R_j))=(\{a\},\{b\})$ and for all $k\notin \{i,j\}$, $\tau(R_k)\notin \{\{a\},\{b\}\}$. Then, define 
$$\alpha_{ii}\coloneqq\varphi_{\{a\}}(R_N)$$ and $$\alpha_{ij}\coloneqq\varphi_{\{a,b\}}(R_N).$$
The next lemma shows that $\varphi$ is the random bi-dictatorial rule with respect to the parameters $\{\alpha_{pq}\}_{(p,q)\in N^{(2)}}$. Note that in this lemma we will, for the first time, make use of the fact that $m\geq 4$.
\begin{lemma}\label{de_lem_7} 
    For all $R_N\in \de^n$ and all $X\in \mathcal{A}$, 
    \begin{equation}\label{fq_1}
         \varphi_X(R_N)=\sum_{\{(i,j)\in N^{(2)} \mid \tau(R_i)\cup \tau(R_j)=X\}}\alpha_{ij}.
    \end{equation}
   
    \end{lemma}

\begin{proof}
Consider a profile $\bar{R}_N\in \de^n$. In view of Theorem \ref{de_tops-only}, we have $\varphi_X(\bar{R}_N)=0$ if $|X|\geq 3$ or if $|X|\leq 2$ with $x\in X$ such that $x\notin \{\tau(\bar{R}_i)\mid i\in N\}$. Thus, it remains to show (\ref{fq_1}) holds for all $X$ with $|X|\leq 2$ and $X\subseteq  \{\tau(R_i)\mid i\in N\}$. Suppose $a,b\in A$ such that $a,b\in \{\tau(\bar{R}_i)\mid i\in N\}$. First, assume that $|\{i\in N\mid \tau(\bar{R}_i)=\{a\}\}|=|\{i\in N\mid \tau(\bar{R}_i)=\{b\}\}|=1$. Without loss of generality, we may assume that $\tau(\bar{R}_1)=\{a\}$ and $\tau(\bar{R}_2)=\{b\}$. By Lemma \ref{lem_fc_1} and the definition of $\alpha$, this implies $\varphi_{\{a\}}(\bar{R}_N)=\alpha_{11}$ and $\varphi_{\{a,b\}}(\bar{R}_N)=\alpha_{12}$. Thus, (\ref{fq_1}) holds for all  $X\in \{\{x\},\{x,y\}\}$ and all $R_N$ with $|\{i\in N\mid \tau(R_i)=\{x\}\}|=|\{i\in N\mid \tau(R_i)=\{y\}\}|=1$ where $x,y\in A$. Now consider the following induction hypothesis.

\textbf{Induction Hypothesis (IH):} Assume that (\ref{fq_1}) holds for all  $X\in \{\{x\},\{x,y\}\}$ and all $R_N$ with $|\{i\in N\mid \tau(R_i)=\{x\}\}|<k$, $|\{i\in N\mid \tau(R_i)=\{y\}\}|<k$ where $x,y\in A$ and $k\leq n$.

Suppose $|\{i\in N\mid \tau(\bar{R}_i)=\{a\}\}|=k$ and $|\{i\in N\mid \tau(\bar{R}_i)=\{b\}\}|\leq k$, and we show that (\ref{fq_1}) holds for $X=\{a\}$ and $X=\{a,b\}$ at the profile $\bar{R}_N$.
% If $|\{i\in N\mid \tau(\bar{R}_i)=\{a\}\}|$ and $|\{i\in N\mid \tau(\bar{R}_i)=\{b\}\}|$ are both less than $k$, we are done by the IH. So, assume that $|\{i\in N\mid \tau(\bar{R}_i)=\{a\}\}|=k$. 
Without loss of generality, we may further assume that $\tau(R_1)=\cdots=\tau(R_k)=\{a\}$, $\tau(R_{k+1})=\cdots=\tau(R_{k+l})=\{b\}$ for some $l\leq k$. Further, in view of Lemmas \ref{lem_fc_1} and \ref{lem_fc_2}, we may assume that all the remaining agents have $\{c\}$ as their top-ranked set for some $c\in A\setminus \{a,b\}$. To ease the presentation of the remaining proof, in view of the tops-onlyness of $\varphi$, we denote a profile as an $n$-tuple of singleton subsets of $A$, denoting the top-ranked sets of the agents. More formally, we write $(\underbrace{\{x\}, \dots, \{x\}}_{p},\underbrace{\{y\}, \dots, \{y\}}_{q}, \{z\}, \dots, \{z\})$ to denote a profile $R_N$ where $\tau(R_1)=\cdots=\tau(R_p)=\{x\}$, $\tau(R_{p+1})=\cdots=\tau(R_{p+q})=\{y\}$, and the remaining agents have $\{z\}$ as their top-ranked set. In terms of this new notation, $\bar{R}_N\equiv (\underbrace{\{a\}, \dots, \{a\}}_{k},\underbrace{\{b\}, \dots, \{b\}}_{l}, \{c\}, \dots, \{c\})$. We distinguish two cases based on $l$.

\vspace{2mm}
\noindent \textbf{Case 1:} $l<k$ \\
Consider $d\in A\setminus \{a,b,c\}$ (as $m\geq 4$) and the profile $(\bar{R}'_k,\bar{R}_{-k})\equiv (\underbrace{\{a\}, \dots, \{a\}}_{k-1},\{d\},\underbrace{\{b\}, \dots, \{b\}}_{l}, \{c\}, \dots, \{c\})$ where $\tau(\bar{R}_k')=\{d\}$. As in this profile the number of agents who have their top-ranked sets as $\{x\}$, where $x\in \{a,b,d\}$, is less than $k$, we may use IH to find out the probabilities of $\{a\}$, $\{a,b\}$, $\{a,d\}$, $\{b,d\}$, and $\{d\}$ at this profile. We have

\begin{align}
   & \varphi_{\{a\}}(\underbrace{\{a\}, \dots, \{a\}}_{k-1},\{d\},\underbrace{\{b\}, \dots, \{b\}}_{l}, \{c\}, \dots, \{c\})=\sum_{i\in \{1,\ldots,k-1\}}\alpha_{ii}, \nonumber \\
   & \varphi_{\{a,b\}}(\underbrace{\{a\}, \dots, \{a\}}_{k-1},\{d\},\underbrace{\{b\}, \dots, \{b\}}_{l}, \{c\}, \dots, \{c\})=\sum_{i\in \{1,\ldots,k-1\} \text{ and } j\in \{k+1,\ldots,k+l\}}\alpha_{ij}, \nonumber \\
    & \varphi_{\{a,d\}}(\underbrace{\{a\}, \dots, \{a\}}_{k-1},\{d\},\underbrace{\{b\}, \dots, \{b\}}_{l}, \{c\}, \dots, \{c\})=\sum_{i\in \{1,\ldots,k-1\} }\alpha_{ik}, \nonumber \\
   & \varphi_{\{b,d\}}(\underbrace{\{a\}, \dots, \{a\}}_{k-1},\{d\},\underbrace{\{b\}, \dots, \{b\}}_{l}, \{c\}, \dots, \{c\})=\sum_{j\in \{k+1,\ldots,k+l\}}\alpha_{kj}, \text{ and } \nonumber \\
   & \varphi_{\{d\}}(\underbrace{\{a\}, \dots, \{a\}}_{k-1},\{d\},\underbrace{\{b\}, \dots, \{b\}}_{l}, \{c\}, \dots, \{c\})=\alpha_{kk}. \label{fq_2} 
\end{align}

Now, as only agent $k$ is changing her top-ranked set from $\{a\}$ to $\{d\}$ between $\bar{R}_N$ and $(\bar{R}_k',\bar{R}_{-k})$, using Lemma \ref{de_lem_6}, we have 
\begin{align}
    &\varphi_{\{\{a\},\{a,d\},\{d\}\}}(\bar{R}_N)=\varphi_{\{\{a\},\{a,d\},\{d\}\}}(\bar{R}_k',\bar{R}_{-k}), \text{ and } \nonumber\\ 
				&\varphi_{\{\{a,b\},\{b,d\}\}}(\bar{R}_N)=\varphi_{\{\{a,b\},\{b,d\}}(\bar{R}_k',\bar{R}_{-k}). \label{fq_3}
\end{align}
Moreover, as $d\notin \{\tau(\bar{R}_i)\mid i\in N\}$, by Theorem \ref{de_tops-only}, $\varphi_{\{\{d\},\{a,d\},\{d,b\}\}}(\bar{R}_N)=0$. This, together with (\ref{fq_2}) and (\ref{fq_3}), implies that 
\begin{align*}
    & \varphi_{\{a\}}(\underbrace{\{a\}, \dots, \{a\}}_{k},\underbrace{\{b\}, \dots, \{b\}}_{l}, \{c\}, \dots, \{c\})=\sum_{i\in \{1,\ldots,k\}}\alpha_{ii}, \text{ and }\\
   & \varphi_{\{a,b\}}(\underbrace{\{a\}, \dots, \{a\}}_{k},\underbrace{\{b\}, \dots, \{b\}}_{l}, \{c\}, \dots, \{c\})=\sum_{i\in \{1,\ldots,k\} \text{ and } j\in \{k+1,\ldots,k+l\}}\alpha_{ij}.  \\
\end{align*} Thus, (\ref{fq_1}) holds for $\{a\}$ and $\{a,b\}$ at $\bar{R}_N$ when $|\{i\in N\mid \tau(\bar{R}_i)=\{a\}\}|=k$ and $|\{i\in N\mid \tau(\bar{R}_i)=\{b\}\}|< k$. 

\vspace{2mm}
\noindent \textbf{Case 2:} $l=k$ \\
We have to show that (\ref{fq_1}) holds for $\{a\}$ and $\{a,b\}$ at $\bar{R}_N$. Consider the profile $(\bar{R}_{2k}',\bar{R}_{-2k})$ where $\tau(\bar{R}'_{2k})=\{d\}$ for some $d\in A\setminus \{a,b,c\}$. As $k$ many agents with top-ranked set $\{a\}$ and $k-1$ many agents with top-ranked set $\{b\}$ at the profile $(\bar{R}_{2k}',\bar{R}_{-2k})$, by Case 1, $\varphi_{\{a\}}(\underbrace{\{a\}, \dots, \{a\}}_{k},\underbrace{\{b\}, \dots, \{b\}}_{k-1}, \{d\}, \{c\}, \dots, \{c\})=\sum_{i\in \{1,\ldots,k\}}\alpha_{ii}$. Moreover, by Lemma \ref{de_lem_6}, $\varphi_{\{a\}}(\underbrace{\{a\}, \dots, \{a\}}_{k},\underbrace{\{b\}, \dots, \{b\}}_{k-1}, \{d\}, \{c\}, \dots, \{c\})=\varphi_{\{a\}}(\underbrace{\{a\}, \dots, \{a\}}_{k},\underbrace{\{b\}, \dots, \{b\}}_{k}, \{c\}, \dots, \{c\})$. Thus, we have $\varphi_{\{a\}}(\bar{R}_N)=\sum_{i\in \{1,\ldots,k\}}\alpha_{ii}$. This shows (\ref{fq_1}) holds for $\{a\}$ at $\bar{R}_N$ when $|\{i\in N\mid \tau(\bar{R}_i)=\{a\}\}|=k$ and $|\{i\in N\mid \tau(\bar{R}_i)=\{b\}\}|=k$.

We now show it for $\{a,b\}$. This proof follows a similar pattern to the proof for Case 1. We make use of both IH and Case 1 here. The assumption of the case implies $\bar{R}_N\equiv (\underbrace{\{a\}, \dots, \{a\}}_{k},\underbrace{\{b\}, \dots, \{b\}}_{k}, \{c\}, \dots, \{c\})$. Consider the profile $(\bar{R}_{2k}',\bar{R}_{-2k})$ where $\tau(\bar{R}'_{2k})=\{d\}$ for some $d\in A\setminus \{a,b,c\}$ as in the previous paragraph. This profile has $k$ many agents with top-ranked set $\{a\}$, $k-1$ many agents with top-ranked set $\{b\}$, and $1$ agent with top-ranked set $\{d\}$. Hence, by Case 1,
\begin{align}
   & \varphi_{\{a,b\}}(\underbrace{\{a\}, \dots, \{a\}}_{k},\underbrace{\{b\}, \dots, \{b\}}_{k-1}, \{d\}, \{c\}, \dots, \{c\})=\sum_{i\in \{1,\ldots,k\} \text{ and } j\in \{k+1,\ldots,2k-1\}}\alpha_{ij}, \nonumber \\
    & \varphi_{\{a,d\}}(\underbrace{\{a\}, \dots, \{a\}}_{k},\underbrace{\{b\}, \dots, \{b\}}_{k-1}, \{d\}, \{c\}, \dots, \{c\})=\sum_{i\in \{1,\ldots,k\} }\alpha_{i(2k)}.  \label{fq_4} 
\end{align}
% and by IH, 
%     \begin{align}
%    & \varphi_{\{b,d\}}(\underbrace{\{a\}, \dots, \{a\}}_{k},\underbrace{\{b\}, \dots, \{b\}}_{k-1}, \{d\}, \{c\}, \dots, \{c\})=\sum_{j\in \{k+1,\ldots,2k-1\}}\alpha_{j(2k)}, \text{ and } \nonumber \\
%    & \varphi_{\{d\}}(\underbrace{\{a\}, \dots, \{a\}}_{k},\underbrace{\{b\}, \dots, \{b\}}_{l}, \{d\}, \{c\}, \dots, \{c\})=\alpha_{(2k)(2k)}. \label{fq_5} 
% \end{align}

Now, as only agent $2k$ is changing her top-ranked set from $\{b\}$ to $\{d\}$ between $\bar{R}_N$ and $(\bar{R}_{2k}',\bar{R}_{-2k})$, using Lemma \ref{de_lem_6}, we have 
\begin{align}
	&\varphi_{\{\{a,b\},\{a,d\}\}}(\bar{R}_N)=\varphi_{\{\{a,b\},\{a,d\}}(\bar{R}_{2k}',\bar{R}_{-2k}). \label{fq_5}
\end{align}
Moreover, as $d\notin \{\tau(\bar{R}_i)\mid i\in N\}$, by Theorem \ref{de_tops-only}, $\varphi_{\{a,d\}}(\bar{R}_N)=0$. This, together with (\ref{fq_4}) and (\ref{fq_5}), implies that 
\begin{align*}
   & \varphi_{\{a,b\}}(\underbrace{\{a\}, \dots, \{a\}}_{k},\underbrace{\{b\}, \dots, \{b\}}_{k}, \{c\}, \dots, \{c\})=\sum_{i\in \{1,\ldots,k\} \text{ and } j\in \{k+1,\ldots,k+l\}}\alpha_{ij}.  \\
\end{align*}
Thus, (\ref{fq_1}) holds for $\{a,b\}$ at $\bar{R}_N$ when $|\{i\in N\mid \tau(\bar{R}_i)=\{a\}\}|=k$ and $|\{i\in N\mid \tau(\bar{R}_i)=\{b\}\}|=k$. 

The two cases together complete the verification of the Induction step. Hence, this completes the proof of the lemma. \end{proof}

As $\varphi$ is unanimous, Lemma \ref{de_lem_7} shows that $\sum_{(p,q)\in N^{(2)}}\alpha_{pq}=1$. Thus, $\sum_{(p,q)\in N^{(2)}}\alpha_{pq}=1$ is valid set of parameters and $\varphi$ is a random bi-dictatorial rule w.r.t these parameters. This completes the proof of Theorem \ref{de_n}.
\end{proof}

 \section{Discussions}\label{sec:discussions}
 \subsection{Connection with \cite{barbera2001strategy} and extreme point property}\label{sec:connections}

The DSCCs on the $\du$ and $\de$ domains were characterized by \cite{barbera2001strategy}. They established that on the $\du$ domain, a DSCC is unanimous and strategy-proof if and only if it is dictatorial. On the $\de$ domain, a DSCC satisfies these properties if and only if it is a bi-dictatorship. These results hold for $m \geq 3$ alternatives. Since random dictatorial rules are convex combinations of dictatorial DSCCs, and random bi-dictatorial rules are convex combinations of bi-dictatorial rules, our results coincide with theirs for the deterministic case on the $\du$ domain for $m \geq 3$ alternatives, and on the $\de$ domain for $m \geq 4$ (or $m=3$ and $n \leq 3$) alternatives. 
 
 However, on the $\de$ domain for $m=3$ and $n \geq 4$, we demonstrate that the class of PSCCs is strictly larger than the set of convex combinations of bi-dictatorial rules, and characterize these rules as coalition-weighted rules. In what follows, we show that the deterministic versions of these rules are bi-dictatorial rules. We call a coalition-weighted rule deterministic if it assigns a degenerate lottery at every profile; equivalently, the corresponding coalition weight function is $\{0,1\}$-valued.
  \begin{proposition}\label{prop_deterministic_cw}
       Every deterministic coalition-weighted rule is a bi-dictatorial rule.
   \end{proposition}
 \begin{proof} Let $\varphi$ be deterministic coalition-weighted rule and let $v$ be the corresponding coalition weight function. This means $v$ must be $\{0,1\}$-valued. Call a coalition $S$ winning if $v(S)=1$. By monotonicity, every superset of a winning coalition is winning. By supermodularity, the intersection of two winning coalitions is also winning. To see this, if $v(S)=v(T)=1$, then $$v(S\cup T)+v(S\cap T)\geq v(S)+v(T)=2.$$ Since both terms on the left are at most one, it follows that $v(S\cup T)=v(S\cap T)=1$. Thus, the family of winning coalitions is closed under intersection. Hence, there is a unique minimal winning coalition, say $D$, and
\[
        v(S)=1 \Longleftrightarrow D\subseteq S.
\]

 We show that $D$ must contain either a single agent or two agents. To see this, note that the partition balance condition rules out $|D|\geq 3$. If $|D|\geq 3$, choose a partition $(S,T,U)$ of $N$ that splits the members of $D$ across all three parts. Then none of $S, T, U, S\cup T, S\cup U, T\cup U$ contains $D$. Hence, all these coalitions have weight zero, and we get a contradiction to the partition balance condition. Normalization rules out $D=\emptyset$. Therefore, $D$ consists of either one agent or two agents.

If $D=\{i\}$, then $v(S)=1$ if and only if $i\in S$. The induced rule selects $\tau(R_i)$ at every profile, which is a dictatorship. If $D=\{i,j\}$ with $i\neq j$, then $v(S)=1$ if and only if $\{i,j\}\subseteq S$. At any profile, if agents $i$ and $j$ have the same top-ranked singleton, say $\{x\}$, then $N_x(R_N)$ contains $D$, and the rule selects $\{x\}$. If they have different top-ranked singletons, say $\{x\}$ and $\{y\}$, then neither $N_x(R_N)$ nor $N_y(R_N)$ contains $D$, but $N_x(R_N)\cup N_y(R_N)$ does. Hence
\[
        \varphi_{\{x,y\}}(R_N)=1.
\]
Thus, in all cases, the rule selects $\tau(R_i)\cup \tau(R_j)$. This is exactly a bi-dictatorial rule. Therefore, for $m=3$, deterministic coalition-weighted rules are bi-dictatorial rules.
 \end{proof}
 
 In the existing literature, very few domains have been identified where probabilistic rules are not merely convex combinations of deterministic ones (see \cite{peters2021unanimous} and \cite{chatterji2022probabilistic}). While those known cases typically persist for three or more agents regardless of the number of alternatives, our findings reveal that the $\de$ domain is uniquely sensitive to the specific value of $m$. Given this anomaly for $m=3$ and $n \geq 4$, characterizing the extreme points of the class of unanimous and strategy-proof PSCCs on the $\de$ domain remains an intriguing open question.

\subsection{Coalition-weighted rules for more than 3 alternatives}\label{sec:rcw_mgeq4}
 In this section, we define a coalition-weighted rule for $m\geq 4$ alternatives. This is in the same spirit as the coalition-weighted rules defined for $m=3$ alternatives. We show that for $m\geq 4$, every coalition-weighted rule becomes a random bi-dictatorial rule. To proceed, consider the following definitions.

\begin{definition}\label{defn_cwf}
A function $v:2^N\to[0,1]$ is a \emph{coalition weight function} if it satisfies the following conditions.
\begin{enumerate}
    \item [(i)] \textbf{Normalization:} $v(\emptyset)=0,\quad v(N)=1.$

    \item [(ii)] \textbf{Monotonicity:} for all $S,T\subseteq N$,
    $S\subseteq T \implies v(S)\leq v(T).$

    \item [(iii)] \textbf{Supermodularity:} for all $S,T\subseteq N$, $v(S\cup T)+v(S\cap T)\geq v(S)+v(T).$

    \item [(iv)] \textbf{Partition Balance:} for every partition $(S_1,\ldots,S_m)$ of $N$,
    $$\sum_{k<l} v(S_k\cup S_l)=1+(m-2)\sum_{p=1}^m v(S_p).$$
\end{enumerate}
\end{definition}

Observe that the formulation of the coalition weight function $v$ in Definition \ref{defn_cwf} is identical to its counterpart in Definition \ref{defn_tbc}, with the sole exception of the fourth condition, Partition Balance. Structurally, the partition balance condition serves as the restriction ensuring that only singletons and doubletons are assigned positive probabilities at any given preference profile. Furthermore, setting $m=3$ in the generalized partition balance condition of Definition \ref{defn_cwf} directly collapses back to the baseline partition balance requirement of Definition \ref{defn_tbc}. This consistency across dimensions allows us to now formally define the class of coalition-weighted rules.

% Observe that this definition of a coalition weight function $v$ is identical to the coalition weight function defined in Definition \ref{defn_tbc}, except for the fourth condition of Partition Balance. Recall that the partition balance condition is to ensure that only singletons and doubletons receive positive probabilities at any profile. Further, if we put $m=3$ in the partition balance condition of Definition \ref{defn_cwf}, we get back the partition balance condition of Definition \ref{defn_tbc}. We are now ready to formally define random coalition-weighted rules.

\begin{definition}
    A PSCC $\varphi:\de^n \to \Delta \mathcal{A}$ is called a coalition-weighted rule if there exists a coalition weight function $v$ such that for every profile $R_N\in \de^n$ and every $x\in A$,
$$\varphi_{\{x\}}(R_N)=v(N_x(R_N)),$$
for every distinct $x,y\in A$,
$$\varphi_{\{x,y\}}(R_N)=v(N_x(R_N)\cup N_y(R_N))-v(N_x(R_N))-v(N_y(R_N)), \text{ and }$$
$$\varphi_{X}(R_N)=0 \text{ for all } X \text{ with }|X|\geq 3.$$
\end{definition}

% We now formally show that for $m\geq 4$, every random coalition-weighted rule is a random bi-dictatorial rule. This, together with Theorem \ref{de_n}, yields another representation of unanimous and strategy-proof PSCCs for $m\geq 4$. A PSCC is unanimous and strategy-proof if and only if it is a random coalition-weighted rule.  

% We now formally demonstrate that for $m \ge 4$, the class of random coalition-weighted rules collapses uniquely into the class of random bi-dictatorial rules. Combined with Theorem \ref{de_n}, this equivalence provides an alternative structural representation of all unanimous and strategy-proof PSCCs on larger choice sets. Specifically, we show that a PSCC is unanimous and strategy-proof if and only if it can be represented as a random coalition-weighted rule.

We now formally show that for $m \ge 4$, the class of coalition-weighted rules collapses to the class of random bi-dictatorial rules. Combined with Theorem \ref{de_n}, this equivalence yields an alternative representation of unanimous and strategy-proof PSCCs. Specifically, a PSCC is unanimous and strategy-proof if and only if it is a coalition-weighted rule.
\begin{theorem}\label{thm_mgeq4}
    Let $m\geq 4$. Then every coalition-weighted rule is a random bi-dictatorial rule. 
\end{theorem}
\begin{proof}
    Let $\varphi$ be a coalition-weighted rule with respect to the coalition weight function $v$. We show that $v$ is pair-additive, i.e., there exists $\{\alpha_{ij}\}_{(i,j)\in N^{(2)}}$ such that for all $S\subseteq N$, 
    \begin{equation}\label{rcw_0}
        v(S) =  \sum_{\substack{\{i, j\} \subseteq S \\ (i,j)\in N^{(2)}}} \alpha_{ij}.
    \end{equation}
    Define $\{\alpha_{ij}\}_{(i,j)\in N^{(2)}}$ as follows:
    \begin{align*}
         \alpha_{ii}:=v(\{i\})  \hspace{20mm}   \alpha_{ij}:=v(\{i,j\})-v(\{i\})-v(\{j\}).
    \end{align*}
 For an agent $i\in N$ and $W\subseteq N\setminus \{i\}$, let $\Delta_iv(W)$ denote the first order marginal contribution of $i$ for the coalition $W$. Formally, $\Delta_iv(W)=v(W\cup \{i\})-v(W)$. We first prove a claim. 
   \begin{claim}\label{cl_rcw_1}
       For all distinct $i,j\in N$ and $W\subseteq N\setminus \{i,j\}$, $$\Delta_i v(W \cup \{j\}) - \Delta_i v(W) = \alpha_{ij}.$$
   \end{claim}
   \textbf{Proof of the claim:}
   Let $(X_1, X_2, X_3, X_4)$ be a partition of $N$. Consider a partition $(S_1, \ldots, S_m)$ of $N$ such that $S_l=X_l$ for $l\leq 4$ and $S_l=\emptyset$ for all $4<l\leq m$. Applying the partition balance condition for $(S_1, \ldots, S_m)$, we have
   \begin{align}
   & \sum_{1 \leq k < l \leq 4} v(X_k \cup X_l) + (m - 4) \sum_{p=1}^{4} v(X_p) = 1 + (m - 2) \sum_{p=1}^{4} v(X_p) \nonumber \\
  \implies & \sum_{1 \le k < l \le 4} v(X_k \cup X_l) = 1 + 2 \sum_{p=1}^{4} v(X_p).\label{rcw_1}
   \end{align}
  Let $k \in N$ be a fixed agent, and consider a partition $(A,B,\{k\},D)$ of $N$. Applying (\ref{rcw_1}) to $(X_1,X_2,X_3,X_4)=(A,B,\{k\},D)$, we have 
  \begin{equation}\label{rcw_2}
v(A \cup B) + v(A \cup \{k\}) + v(A \cup D) + v(B \cup \{k\}) + v(B \cup D) + v(D \cup \{k\})= 1 + 2[v(A) + v(B) + v(\{k\}) + v(D)]
\end{equation}
Similarly, applying (\ref{rcw_1}) to $(X_1,X_2,X_3,X_4)=(A\cup \{k\},B,\emptyset,D)$, we have 
\begin{equation*}
v(A \cup B \cup \{k\}) + v(A \cup \{k\}) + v(A \cup D \cup \{k\}) + v(B) + v(B \cup D) + v(D) = 1 + 2[v(A \cup \{k\}) + v(B) + v(D)].
\end{equation*}
Simplifying the above by canceling $v(A\cup \{k\})$ from both the sides we have 
\begin{equation}\label{rcw_3}
v(A \cup B \cup \{k\}) + v(A \cup D \cup \{k\}) + v(B) + v(B \cup D) + v(D)
= 1 + v(A \cup \{k\}) + 2v(B) + 2v(D)
\end{equation}
Now subtracting (\ref{rcw_2}) from (\ref{rcw_3}), we have
\begin{align*}
& [v(A \cup B \cup \{k\}) - v(A \cup B)] + [v(A \cup D \cup \{k\}) - v(A \cup D)] - v(A \cup \{k\}) - v(B \cup \{k\}) - v(D \cup \{k\}) + v(B) + v(D)  \\
&= v(A \cup \{k\}) - 2v(A) - 2v(\{k\})
\end{align*}
Further, substituting the $\Delta_k$ in the above equation and arranging the terms, it yields
\begin{align}\label{rcw_4}
&[\Delta_k v(A \cup B) - \Delta_k v(A) - \Delta_k v(B) + v(\{k\})] + [\Delta_k v(A \cup D) - \Delta_k v(A) - \Delta_k v(D) + v(\{k\})] = 0
\end{align}
Define the function $h_k(X, Y) = \Delta_k v(X \cup Y) - \Delta_k v(X) - \Delta_k v(Y) + v(\{k\})$. By definition, this function is symmetric in its arguments: $h_k(X, Y) = h_k(Y, X)$. Thus, (\ref{rcw_4}) simplifies to:
\begin{equation}\label{rcw_5}
h_k(A, B) + h_k(A, D) = 0
\end{equation}
Similar to (\ref{rcw_5}), if we merge $k$ into either $B$ or $D$ and use the symmetry of $h_k$, we get
\begin{align}\label{rcw_6}
   & h_k(A, B) + h_k(B, D) = 0 \text{ and }\nonumber \\
   & h_k(A, D) + h_k(B, D) = 0.
\end{align}
Note that (\ref{rcw_5}) and (\ref{rcw_6}), together imply that $h_k(A,B)=0$. Therefore, by the definition of $h_k$,
\begin{equation}\label{rcw_7}
    \Delta_k v(A \cup B) - \Delta_k v(A) = \Delta_k v(B) - v(\{k\}).
\end{equation}
We now complete the proof of the claim. Consider distinct $i,j\in N$ and $W\subseteq N\setminus \{i,j\}$. Putting $A=W$, $B=\{j\}$ and $k=i$ in (\ref{rcw_7}), we have
\begin{align*}
    \Delta_i v(W \cup \{j\}) - \Delta_i v(W) = &\Delta_i v(\{j\}) - v(\{i\})\nonumber\\
 =  &\; v(\{i, j\}) - v(\{j\}) - v(\{i\}) \hspace{20mm} (\text{as } \Delta_i v(\{j\}) = v(\{i, j\}) - v(\{j\}))\\
 = &\; \alpha_{ij}. \hspace{128mm}  \square
\end{align*}

We now prove another claim.
\begin{claim}\label{cl_rcw_2}
    For all $X\subseteq N$ and all $i\notin N$,
    $$\Delta_i v(X) = v(\{i\}) + \sum_{j \in X} \alpha_{ij}$$
\end{claim}
\textbf{Proof of the claim:} Let $X \subseteq N \setminus \{i\}$ be a subset of agents, and let its elements be ordered arbitrarily as $X = \{j_1, j_2, \dots, j_k\}$. We can express the marginal contribution of agent $i$ to $X$, $\Delta_i v(X)$, as a telescoping sum starting from the empty set: $$\Delta_i v(X) = \Delta_i v(\emptyset) + \sum_{r=1}^k \left[ \Delta_i v(\{j_1, \dots, j_r\}) - \Delta_i v(\{j_1, \dots, j_{r-1}\}) \right].$$
By applying Claim \ref{cl_rcw_1}, the $r^{th}$  difference term in the sum is equal to $\alpha_{i j_r}$. Additionally, by the definition of the marginal operator and normalization ($v(\emptyset)=0$), we have $\Delta_i v(\emptyset) = v(\{i\}) - v(\emptyset) = v(\{i\})$. Substituting these into the telescoping sum yields: $$\hspace{50mm} \Delta_i v(X) = v(\{i\}) + \sum_{r=1}^k \alpha_{i j_r} = v(\{i\}) + \sum_{j \in X} \alpha_{ij}. \hspace{55mm} \square $$

We are now ready to show that $v$ is pair-additive, i.e., (\ref{rcw_0}) holds. Consider $S= \{i_1, i_2, \dots, i_p\} \subseteq N$ with $i_1<\cdots<i_p$. We can write $v(S)$ as a telescoping sum of marginal contributions, adding one agent at a time: $$v(S) = v(\emptyset) + \sum_{k=1}^p \Delta_{i_k} v(\{i_1, \dots, i_{k-1}\}).$$ Further, using $v(\emptyset) = 0$ and Claim \ref{cl_rcw_2} for each marginal contribution in the above sum yield: \begin{align*}
    v(S) =& \sum_{k=1}^p \left( v(\{i_k\}) + \sum_{r=1}^{k-1} \alpha_{i_r i_k} \right)\\
    =&  \sum_{k=1}^p v(\{i_k\}) + \sum_{k=1}^p \sum_{r=1}^{k-1} \alpha_{i_r i_k} \\
    =&  \sum_{k=1}^p \alpha_{i_ki_k} + \sum_{k=1}^p \sum_{r=1}^{k-1} \alpha_{i_r i_k} \\
     =& \sum_{\substack{\{i, j\} \subseteq S \\ (i,j)\in N^{(2)}}} \alpha_{ij}.
    \end{align*}
    Thus, $v$ is pair-additive. We now show that $\varphi$ is random bi-dictatorial with respect to the parameters $\{\alpha_{ij}\}_{(i,j)\in N^{(2)}}$. Note that as $v(N)=1$, by pair-additivity of $v$, $\sum_{(i,j)\in N^{(2)}}\alpha_{ij}=1$, implying $\{\alpha_{ij}\}_{(i,j)\in N^{(2)}}$ is a valid set of parameters for a random bi-dictatorial rule. Consider a profile $R_N\in \de^n$ and a singleton set $\{a\}$. By the definition of $\varphi$ 
    \begin{align*}
        \varphi_{\{a\}}(R_N)=& v(N_a(R_N))\\
        =& \sum_{\substack{\{i, j\} \subseteq N_a(R_N) \\ (i,j)\in N^{(2)}}} \alpha_{ij}\hspace{10mm} (\text{by pair additivity of $v$})\\
        =& \sum_{\substack{(i,j)\in N^{(2)} \\ \tau(R_i)\cup \tau(R_j) =\{a\}}}\alpha_{ij}.
    \end{align*}
    Now consider distinct $a,b\in A$. Again, by the definition of $\varphi$
\begin{align*}
        \varphi_{\{a,b\}}(R_N)=& v(N_a(R_N)\cup N_n(R_N))-v(N_a(R_N))-v(N_b(R_N))\\
        =& \sum_{\substack{\{i, j\} \subseteq N_a(R_N)\cup N_b(R_N) \\ (i,j)\in N^{(2)}}} \alpha_{ij}-\sum_{\substack{\{i, j\} \subseteq N_a(R_N) \\ (i,j)\in N^{(2)}}} \alpha_{ij}-\sum_{\substack{\{i, j\} \subseteq N_b(R_N) \\ (i,j)\in N^{(2)}}} \alpha_{ij}\hspace{10mm} (\text{by pair additivity of $v$})\\
        =& \sum_{\substack{(i,j)\in N^{(2)} \\ \tau(R_i)\cup \tau(R_j) =\{a,b\}}}\alpha_{ij}.
    \end{align*}
    As $\varphi_X(R_N)=0$ for all $|X|\geq 3$, this means for all $X\in \mathcal{A}$, $\varphi_X(R_N)=\sum_{\substack{(i,j)\in N^{(2)} \\ \tau(R_i)\cup \tau(R_j) =X}}\alpha_{ij}$. Thus, $\varphi$ is a random bi-dictatorial rule with respect to the parameters $\{\alpha_{ij}\}_{(i,j)\in N^{(2)}}$. 
\end{proof}

\begin{appendix}

\section{Existence of preferences used in Lemma  \ref{lem_1} and Lemma \ref{du_lem_3}}\label{appen_2}

\begin{lemma}\label{du_lem_pref_exis}
    For all distinct $x,y,z\in A$, the following preferences exist in the $\du$ domain.
    \begin{enumerate}[(i)]
    \item $R_1\equiv \{x\}\{x,y\}\{y\}\cdots$ 
        \item $R_2\equiv \{x\}\cdots \{x,y\}\{y\}$
        \item $R_3\equiv \{x\}\{x,y\}\cdots \{y\}$
        \item $R_4\equiv \{x\}\{x,y\}\cdots \{x,z\} \{z\}$
        %\item $R_5\equiv \{y\}\{x,y\}\{x\}\cdots \{y,z\} \{z\}$
        \item $R_5\equiv \{x\}\{x,y\}\{x,z\}\cdots $
        \item $R_6\equiv \{x\}\{x,y\}\cdots \{z\} \{y,z\} \{y\}$
       \item  $R_7\equiv \{x\}\{x,y\}\cdots \{y\} \{y,z\} \{z\}$
    \end{enumerate}
    
\end{lemma}
\begin{proof} Fix distinct $x,y,z\in A$. Let $A\setminus\{x,y,z\}=\{b_1,\dots,b_{m-3}\}$. It is enough to choose positive real numbers $\mu(a)$ for each $a\in A$ and a utility function $v$.\footnote{Note that we can normalize $\mu(a)$ for each $a \in A$ and get hold of an assesment $\lambda$ by setting $\lambda(a):=\frac{\mu(a)}{\sum_{c\in A}\mu(c)}$.}

\begin{enumerate}[(i)]
    \item The existence of $R_1$ is straightforward, hence, we omit it. 

   \item  Choose $\mu(x)=\frac{1}{10}$ and  $\mu(a)=1$ for all $a \in A \setminus \{x\}$. Define the utility function $v$ such that
\begin{itemize}
    \item $v(x)=4$,
    \item $v(y)=0$,
    \item $v(z)=2$, and 
    \item $v(b_r)=1+\frac{r}{m}$ for all $r \in \{1,\ldots,m-3\}$.
\end{itemize}

Thus, $1<v(b_r)<2$ for all $r \in \{1,\ldots,m-3\}$. Now, we compute the following:
$$
v(x)=4,\qquad
v(\{x,y\})=\frac{\frac{1}{10}\cdot 4+0}{\frac{1}{10}+1}=\frac{4}{11},
\qquad
v(y)=0.
$$

First, we show that $v(x)>v(B)$ for all $B \in \mathcal{A} \setminus \{\{x\}\}$. Consider any $B \in \mathcal{A} \setminus \{\{x\}\}$. This implies that $B \cap \{y,z,b_1,b_2,\ldots,b_{m-3}\} \neq \emptyset$. Let $a_0 \in B \cap \{y,z,b_1,b_2,\ldots,b_{m-3}\}$. Since $v(a)\leq 4$ for all $a \in B$ and $v(a_0)<4$, it follows that $v(B)<4=v(x)$.

Next, we show that $v(B)>v(\{x,y\})>v(y)$ for all $B \in \mathcal{A} \setminus \{\{y\},\{x,y\}\}$. Since $v(\{x,y\})=\dfrac{4}{11}$ and $v(y)=0$, it follows that $v(\{x,y\})>v(y)$. Consider any $B \in \mathcal{A} \setminus \{\{y\},\{x,y\}\}$. We will show that $v(B)>v(\{x,y\})$. 

If $y\notin B$, then $v(a)>1$ for all $a \in B$. Therefore, $v(B)>1>\frac{4}{11}=v(\{x,y\})>0=v(y)$.

If $y\in B$ and $x\notin B$, set $S:=B\setminus\{y\}$.
Then $S\neq\varnothing$ and every $a\in S$ satisfies $v(a)\geq v(b_1)=1+\frac{1}{m}$, and $\mu(a)=1$ for all $a\in S$. Hence,
$$
v(B)\geq
\frac{|S|\left(1+\frac{1}{m}\right)}{1+|S|}
\geq
\frac{1+\frac{1}{m}}{2}
>
\frac12
>
\frac{4}{11}
=
v(\{x,y\}).
$$

Finally, suppose $\{x,y\}\subsetneq B$. Set $T:=B\setminus\{x,y\}$.
Then $T\neq\varnothing$ and every $a\in T$ satisfies $v(a)>\frac{4}{11}=v(\{x,y\})$. Therefore,
$$
v(B)
=
\frac{\bigl(\mu(x)+\mu(y)\bigr)v(\{x,y\})+\sum_{a\in T}\mu(a)v(a)}
{\mu(x)+\mu(y)+\sum_{a\in T}\mu(a)}
>
v(\{x,y\}).
$$

Hence, $v(B)>v(\{x,y\}$. This completes the proof. 
    
 %    If $y\in B$ and $x\notin B$, assume $|B|=k\geq 2$ and we have,
 %    \begin{align*}
 %        v(B)&=\frac{\sum_{w\in B} v(w)\lambda(w)}{\sum_{w\in B}\lambda(w)}\\
 %        &\geq \frac{m\sum_{w\in B\setminus y}\lambda(w)}{\sum_{w\in B}\lambda(w)}\\
 %        &= m - \frac{m\lambda(y)}{\sum_{w\in B}\lambda(w)}\\
 %          &\geq  m - \frac{m\lambda(y)}{(k-1)m \lambda(y)+\lambda(y)}\\
 %          &=  m - \frac{m}{(k-1)m +1}\\
 %           &\geq   m - \frac{m}{m +1}> m-1.\\
 %    \end{align*}
 % If $\{x,y\} \subseteq B$, we must have $|B|>2$. This gives us
 % \begin{align*}
 %        v(B)&=\frac{\sum_{w\in B} v(w)\lambda(w)}{\sum_{w\in B}\lambda(w)}\\
 %        &\geq \frac{m\sum_{w\in B\setminus y}\lambda(w)}{\sum_{w\in B}\lambda(w)}\\
 %        &= m - \frac{m\lambda(y)}{\sum_{w\in B}\lambda(w)}\\
 %          &\geq  m - \frac{m\lambda(y)}{(k-2)m \lambda(y)+\lambda(y)+\tfrac{\lambda(y)}{m}}\\
 %          &=  m - \frac{m}{(k-2)m +1+\tfrac{1}{m}}\\
 %           &\geq   m - \frac{m}{m +1}> m-1.\\
 %    \end{align*}
   
 \item Consider the $v$ and $\mu$ in (ii) and interchange the role of $x$ and $y$ to form $v_3$ and $\mu_3$. The pair $(v_3,\mu)$ will give us a preference $R\equiv \{y\}\cdots \{x,y\}\{x\}$. Now take $(-v_3,\mu_3)$ and consider the corresponding preference. Clearly, it will have the structure $R_3\equiv \{x\}\{x,y\}\cdots \{y\}$. This completes the proof.
 
%  We deduce this from part (ii). Choose $\mu(y)=\frac{1}{10}$ and  $\mu(a)=1$ for all $a \in A \setminus \{y\}$. Define the utility function $\tilde v$ such that
% \begin{itemize}
%     \item $\tilde v(y)=4$,
%     \item $\tilde v(x)=0$,
%     \item $\tilde v(z)=2$, and 
%     \item $\tilde v(b_r)=1+\frac{r}{m}$ for all $r \in \{1,\ldots,m-3\}$.
% \end{itemize}

% Then, by the proof of Part (ii), it follows that there exists a CEUC ordering $\tilde R\in\mathcal D_U$ such that
% $$
% \tilde R \equiv\{y\}\cdots \{x,y\}\{x\}.
% $$

% Define a new utility function $v$ such that $v(a):=- \tilde v(a)$ for all $a\in A$ and consider the same $\mu$. Then for every nonempty $B\subseteq A$,
% $$
% v(B)
% =
% \frac{\sum_{a\in B}\mu(a)v(a)}{\sum_{a\in B}\mu(a)}
% =
% \frac{\sum_{a\in B}\mu(a)(- \tilde v(a))}{\sum_{a\in B}\mu(a)}
% =
% - \tilde v(B).
% $$

% Hence, the CEUC ordering $R$ induced by $(v,\mu)$ \uk{ei ``induced by" define korte hobe} satisfies
% $$
% XR Y
% \iff v(X)\geq v(Y)
% \iff - \tilde v(X)\geq - \tilde v(Y)
% \iff \tilde v(X)\leq \tilde v(Y) \iff Y \tilde R X.
% $$
% Therefore, the ordering $R$ is exactly the reverse of $\tilde R$ \uk{reverse ta define korte hobe}. This, together with the fact that $\tilde R \equiv \{y\}\cdots \{x,y\}\{x\}$ implies $R \equiv \{x\}\{x,y\}\cdots \{y\}$.
% 

\item Choose $\mu(x)=\frac{1}{10}$ and $\mu(a)=1$ for all $a \in A \setminus  \{x\}$. Define the utility function $v$ in the following manner:
\begin{itemize}
    \item $v(x)=4$,
    \item $v(y)=3$,
    \item $v(z)=0$, and
    \item $v(b_r)=1+\frac{r}{m}$ for all $r \in \{1,\dots,m-3\}$.
\end{itemize}

Thus $1<v(b_r)<2$ for all $r \in \{1,\dots,m-3\}$. Now, we compute the following:
\[
v(x)=4,\qquad
v(\{x,y\})=\frac{\frac{1}{10}\cdot 4+3}{\frac{1}{10}+1}=\frac{34}{11},
\]
\[
v(\{x,z\})=\frac{\frac{1}{10}\cdot 4+0}{\frac{1}{10}+1}=\frac{4}{11},
\qquad
v(z)=0.
\]

First, we show that $v(x)>v(\{x,y\})>v(B)$ for all $B \in \mathcal{A} \setminus \{\{x\}, \{x,y\}\}$. As computed above, it follows that $v(x)>v(\{x,y\})$. Now, we proceed to show that $v(\{x,y\})>v(B)$ for all $B \in \mathcal{A} \setminus \{\{x\}, \{x,y\}\}$. Let $B \in \mathcal{A} \setminus \{\{x\}, \{x,y\}\}$. If $x\notin B$, then $v(\{a\})\leq 3$ for all $a \in B$. Hence, $v(B)\leq 3<\frac{34}{11}=v(\{x,y\})$.

Now, suppose $x\in B$ and $B\neq\{x\}$. Define $S:=B\setminus\{x\}$. Then $S\neq\varnothing$, and since $\mu(a)=1$ for all $a\in S$,
$
\mu(S):=\sum_{a\in S}\mu(a)=|S|\geq 1$.
Also $v(a)\leq 3$ for every $a\in S$. Therefore,
$v(B)\leq \frac{\frac{1}{10}\cdot 4+3\mu(S)}{\frac{1}{10}+\mu(S)}$.
Consider the function $f: [1, \infty) \to \mathbb{R}$ such that $f(u):=\frac{\frac{2}{5}+3u}{\frac{1}{10}+u}$ for all $u\geq 1$.
Then
\[
f'(u)
=
\frac{3\left(\frac{1}{10}+u\right)-\left(\frac{2}{5}+3u\right)}
{\left(\frac{1}{10}+u\right)^2}
=
-\frac{1}{10}\cdot \frac{1}{\left(\frac{1}{10}+u\right)^2}
<0.
\]
Hence, $f$ is strictly decreasing on $[1,\infty)$. Therefore, $v(B)\leq f(\mu(S))\leq f(1)=\frac{34}{11}=v(\{x,y\})$.
Equality holds only when $\mu(S)=1$ and every element of $S$ has utility $3$, i.e.\ only when $S=\{y\}$. Therefore $v(B)< v(\{x,y\})$.

Next, we show that $v(B)>v(\{x,z\})>v(\{z\})$ for all $B \in \mathcal{A} \setminus \{\{z\}, \{x,z\}\}$. Since $v(\{x,z\})=\dfrac{4}{11}$ and $v(z)=0$, it follows that $v(\{x,z\})>v(z)$. Now, we proceed to show that $v(B)>v(\{x,z\})$ for all $B \in \mathcal{A} \setminus \{\{z\}, \{x,z\}\}$. Consider any $B \in \mathcal{A} \setminus \{\{z\}, \{x,z\}\}$. If $z\notin B$, then $v(\{a\}>1$ for all $a \in B$. Therefore, $v(B)>1>\frac{4}{11}=v(\{x,z\})$.

If $z\in B$ and $x\notin B$, set $T:=B\setminus\{z\}$.
Then $T\neq\varnothing$, every $a\in T$ satisfies $v(a)\geq v(b_1)=1+\frac{1}{m}$, and $\mu(a)=1$ for all $a\in T$. Hence,
\[
v(B)\geq \frac{|T|\left(1+\frac{1}{m}\right)}{1+|T|}\geq
\frac{1+\frac{1}{m}}{2}>\frac12>\frac{4}{11}.
\]

Finally, suppose $\{x,z\}\subsetneq B$. Set $T:=B\setminus\{x,z\}$.
Then $T\neq\varnothing$, and every $a\in T$ satisfies $v(a)>\frac{4}{11}=v(\{x,z\})$.
Therefore,
\[
v(B)
=
\frac{\bigl(\mu(x)+\mu(z)\bigr)v(\{x,z\})+\sum_{a\in T}\mu(a)v(a)}
{\mu(x)+\mu(z)+\sum_{a\in T}\mu(a)}
>
v(\{x,z\}).
\]

Hence, $v(B)>v(\{x,z\})$. This completes the proof.

   \item Choose $\mu(a)=1$ for all $a\in A$. Define the utiliity function $v$ such that
\begin{itemize}
    \item $v(x)=5$,
    \item $v(y)=2$,
    \item $v(z)=1$, and
    \item $v(b_r)=\frac{r}{m}$ for all $r \in \{1,\dots,m-3\}$.
\end{itemize}

Thus, $0<v(b_r)<1$ for all $r \in \{1,\dots,m-3\}$. Now, we compute the following:
$$
v(x)=5,\qquad
v(\{x,y\})=\frac{5+2}{2}=\frac72,\qquad
v(\{x,z\})=\frac{5+1}{2}=3.
$$

Therefore, we have $v(x)>v(\{x,y\})>v(\{x,z\})$. We show that $v(\{x,z\})>v(B)$ for all $B \in \mathcal{A} \setminus \{\{x\},\{x,y\},\{x,z\}\}$. Consider any $B \in \mathcal{A} \setminus \{\{x\},\{x,y\},\{x,z\}\}$. If $x\notin B$, then $v(a)\leq 2$ for every $a \in B$. Therefore, $v(B)\leq 2<3=v(\{x,z\})$.

Now, suppose $x\in B$ and set $S:=B\setminus\{x\}$. Then, $S\neq \varnothing$ because $B \in \mathcal{A} \setminus \{\{x\},\{x,y\},\{x,z\}\}$. If $|S|=1$, then $S\in \{\{y\}, \{z\}, \{b_r\}\}$ for some $r \in \{1,\ldots,m-3\}$. Since $x \in B$ and $B \in \mathcal{A} \setminus \{\{x\},\{x,y\},\{x,z\}\}$, it must be the case that $S=\{b_r\}$. This, together with the fact that $0<v(b_r)<1$ implies that $v(B)= \dfrac{5+v(b_r)}{2}<3=v(\{x,z\})$.

If $|S|\geq 2$, then among the elements of $S$,
at most one has utility $2$ (namely $y$), at most one has utility $1$ (namely $z$), and all the remaining elements have utility $<1$. Therefore, $\sum_{a\in S}v(a)\leq 2+1+\bigl(|S|-2\bigr)\cdot 1=|S|+1$.
Hence,
\[
v(B)
=
\frac{5+\sum_{a\in S}v(a)}{1+|S|}
\leq
\frac{5+(|S|+1)}{1+|S|}
=
1+\frac{5}{|S|+1}.
\]
This, together with the fact that $|S|\geq 2$ implies $v(B)\leq 1+\frac{5}{|S|+1}\leq 1+\frac{5}{3}<3=v(\{x,z\})$. Therefore, $v(\{x,z\})>v(B)$ for all $B \in \mathcal{A} \setminus \{\{x\},\{x,y\},\{x,z\}\}$. This completes the proof.
    \item Choose $\mu(y)=\frac14$ and $\mu(a)=1$ for all $a \in A \setminus  \{y\}$. Define the utility function $v$ such that
\begin{itemize}
    \item $v(x)=10$,
    \item $v(y)=0$,
    \item $v(z)=1$, and 
    \item $v(b_r)=3+\frac{r}{m}$ for all $r \in \{1,\dots,m-3\}$.
\end{itemize}

Thus $3<v(b_r)<4$ for all $r \in \{1,\dots,m-3\}$. Next, we compute the following:
$$
v(x)=10,\qquad
v(\{x,y\})=\frac{10+\frac14\cdot 0}{1+\frac14}=8,
$$

$$
v(z)=1,\qquad
v(\{y,z\})=\frac{\frac14\cdot 0+1}{\frac14+1}=\frac45,
\qquad
v(y)=0.
$$

First, we show that $v(x)>v(\{x,y\})>v(B)$ for all $B \in \mathcal{A} \setminus \{\{x\},\{x,y\}\}$.  From the above computation, it follows that $v(x)>v(\{x,y\})$. Consider any $B \in \mathcal{A} \setminus \{\{x\},\{x,y\}\}$. We will show that $v(\{x,y\})>v(B)$. If $x\notin B$, then $v(a)\leq 4$ for every $a \in B$. Therefore, $v(B)\leq 4<8=v(\{x,y\})$. 

Now, suppose $x\in B$. Set $S:=X\setminus\{x\}$.
Then, it must be the case that $S \cap \{z,b_1,b_2,\ldots,b_{m-3}\}\neq\varnothing$ because $B \in \mathcal{A} \setminus \{\{x\},\{x,y\}\}$. Hence, $\mu(S):=\sum_{a\in S}\mu(a)\geq 1$. Moreover, $v(a)\leq 4$ for each $a \in S$. Therefore,
$v(B)\leq \frac{10+4\mu(S)}{1+\mu(S)}$. Consider the function $g: [1,\infty) \to \mathbb{R}$ defined by $g(u):=\frac{10+4u}{1+u}$ for all $u\geq 1$. Then,
$$
g'(u)=\frac{4(1+u)-(10+4u)}{(1+u)^2}
=
-\frac{6}{(1+u)^2}<0.
$$

Hence, $g$ is strictly decreasing on $[1,\infty)$, and thus
$
v(B)\leq g(\mu(S))\leq g(1)=7<8=v(\{x,y\})
$.
Therefore,$v(\{x,y\})>v(B)$.

Next, we show that $v(B)>v(z)>v(\{y,z\})>v(y)$ for all $B \in \mathcal{A} \setminus \{\{z\}, \{y,z\}, \{y\}\}$. We already know from our computation (before) that $v(z)=1> v(\{y,z\})=\frac45> v(y)=0$. Consider any $B \in \mathcal{A} \setminus \{\{z\}, \{y,z\}, \{y\}\}$. If $y\notin B$, then it must be the case that $B \cap \{x,b_1,b_2,\ldots,b_{m-3}\} \neq \emptyset$. This, together with the fact that $v(a)\geq 1$ for all $a \in B$ (where equality holds if and only if $a=z$) implies that $v(B)>1=v(\{z\})$.

Now, suppose $y\in B$. Set $S:=X\setminus\{y\}$. Since $B \in \mathcal{A} \setminus \{\{z\}, \{y,z\}, \{y\}\}$, it must be the case that $B \cap \{x,b_1,b_2,\ldots,b_{m-3}\} \neq \emptyset$. Let $a_0 \in B \cap \{x,b_1,b_2,\ldots,b_{m-3}\}$. Then, by the choice of $v$, it follows that $v(a_0)\geq 3+ \frac{1}{m}$. Also, $v(a)\geq 1$ for all $a \in S$. Therefore, it must be the case that
$$
\sum_{a\in S}v(a)\geq \sum_{a\in S \setminus \{a_0\}}v(a) + v(a_0) \geq \left(|S|-1\right) + \left(3+ \frac{1}{m}\right)= |S|+\left(2+\frac{1}{m}\right)
$$
Also, notice that $\mu(a)=1$ for all $a \in S$. These, together with the fact that $v(y)=0$ implies that
$$
v(B)
=
\frac{\sum_{a\in S}v(a)}{\frac14+|S|}
\geq
\frac{|S|+2+\frac{1}{m}}{\frac14+|S|}
>1=v(z).
$$
The strict inequality follows from the fact that  $2+\frac{1}{m}>\frac14$ for any natural number $m$. Therefore, $v(B)>v(z)$. This completes the proof.

     \item We choose the same $\mu$ and a similar $v$ (just interchange the values of $y$ and $z$) as in the previous part (vi). Choose $\mu(y)=\frac14$ and $\mu(a)=1$ for all $a \in A \setminus  \{y\}$. Define the utility function $v$ such that
\begin{itemize}
    \item $v(x)=10$,
    \item $v(y)=1$,
    \item $v(z)=0$, and 
    \item $v(b_r)=3+\frac{r}{m}$ for all $r \in \{1,\dots,m-3\}$.
\end{itemize}

Thus $3<v(b_r)<4$ for all $r \in \{1,\dots,m-3\}$. Next, we compute the following:
$$
v(x)=10,\qquad
v(\{x,y\})=\frac{10+\frac14\cdot 1}{1+\frac14}=\frac{41}{5},
$$

$$
v(y)=1,\qquad
v(\{y,z\})=\frac{\frac14\cdot 1+0}{\frac14+1}=\frac15,
\qquad
v(z)=0.
$$

Similar to Part (vi), first, we show that $v(x)>v(\{x,y\})>v(B)$ for all $B \in \mathcal{A} \setminus \{\{x\},\{x,y\}\}$.  From the above computation, it follows that $v(x)>v(\{x,y\})$. Consider any $B \in \mathcal{A} \setminus \{\{x\},\{x,y\}\}$. We will show that $v(\{x,y\})>v(B)$. If $x\notin B$, then $v(a)\leq 4$ for every $a \in B$. Therefore, $v(B)\leq 4<\dfrac{41}{5}=v(\{x,y\})$. 

Now, suppose $x\in B$. Set $S:=X\setminus\{x\}$.
Then, it must be the case that $S \cap \{z,b_1,b_2,\ldots,b_{m-3}\}\neq\varnothing$ because $B \in \mathcal{A} \setminus \{\{x\},\{x,y\}\}$. Hence, $\mu(S):=\sum_{a\in S}\mu(a)\geq 1$. Moreover, $v(a)\leq 4$ for each $a \in S$. Therefore,
$
v(B)\leq \frac{10+4\mu(S)}{1+\mu(S)}$. Since the function $g(u)=\frac{10+4u}{1+u}$ is strictly decreasing on $[1,\infty)$ as shown in the previous part (vi), it follows that
$
v(B)\leq g(\mu(S))\leq g(1)=7<\frac{41}{5}$. Therefore,$v(\{x,y\})>v(B)$.

Next, we show that $v(B)>v(y)>v(\{y,z\})>v(z)$ for all $B \in \mathcal{A} \setminus \{\{z\}, \{y,z\}, \{y\}\}$. We already know from our computation (before) that $v(y)=1> v(\{y,z\})=\frac15> v(z)=0$. Consider any $B \in \mathcal{A} \setminus \{\{z\}, \{y,z\}, \{y\}\}$.

If $z\notin B$, then it must be the case that $B \cap \{x,b_1,b_2,\ldots,b_{m-3}\} \neq \emptyset$. This, together with the fact that $v(a)\geq 1$ for all $a \in B$ (where equality holds if and only if $a=y$) implies that $v(B)>1=v(\{y\})$.

Now suppose $z\in B$. Set $S:=X\setminus\{z\}$. Since $B \in \mathcal{A} \setminus \{\{z\}, \{y,z\}, \{y\}\}$, it must be the case that $B \cap \{x,b_1,b_2,\ldots,b_{m-3}\} \neq \emptyset$. Let $a_0 \in B \cap \{x,b_1,b_2,\ldots,b_{m-3}\}$. Then, by the choice of $v$ and $\mu$, it follows that $v(a_0)\geq 3+ \frac{1}{m} \mbox{ and } \mu(a_0)=1$. Also, $v(a)\geq 1$ for all $a \in S$. Therefore, it must be the case that
$$
\sum_{a\in S}\mu(a)v(a)= \sum_{a\in S \setminus \{a_0\}}\mu(a)v(a) + \mu(a_0)v(a_0) \geq \sum_{a\in S \setminus \{a_0\}}\mu(a) + \left(3+ \frac{1}{m}\right)= \sum_{a\in S}\mu(a)+ \left(2+ \frac{1}{m}\right)
$$

Hence,
$$
v(B)
=
\frac{\sum_{a\in S}\mu(a)v(a)}{1+\sum_{a\in S}\mu(a)}
\geq
\frac{\sum_{a\in S}\mu(a)+2+\frac{1}{m}}{1+\sum_{a\in S}\mu(a)}
>1=v(y),
$$
The strict inequality follows from the fact that $2+\frac{1}{m}>1$ for any natural number $m$. Therefore, $v(B)>v(y)$. This completes the proof.
  \end{enumerate}
The proof of the lemma is now complete. 
\end{proof}

\section{Proof of Theorem \ref{du_n}}\label{ap_dictatorial}
\begin{proof}  Random dictatorial rules are unanimous and strategy-proof on any domain. So, we only prove the necessary part here. We prove the result by induction. The base case, when the number of agents is 2 has already been proved in Section \ref{sec:du}. Assume that the theorem holds for all sets with $k<n$ agents.  Let $N^* = {\{1,3,\ldots,n\}}$ and define the PSCC $g :
\du^{n-1} \to \Delta A$ for the set of agents in
$N^*$ as follows: For all $R_{N^*}=(R_1,R_{-\{1,2\}}) \in \du^{n- 1}$,
\begin{equation*}
g(R_1,R_{-\{1,2\}}) = \varphi(R_1,R_1,R_{-\{1,2\}}).
\end{equation*}
\begin{claim}\label{du_cl_1}
    $g$ is a random dictatorship.
\end{claim}
 
\textbf{Proof of the claim:} Note that $g$ inherits unanimity from $\varphi$. We show that $g$ also inherits strategy-proofness. It is easy to see that agents other than $1$ cannot manipulate $g$
since  $\varphi$ is strategy-proof. Let $(R_1,R_{-\{1,2\}}) \in \du^{n - 1}$ and $Q_1 \in \du$. For all $X \in \mathcal{A}$,
we have
\begin{align}
g_{U(X,R_1)}(R_1,R_{-\{1,2\}}) & =
\varphi_{U(X,R_1)}(R_1,R_1,R_{-\{1,2\}})  \nonumber \\
& \geq \varphi_{U(X,R_1)}(Q_1,R_1,R_{-\{1,2\}}) \nonumber \\
&\geq \varphi_{U(X,R_1)}(Q_1,Q_1,R_{-\{1,2\}}) \nonumber \\
&=g_{U(X,R_1)}(Q_1,R_{-\{1,2\}}), \nonumber
\end{align}
\noindent where the inequalities follow from the strategy-proofness of $\varphi$. The proof of Claim \ref{du_cl_1} is now complete by the induction hypothesis.\footnote{We have included the proof of Claim \ref{cl_1} for completeness. It can also be found in \cite{sen2011gibbard}.} \hfill $\square$

Let $\{\beta_i\}_{i\in N^*}$ be the coefficients associated with the random dictatorship $g$. For all profile $R_N\in \du^n$ and all $X\in \mathcal{A}$, let $\beta^X(R_N):=\sum_{\{i\in \{3,\ldots,n\}\mid \tau(P_i)=X\}}\beta_i$. This, together with the definition of $g$ and the fact that $g$ is a random dictatorship, implies for a profile $(R_1,R_1,R_{-\{1,2\}})\in \du^{n}$ with $\tau(R_1)=\{a\}$ for some $a\in A$ and $b\neq a$, 
\begin{align}\label{du_eq_1}
    & \varphi_{\{a\}}(R_1,R_1,R_{-\{1,2\}})=g_{\{a\}}(R_1,R_{-\{1,2\}})=\beta_1+\beta^{\{a\}}(R_1,R_1,R_{-\{1,2\}}), \text{ and } \nonumber\\ 
    & \varphi_{\{b\}}(R_1,R_1,R_{-\{1,2\}})=g_{\{b\}}(R_1,R_{-\{1,2\}})=\beta^{\{b\}}(R_1,R_1,R_{-\{1,2\}}).
\end{align}
   Two more observations that will be frequently used in the rest of the proof: (i) since $|\tau(R)|=1$ for all $R\in \du$, we have $\beta^X(R_N)=0$ for all $R_N\in \du^n$ and all $X \in \mathcal{A} $ with $|X|\geq 2$, and (ii) for any two profiles $R_N$ and $R_N'$ in $\du$ with $\tau(R_i)=\tau(R_i')$ for all $i\in \{3,\ldots,n\}$, we have $\beta^X(R_N)=\beta^X(R_N')$ for all $X\in \mathcal{A}$.
\begin{lemma}\label{du_lem_1}
    Let $a\in A$ and $R_N\in \du^n$ be such that $\tau(R_1)=\tau(R_2)=\{a\}$. Then
    \begin{enumerate}
        \item [(i)] $ \varphi_{\{a\}}(R_N)=\beta_{1}+\beta^{\{a\}}(R_N)$.
        \item [(ii)] $ \varphi_{X}(R_N)=\beta^{X}(R_N)$ for all $X\in \mathcal{A}\setminus\{a\}$
    \end{enumerate}
\end{lemma}

\begin{proof}
     Let $R_N\in \du^n$ be a profile with $\tau(R_1)=\tau(R_2)=\{a\}$. Strategy-proofness of $\varphi$ implies, $\varphi_{\{a\}}(R_N)=\varphi_{\{a\}}(R_1,R_1,R_{-\{1,2\}})$. Moreover, by (\ref{du_eq_1}), $\varphi_{\{a\}}(R_1,R_1,R_{-\{1,2\}})=\beta_{1}+\beta^{\{a\}}(R_1,R_1,R_{-\{1,2\}})$. Now as $\beta^{\{a\}}(R_1,R_1,R_{-\{1,2\}})=\beta^{\{a\}}(R_N)$, we have $ \varphi_{\{a\}}(R_N)=\beta_{1}+\beta^{\{a\}}(R_N)$, completing the proof of (i).

     Next, we show (ii), i.e., for all $X\in \mathcal{A}\setminus \{a\}$, $ \varphi_{X}(R_N)=\beta^{X}(R_N)$. Let $Y\in \mathcal{A}\setminus \{a\}$ be such that $\varphi_Y(R_N)<\beta^Y(R_N)$. Since $\beta^Y(R_N)=0$ for all $Y$ with $|Y|\geq 2$, $Y$ must be a singleton set. Let $Y=\{b\}$ for some $b\neq a$ and let  $\bar{R}_1\equiv \{a\}\cdots \{b\}\in \du$.  By strategy-proofness of $\varphi$, $\varphi_{\{b\}}(R_N)\geq \varphi_{\{b\}}(\bar{R}_1,R_2,R_{-\{1,2\}})\geq \varphi_{\{b\}}(\bar{R}_1,\bar{R}_1,R_{-\{1,2\}})$. This, together with $\beta^{\{b\}}(R_N)>\varphi_{\{b\}}(R_N)$, yields $\beta^{\{b\}}(R_N)>\varphi_{\{b\}}(\bar{R}_1,\bar{R}_1,R_{-\{1,2\}})$. However, this is a contradiction as $\tau(\bar{R}_1)\neq\{b\}$ and thus, by (\ref{du_eq_1}), $$\varphi_{\{b\}}(\bar{R}_1,\bar{R}_1,R_{-\{1,2\}})=\beta^{\{b\}}(\bar{R}_1,\bar{R}_1,R_{-\{1,2\}})=\beta^{\{b\}}(R_N).$$
Hence, $\varphi_{X}(R_N)\geq \beta^{X}(R_N)$ for all $X\in \mathcal{A}\setminus \{a\}$. Now, as $\beta^{W}(R_N)=\beta^{W}(R_1,R_1,R_{-\{1,2\}})$ for all $W\in \mathcal{A}$ and $\varphi_Z(R_1,R_1,R_{-\{1,2\}})=\beta^{Z}(R_1,R_1,R_{-\{1,2\}})$ for all  $Z\in \mathcal{A}\setminus \{a\}$, we have $\varphi_{X}(R_N)\geq \varphi_{X}(R_1,R_1,R_{-\{1,2\}})$ for all $X\in \mathcal{A}\setminus \{a\}$. This, together with $\varphi_{\{a\}}(R_N)=\varphi_{\{a\}}(R_1,R_1,R_{-\{1,2\}})$ (shown in (i)), implies that $\varphi_{X}(R_N)= \beta^{X}(R_N)$ for all $X\in \mathcal{A}\setminus \{a\}$. This completes the proof of (ii).
\end{proof}
   
\begin{lemma}\label{du_lem_2}
 Let $a,b\in A$ be distinct and consider three preferences $\bar{R}_1\equiv \{a\}\{a,b\}\cdots$, $\bar{R}_2\equiv \{b\}\{a,b\}\cdots$, and $\hat{R}_2\equiv \{b\}\cdots\{a\}$ in $\du$. Then, for all $X\in \{\{a\},\{b\},\{a,b\}\}$ and all $R_{-\{1,2\}}\in \du^{n-2}$, $$\varphi_X(\bar{R}_1,\bar{R}_2,R_{-\{1,2\}})=\varphi_X(\bar{R}_1,\hat{R}_2,R_{-\{1,2\}}).$$
\end{lemma}

\begin{proof}
Fix $R_{-\{1,2\}}\in \du^{n-2}$. For $X=\{b\}$, note that, as $\tau(\bar{R}_2)=\tau(\hat{R}_2)=\{b\}$, by strategy-proofness of $\varphi$, we have $\varphi_{\{b\}}(\bar{R}_1,\bar{R}_2,R_{-\{1,2\}})=\varphi_{\{b\}}(\bar{R}_1,\hat{R}_2,R_{-\{1,2\}})$. Thus, to complete the proof of the lemma, it remains to show that
\begin{equation}\label{du_eq_1.5}
    \varphi_X(\bar{R}_1,\bar{R}_2,R_{-\{1,2\}})=\varphi_X(\bar{R}_1,\hat{R}_2,R_{-\{1,2\}}) \text{ for all } X\in \{\{a\},\{a,b\}\}.
\end{equation}
Consider a preference $\tilde{R}_1\equiv \{a\}\{a,b\}\{b\}\cdots\in \du$. As $\bar{R}_1\equiv \{a\}\{a,b\}\cdots$, in view of strategy-proofness of $\varphi$, showing (\ref{du_eq_1.5}) is equivalent to showing 
\begin{equation}\label{du_eq_1.7}
    \varphi_X(\tilde{R}_1,\bar{R}_2,R_{-\{1,2\}})=\varphi_X(\tilde{R}_1,\hat{R}_2,R_{-\{1,2\}}) \text{ for all } X\in \{\{a\},\{a,b\}\}.
\end{equation}
Let $\tilde{R}_2\equiv \{b\}\{a,b\}\{a\}\cdots\in \du$. By strategy-proofness of $\varphi$, for all $R_2\equiv \{b\}\cdots$, $$\varphi_{\{\{a\},\{a,b\},\{b\}\}}(\tilde{R}_1,R_2,R_{-\{1,2\}})=\varphi_{\{\{a\},\{a,b\},\{b\}\}}(\tilde{R}_2,R_2,R_{-\{1,2\}}).$$
Moreover, as $\tau(\tilde{R}_2)=\tau(R_2)=\{b\}$, Lemma \ref{du_lem_1} gives us  $$\varphi_{\{\{a\},\{a,b\},\{b\}\}}(\tilde{R}_2,R_2,R_{-\{1,2\}})=\beta_1+\beta^{\{a\}}(\tilde{R}_2,R_2,R_{-\{1,2\}})+\beta^{\{b\}}(\tilde{R}_2,R_2,R_{-\{1,2\}}).$$ 
As $\beta^X(\tilde{R}_2,R_2,R_{-\{1,2\}})=\beta^X(R_2,R_2,R_{-\{1,2\}})$ for all $X\in \mathcal{A}$, combining the above two observations, we have for all $R_2\equiv \{b\}\cdots$,
\begin{equation}\label{du_eq_2}
    \varphi_{\{\{a\},\{a,b\},\{b\}\}}(\tilde{R}_1,R_2,R_{-\{1,2\}})=\beta_1+\beta^{\{a\}}(R_2,R_2,R_{-\{1,2\}})+\beta^{\{b\}}(R_2,R_2,R_{-\{1,2\}}).
\end{equation}
Further, as $\tau(\bar{R}_2)=\tau(\hat{R}_2)=\{b\}$, considering $\bar{R}_2$ and $\hat{R}_2$ in place of $R_2$ in (\ref{du_eq_2}), we get 
\begin{align}\label{du_eq_3}
     & \varphi_{\{\{a\},\{a,b\},\{b\}\}}(\tilde{R}_1,\bar{R}_2,R_{-\{1,2\}})=\beta_1+\beta^{\{a\}}(\bar{R}_2,\bar{R}_2,R_{-\{1,2\}})+\beta^{\{b\}}(\bar{R}_2,\bar{R}_2,R_{-\{1,2\}}) \text{ and } \nonumber\\
     & \varphi_{\{\{a\},\{a,b\},\{b\}\}}(\tilde{R}_1,\hat{R}_2,R_{-\{1,2\}})=\beta_1+\beta^{\{a\}}(\hat{R}_2,\hat{R}_2,R_{-\{1,2\}})+\beta^{\{b\}}(\hat{R}_2,\hat{R}_2,R_{-\{1,2\}}).
\end{align}
   However, as $\beta^X(\bar{R}_2,\bar{R}_2,R_{-\{1,2\}})=\beta^X(\hat{R}_2,\hat{R}_2,R_{-\{1,2\}})$ for all $X\in \mathcal{A}$, (\ref{du_eq_3}) yields $$\varphi_{\{\{a\},\{a,b\},\{b\}\}}(\tilde{R}_1,\bar{R}_2,R_{-\{1,2\}})=\varphi_{\{\{a\},\{a,b\},\{b\}\}}(\tilde{R}_1,\hat{R}_2,R_{-\{1,2\}}).$$
   As, by strategy-proofness, $\varphi_{\{b\}}(\tilde{R}_1,\bar{R}_2,R_{-\{1,2\}})=\varphi_{\{b\}}(\tilde{R}_1,\hat{R}_2,R_{-\{1,2\}})$, the above equation reduces to 
   \begin{equation}\label{du_eq_4}
       \varphi_{\{\{a\},\{a,b\}\}}(\tilde{R}_1,\bar{R}_2,R_{-\{1,2\}})=\varphi_{\{\{a\},\{a,b\}\}}(\tilde{R}_1,\hat{R}_2,R_{-\{1,2\}}).
   \end{equation}
We claim that $\varphi_{\{a,b\}}(\tilde{R}_1,\bar{R}_2,R_{-\{1,2\}})=\varphi_{\{a,b\}}(\tilde{R}_1,\hat{R}_2,R_{-\{1,2\}})$. Suppose not. Then, as $\bar{R}_2\equiv \{b\}\{a,b\}\ldots$, it must be that $\varphi_{\{a,b\}}(\tilde{R}_1,\bar{R}_2,R_{-\{1,2\}})>\varphi_{\{a,b\}}(\tilde{R}_1,\hat{R}_2,R_{-\{1,2\}})$. This, together with (\ref{du_eq_4}), yields $\varphi_{\{a\}}(\tilde{R}_1,\bar{R}_2,R_{-\{1,2\}})<\varphi_{\{a\}}(\tilde{R}_1,\hat{R}_2,R_{-\{1,2\}})$. However, as $\hat{R}_2\equiv \{b\}\cdots\{a,b\}\{a\}$, this means agent 2 manipulates at $(\tilde{R}_1,\hat{R}_2,R_{-\{1,2\}})$ via $\bar{R}_2$, a contradiction. Thus, $\varphi_{\{a,b\}}(\tilde{R}_1,\bar{R}_2,R_{-\{1,2\}})=\varphi_{\{a,b\}}(\tilde{R}_1,\hat{R}_2,R_{-\{1,2\}})$ and hence, by (\ref{du_eq_4}), $\varphi_{\{a\}}(\tilde{R}_1,\bar{R}_2,R_{-\{1,2\}})=\varphi_{\{a\}}(\tilde{R}_1,\hat{R}_2,R_{-\{1,2\}})$. Therefore, (\ref{du_eq_1.7}) holds. \end{proof}

\begin{lemma}\label{du_lem_3}
 Let $a,b\in A$ be distinct and consider two preferences $R_1=\{a\}\{a,b\}\{b\}\cdots$, $R_2=\{b\}\{a,b\}\{a\}\cdots$ in $\du$. Then,  for all $R_{-\{1,2\}}\in \du^{n-2}$, $\varphi_{\{a,b\}}(R_1,R_2,R_{-\{1,2\}})=0$.
\end{lemma}

\begin{proof}
In this proof, we utilize several preference orderings within the $\du$ domain. The existence of these preferences is established in Appendix \ref{appen_2}. To facilitate the reader's understanding, we first provide a consolidated list of these preferences below for ease of reference throughout the proof. We write them in the order they are used in the proof. 
\begin{itemize}
    \item $R_1'\equiv \{a\}\{a,x\}\{a,b\}\cdots$,
    \item $R_2'\equiv \{a\}\cdots\{a,b\}\{b\}$,
    \item $R_2''\equiv \{x\}\{a,x\}\cdots\{a\}\{a,b\}\{b\}$,
    \item $\tilde{R}_2\equiv \{x\}\{a,x\}\cdots\{b\}\{a,b\}\{a\}$,
    \item $\bar{R}_1\equiv \{a\}\{a,b\}\cdots\{a,x\}\{x\}$, and 
    \item $\hat{R}_2\equiv \{b\}\cdots\{a,b\}\{a\}$.
    \end{itemize}
Fix $R_{-\{1,2\}}\in \du^{n-2}$. Let $x \in A\setminus \{a,b\}$. Note that such a choice is feasible as $|A|\geq 3$. Note that as $\tau(R_1')=\tau(R_2')=\{a\}$, by part (ii) of Lemma \ref{du_lem_1}, $\varphi_{\{a,b\}}(R_1',R_2',R_{-\{1,2\}})=\beta^{\{a,b\}}(R_1',R_2',R_{-\{1,2\}})=0$. Further, as $R_2'\equiv \{a\}\cdots\{a,b\}\{b\}$ and $R_2''\equiv \{x\}\{a,x\}\cdots\{a\}\{a,b\}\{b\}$, by strategy-proofness of $\varphi$ and $\varphi_{\{a,b\}}(R_1',R_2',R_{-\{1,2\}})=0$, we have $\varphi_{\{a,b\}}(R_1',R_2'',R_{-\{1,2\}})=0$.   Recall $\tilde{R}_2\equiv \{x\}\{a,x\}\cdots\{b\}\{a,b\}\{a\}$, and we claim $\varphi_{\{a,b\}}({R}_1',{R}''_2,R_{-\{1,2\}})\geq \varphi_{\{a,b\}}({R}_1',\tilde{R}_2,R_{-\{1,2\}})$, which implies $\varphi_{\{a,b\}}({R}_1',\tilde{R}_2,R_{-\{1,2\}})=0$. To see this, first note that by the strategy-proofness of $\varphi$,
    \begin{equation}\label{du_eq_5}
        \varphi_{\{\{a\},\{a,b\},\{b\}\}}({R}_1',{R}''_2,R_{-\{1,2\}})= \varphi_{\{\{a\},\{a,b\},\{b\}\}}({R}_1',\tilde{R}_2,R_{-\{1,2\}}).
    \end{equation}
    Moreover, as $R_1'\equiv  \{a\}\{a,x\}\{a,b\}\cdots$, $R_2''\equiv \{x\}\{a,x\}\cdots\{a\}\{a,b\}\{b\}$, and $\tilde{R}_2=\{x\}\{a,x\}\cdots\{b\}\{a,b\}\{a\}$, by Lemma \ref{du_lem_2}, $\varphi_{\{a\}}({R}_1',{R}''_2,R_{-\{1,2\}})= \varphi_{\{a\}}({R}_1',\tilde{R}_2,R_{-\{1,2\}})$. Combining this with (\ref{du_eq_5}), we have
    $$\varphi_{\{\{a,b\},\{b\}\}}({R}_1',{R}''_2,R_{-\{1,2\}})= \varphi_{\{\{a,b\},\{b\}\}}({R}_1',\tilde{R}_2,R_{-\{1,2\}}).$$
    Hence, if $\varphi_{\{a,b\}}({R}_1',{R}''_2,R_{-\{1,2\}})< \varphi_{\{a,b\}}({R}_1',\tilde{R}_2,R_{-\{1,2\}})$, it must be that  $\varphi_{\{b\}}({R}_1',{R}''_2,R_{-\{1,2\}})> \varphi_{\{b\}}({R}_1',\tilde{R}_2,R_{-\{1,2\}})$. But this leads to a manipulation of $\varphi$ by agent 2 at $({R}_1',{R}''_2,R_{-\{1,2\}})$ via $\tilde{R}_2$, as $R''_2\equiv\{x\}\{a,x\}\cdots\{a\}\{a,b\}\{b\}$, a contradiction. Therefore, we have
    \begin{equation}\label{du_eq_6}
        \varphi_{\{a,b\}}({R}_1',\tilde{R}_2,R_{-\{1,2\}})=0.
    \end{equation}
    
    Now, recall $\bar{R}_1\equiv \{a\}\{a,b\}\cdots\{a,x\}\{x\}$. As $R_1'\equiv  \{a\}\{a,x\}\{a,b\}\cdots$ and $\tilde{R}_2=\{x\}\{a,x\}\cdots\{b\}\{a,b\}\{a\}$, by Lemma \ref{du_lem_2}, $$\varphi_X({R}_1',\tilde{R}_2,R_{-\{1,2\}})= \varphi_{X}(\bar{R}_1,\tilde{R}_2,R_{-\{1,2\}}) \text{ for all }X\in \{\{a\},\{a,x\}\}.$$
   This, together with the strategy-proofness of $\varphi$ and $R_1'\equiv  \{a\}\{a,x\}\{a,b\}\cdots$, implies $\varphi_{\{a,b\}}({R}_1',\tilde{R}_2,R_{-\{1,2\}})\geq \varphi_{\{a,b\}}(\bar{R}_1,\tilde{R}_2,R_{-\{1,2\}})$. However, as $\varphi_{\{a,b\}}({R}_1',\tilde{R}_2,R_{-\{1,2\}})=0$ (shown in (\ref{du_eq_6})), we have $\varphi_{\{a,b\}}(\bar{R}_1,\tilde{R}_2,R_{-\{1,2\}})=0$.  Finally, we use $\hat{R}_2\equiv \{b\}\cdots\{a,b\}\{a\}$. As $\tilde{R}_2=\{x\}\{a,x\}\cdots\{b\}\{a,b\}\{a\}$, by strategy-proofness of $\varphi$, we have $\varphi_{\{a,b\}}(\bar{R}_1,\tilde{R}_2,R_{-\{1,2\}})= \varphi_{\{a,b\}}(\bar{R}_1,\hat{R}_2,R_{-\{1,2\}})$, which in turn implies, $\varphi_{\{a,b\}}(\bar{R}_1,\hat{R}_2,R_{-\{1,2\}})=0$.

    We now complete the proof of the lemma. Recall the two preferences stated in the statement of the lemma, $R_1\equiv \{a\}\{a,b\}\ldots$ and $R_2\equiv \{b\}\{a,b\}\ldots$. As $\bar{R}_1\equiv \{a\}\{a,b\}\cdots\{a,x\}\{x\}$ and $\hat{R}_2\equiv \{b\}\cdots\{a,b\}\{a\}$, by Lemma \ref{du_lem_2}, $\varphi_{\{a,b\}}(\bar{R}_1,\hat{R}_2,R_{-\{1,2\}})=\varphi_{\{a,b\}}(R_1,R_2,R_{-\{1,2\}})$. However, we have already shown the previous paragraph that $\varphi_{\{a,b\}}(\bar{R}_1,\hat{R}_2,R_{-\{1,2\}})=0$. Hence, $\varphi_{\{a,b\}}(R_1,R_2,R_{-\{1,2\}})=0$. \end{proof}

\begin{lemma}\label{du_lem_4}
 Let $a,b\in A$ be distinct and consider $\bar{R}_1=\{a\}\{a,b\}\{b\}\cdots$, $\bar{R}_2=\{b\}\{a,b\}\{a\}\cdots$ in $\du$. Then, for all $R_N\in \du^n$ with $R_1\equiv \{a\}\cdots$ and $R_2\equiv \{b\}\cdots$, we have
 \begin{enumerate}[(i)]
     \item $\varphi_X(R_N)=\varphi_X(\bar{R}_1,\bar{R}_2,R_{-\{1,2\}})$ for all $X\in \{\{a\},\{b\},\{a,b\}\}$. 
     \item $\varphi_{\{\{a\},\{b\}\}}(R_N)=\beta_{1}+\beta^{\{a\}}(R_N)+\beta^{\{b\}}(R_N)$.
 \end{enumerate} 
\end{lemma}

\begin{proof}
Fix $R_N\in \du^n$ with $R_1\equiv \{a\}\cdots$ and $R_2\equiv \{b\}\cdots$. We start with showing (i) of the lemma in two parts. In the first part, we show that
\begin{equation}\label{du_eq_7}
    \varphi_X(\bar{R}_1,\bar{R}_2,R_{-\{1,2\}})=\varphi_X(R_1,\bar{R}_2,R_{-\{1,2\}}) \text{ for all } X\in \{\{a\},\{b\},\{a,b\}\},
\end{equation}
and then in the second part we further show that
\begin{equation}\label{du_eq_8}
    \varphi_X(R_1,\bar{R}_2,R_{-\{1,2\}})=\varphi_X(R_1,R_2,R_{-\{1,2\}}) \text{ for all } X\in \{\{a\},\{b\},\{a,b\}\}.
\end{equation}
Note that (\ref{du_eq_7}) and (\ref{du_eq_8}) together imply (i) of the lemma. To see (\ref{du_eq_7}), as $R_1\equiv \{a\}\ldots$ and $\bar{R}_1\equiv \{a\}\{a,b\}\{b\}\cdots$, applying strategy-proofness of $\varphi$, we have
    \begin{align*}
        &\varphi_{\{a\}}(\bar{R}_1,\bar{R}_2,R_{-\{1,2\}})=\varphi_{\{a\}}(R_1,\bar{R}_2,R_{-\{1,2\}}), \\
        &\varphi_{\{a,b\}}(\bar{R}_1,\bar{R}_2,R_{-\{1,2\}})\geq \varphi_{\{a,b\}}(R_1,\bar{R}_2,R_{-\{1,2\}}) \mbox{ and }\\
       & \varphi_{\{\{a,b\},\{b\}\}}(\bar{R}_1,\bar{R}_2,R_{-\{1,2\}})\geq \varphi_{\{\{a,b\},\{b\}\}}(R_1,\bar{R}_2,R_{-\{1,2\}}).
    \end{align*} 
    Since by Lemma \ref{du_lem_3}, $\varphi_{\{a,b\}}(\bar{R}_1,\bar{R}_2,R_{-\{1,2\}})=0$, we have $\varphi_{\{a,b\}}(R_1,\bar{R}_2,R_{-\{1,2\}})=0$. Hence,  $\varphi_{\{b\}}(\bar{R}_1,\bar{R}_2,R_{-\{1,2\}})\geq \varphi_{\{b\}}(R_1,\bar{R}_2,R_{-\{1,2\}})$. We prove they are equal. Suppose not and $\varphi_{\{b\}}(\bar{R}_1,\bar{R}_2,R_{-\{1,2\}})> \varphi_{\{b\}}(R_1,\bar{R}_2,R_{-\{1,2\}})$. Consider a preference $\hat{R}_1\equiv \{a\}\ldots\{b\}$. By Lemma \ref{du_lem_2}, $\varphi_{\{b\}}(\bar{R}_1,\bar{R}_2,R_{-\{1,2\}})= \varphi_{\{b\}}(\hat{R}_1,\bar{R}_2,R_{-\{1,2\}})$. This means $\varphi_{\{b\}}(\hat{R}_1,\bar{R}_2,R_{-\{1,2\}})> \varphi_{\{b\}}(R_1,\bar{R}_2,R_{-\{1,2\}})$. However, as $\hat{R}_1\equiv \{a\}\ldots\{b\}$, this leads to a manipulation of $\varphi$ at $(\hat{R}_1,\bar{R}_2,R_{-\{1,2\}})$ via $R_1$, a contradiction. So, $\varphi_{\{b\}}(\bar{R}_1,\bar{R}_2,R_{-\{1,2\}})= \varphi_{\{b\}}(R_1,\bar{R}_2,R_{-\{1,2\}})$. This completes the verification of (\ref{du_eq_7}).

    Now, consider (\ref{du_eq_8}). As $R_2\equiv \{b\}\ldots$ and $\bar{R}_2\equiv \{b\}\{a,b\}\{a\}\cdots$, applying strategy-proofness of $\varphi$, we have 
    \begin{align*}
        &\varphi_{\{b\}}(R_1,\bar{R}_2,R_{-\{1,2\}})=\varphi_{\{b\}}(R_1,R_2,R_{-\{1,2\}}), \\
        &\varphi_{\{a,b\}}(R_1,\bar{R}_2,R_{-\{1,2\}})\geq \varphi_{\{a,b\}}(R_1,R_2,R_{-\{1,2\}}) \mbox{ and }\\
       & \varphi_{\{\{a,b\},\{a\}\}}(R_1,\bar{R}_2,R_{-\{1,2\}})\geq \varphi_{\{\{a,b\},\{a\}\}}(R_1,R_2,R_{-\{1,2\}}).
    \end{align*} 
     Since by Lemma \ref{du_lem_3}, $\varphi_{\{a,b\}}(\bar{R}_1,\bar{R}_2,R_{-\{1,2\}})=0$, and by (\ref{du_eq_7}), $\varphi_{\{a,b\}}(\bar{R}_1,\bar{R}_2,R_{-\{1,2\}})=\varphi_{\{a,b\}}(R_1,\bar{R}_2,R_{-\{1,2\}})$, we have $\varphi_{\{a,b\}}(R_1,\bar{R}_2,R_{-\{1,2\}})=0$. Hence, $\varphi_{\{a,b\}}(R_1,R_2,R_{-\{1,2\}})=0$ and  $\varphi_{\{a\}}(R_1,\bar{R}_2,R_{-\{1,2\}})\geq \varphi_{\{a\}}(R_1,R_2,R_{-\{1,2\}})$. We show that they are equal. Suppose not and $\varphi_{\{a\}}(R_1,\bar{R}_2,R_{-\{1,2\}})> \varphi_{\{a\}}(R_1,R_2,R_{-\{1,2\}})$. Consider two preferences $\hat{R}_1\equiv \{a\}\cdots\{b\}$ and $\hat{R}_2\equiv \{b\}\{a,b\}\cdots\{a\}$ in $\du$. Then 
      \begin{align} \label{du_eq_9}
          &\varphi_{\{a\}}(R_1,\bar{R}_2,R_{-\{1,2\}})> \varphi_{\{a\}}(R_1,R_2,R_{-\{1,2\}}) \nonumber\\
        \implies  & \varphi_{\{a\}}(\hat{R}_1,\bar{R}_2,R_{-\{1,2\}})> \varphi_{\{a\}}(\hat{R}_1,R_2,R_{-\{1,2\}}) \hspace{5mm} (\text{by strategy-proofness of } \varphi) \nonumber\\
       \implies   & \varphi_{\{a\}}(\hat{R}_1,\hat{R}_2,R_{-\{1,2\}})> \varphi_{\{a\}}(\hat{R}_1,R_2,R_{-\{1,2\}}) \hspace{5mm} (\text{as by Lemma }\ref{du_lem_2}, \varphi_{\{a\}}(\hat{R}_1,\bar{R}_2,R_{-\{1,2\}})=\varphi_{\{a\}}(\hat{R}_1,\hat{R}_2,R_{-\{1,2\}})).
      \end{align}
However, (\ref{du_eq_9}) is a contradiction to the strategy-proofness of $\varphi$ as $\hat{R}_2\equiv \{b\}\{a,b\}\cdots\{a\}$. Thus, (\ref{du_eq_8}) holds and the proof of (i) of the lemma is complete.

We now show (ii) of the lemma. Note that as $\varphi_{\{a,b\}}(R_N)=0$, we have $\varphi_{\{\{a\},\{b\}\}}(R_N)=\varphi_{\{\{a\},\{a,b\},\{b\}\}}(R_N)$. This, together with (i) of the lemma, implies $\varphi_{\{\{a\},\{b\}\}}(R_N)=\varphi_{\{\{a\},\{a,b\},\{b\}\}}(\bar{R}_1,\bar{R}_2,R_{-\{1,2\}})$. Moreover, as $\bar{R}_1\equiv \{a\}\{a,b\}\{b\}\cdots$ and   $\bar{R}_2\equiv \{b\}\{a,b\}\{a\}\cdots$, by strategy-proofness of $\varphi$ and the previous equality, we have
\begin{equation}\label{du_eq_10}
    \varphi_{\{\{a\},\{b\}\}}(R_N)=\varphi_{\{\{a\},\{a,b\},\{b\}\}}(\bar{R}_1,\bar{R}_1,R_{-\{1,2\}}).
\end{equation}
Further, as $\tau(\bar{R}_1)=\{a\}$, by Lemma \ref{du_lem_1}, $\varphi_{\{\{a\},\{a,b\},\{b\}\}}(\bar{R}_1,\bar{R}_1,R_{-\{1,2\}})=\beta_1+\beta^{\{a\}}(\bar{R}_1,\bar{R}_1,R_{-\{1,2\}})+\beta^{\{b\}}(\bar{R}_1,\bar{R}_1,R_{-\{1,2\}})$. Now as $\beta^X(\bar{R}_1,\bar{R}_1,R_{-\{1,2\}})=\beta^X(R_N)$ for all $X\in \mathcal{A}$, by (\ref{du_eq_10}), we have $$\varphi_{\{\{a\},\{b\}\}}(R_N)=\beta_1+\beta^{\{a\}}(R_N)+\beta^{\{b\}}(R_N).$$ This completes the verification of (ii) of the lemma. \end{proof}

\begin{lemma}\label{du_lem_5}
 Let $R_N\in \du^n$. Then $\varphi_{X}(R_N)=\beta^{X}(R_N)$ for all $X\in \mathcal{A}\setminus \{\tau(R_1),\tau(R_2)\}$. 
\end{lemma}

\begin{proof}
Fix $R_N\in \du^n$. Let's assume that $\tau(R_1)=\{a\}$ and $\tau(R_2)=\{b\}$. We first show that $\varphi_{X}(R_N)\geq \beta^{X}(R_N)$ for all $X\in \mathcal{A}\setminus \{\{a\},\{b\}\}$. Assume, for contradiction, that this does not hold, and let $Y\in \mathcal{A}\setminus \{\{a\},\{b\}\}$ be such that $\beta^Y(R_N)>\varphi_Y(R_N)$. Since $\beta^Y=0$ for all $Y$ with $|Y|\geq 2$, $Y$ must be a singleton set, say $Y=\{c\}$ where $c\in A\setminus \{a,b\}$. Let $\bar{R}_1\equiv \{a\}\cdots \{c\}$ and $\bar{R}_2=\{b\}\cdots \{c\}$. Therefore, by the strategy-proofness of $\varphi$, $$\varphi_{\{c\}}(R_N)\geq \varphi_{\{c\}}(\bar{R}_1,R_2,R_{-\{1,2\}})\geq \varphi_{\{c\}}(\bar{R}_1,\bar{R}_2,R_{-\{1,2\}})=\varphi_{\{c\}}(\bar{R}_1,\bar{R}_1,R_{-\{1,2\}}).$$
This, together with $\beta^{\{c\}}(R_N)>\varphi_{\{c\}}(R_N)$, implies $\beta^{\{c\}}(R_N)>\varphi_{\{c\}}(\bar{R}_1,\bar{R}_1,R_{-\{1,2\}})$. Moreover, as $\tau(\bar{R}_1)=\{a\}$ and $a\neq c$, by Lemma \ref{du_lem_1}, we have $\varphi_{\{c\}}(\bar{R}_1,\bar{R}_1,R_{-\{1,2\}})=\beta^{\{c\}}(\bar{R}_1,\bar{R}_1,R_{-\{1,2\}})$. However, this is a contradiction as $\beta^{\{c\}}(R_N)=\beta^{\{c\}}(\bar{R}_1,\bar{R}_1,R_{-\{1,2\}})$. Thus, $\varphi_{X}(R_N)\geq \beta^{X}(R_N)$ for all $X\in \mathcal{A}\setminus \{\{a\},\{b\}\}$. We now show that the equality holds for all $X\in \mathcal{A}\setminus \{\{a\},\{b\}\}$, as claimed in the statement of the lemma. To see this, first observe that for any such $X$, either, if $a\in X$ then $|X|\geq 2$, and hence, $\varphi_Z(\bar{R}_1,\bar{R}_1,R_{-\{1,2\}})=0=\beta^{Z}(\bar{R}_1,\bar{R}_1,R_{-\{1,2\}})$, or if $a\notin X$, by (ii) Lemma  \ref{du_lem_1}, $\varphi_X(\bar{R}_1,\bar{R}_1,R_{-\{1,2\}})=\beta^{X}(\bar{R}_1,\bar{R}_1,R_{-\{1,2\}})$. This observation, together with $\beta^{X}(R_N)=\beta^{X}(\bar{R}_1,\bar{R}_1,R_{-\{1,2\}})$ for all $X\in \mathcal{A}$, implies that $\varphi_X(\bar{R}_1,\bar{R}_1,R_{-\{1,2\}})=\beta^{X}(R_N)$ for all $X\in \mathcal{A}\setminus \{\{a\},\{b\}\}$. Hence, 
\begin{equation}\label{du_eq_11}
    \varphi_X(R_N)\geq \varphi_X(\bar{R}_1,\bar{R}_1,R_{-\{1,2\}})= \beta^{X}(R_N) \text{ for all } X\in \mathcal{A}\setminus \{\{a\},\{b\}\}.
\end{equation}
 Now, as $\tau(\bar{R}_1)=\{a\}$, by (i) of Lemma \ref{du_lem_1}, $\varphi_{\{a\}}(\bar{R}_1,\bar{R}_1,R_{-\{1,2\}})=\beta_1+\beta^{\{a\}}(\bar{R}_1,\bar{R}_1,R_{-\{1,2\}})$ and by (ii) of Lemma \ref{du_lem_1}, $\varphi_{\{b\}}(\bar{R}_1,\bar{R}_1,R_{-\{1,2\}})=\beta^{\{b\}}(\bar{R}_1,\bar{R}_1,R_{-\{1,2\}})$. Moreover, by Lemma \ref{du_lem_4}, $\varphi_{\{\{a\},\{b\}\}}(R_N)=\beta_{1}+\beta^{\{a\}}(R_N)+\beta^{\{b\}}(R_N)$. Thus, $\varphi_{\{\{a\},\{b\}\}}(R_N)=\varphi_{\{\{a\},\{b\}\}}(\bar{R}_1,\bar{R}_1,R_{-\{1,2\}})$, as $\beta^{X}(R_N)=\beta^{X}(\bar{R}_1,\bar{R}_1,R_{-\{1,2\}})$ for all $X$. Further, combining this with (\ref{du_eq_11}) yield 
 $$\varphi_X(R_N)\geq \varphi_X(\bar{R}_1,\bar{R}_1,R_{-\{1,2\}}) \text{ for all } X\in \mathcal{A},$$
which in turn implies $\varphi_X(R_N)= \varphi_X(\bar{R}_1,\bar{R}_1,R_{-\{1,2\}})$ for all $X\in \mathcal{A}$. Therefore by (\ref{du_eq_11}), we have $\varphi_X(R_N)=\beta^{X}(R_N)$ for all $X\in \mathcal{A}\setminus \{\{a\},\{b\}\}$.\end{proof}

  We now proceed to complete the proof of the theorem. By Lemma \ref{du_lem_4} and Lemma \ref{du_lem_5}, $\varphi$ assigns positive probability only to singleton sets. Hence, we use $\varphi$ to construct a PSCF $ \widehat{\varphi}:\mathcal P^n\to\Delta A $ such that for all $a\in A$, $\widehat{\varphi}_a(Q_N)=\varphi_{\{a\}}(R_N), $ where $Q_i=R_i|_A$ for each $i\in N$. This is well-defined because $\du|_A=\mathcal P$, and strategy-proofness together with singleton support implies that $\varphi$ depends only on the induced profile $R_N|_A$. Formally, if $R_i|_A=R_i'|_A$ for some $i\in N$, then $\varphi(R_i,R_{-i})=\varphi(R_i',R_{-i})$ for all $R_{-i}\in \de^{n-1}$. Unanimity of $\widehat{\varphi}$ follows immediately from unanimity of $\varphi$. Strategy-proofness of $\widehat{\varphi}$ follows by the same argument as in the case for $n=2$. For any agent $i$, any deviation $Q_i'\in\mathcal P$, and any $a\in A$, choose $R_N\in\du^n$ and $R_i'\in\du$ such that $R_j|_A=Q_j$ for all $j$ and $R_i'|_A=Q_i'$. Then  $\widehat{\varphi}_{U(a,Q_i)}(Q_N) = \varphi_{U(\{a\},R_i)}(R_N) \geq \varphi_{U(\{a\},R_i)}(R_i',R_{-i}) = \widehat{\varphi}_{U(a,Q_i)}(Q_i',Q_{-i}),$ where the inequality follows from strategy-proofness of $\varphi$. Thus $\widehat{\varphi}$ is unanimous and strategy-proof. By \cite{gibbard1977manipulation}, $\widehat{\varphi}$ is random dictatorial. Therefore, there exist $\epsilon_1,\ldots,\epsilon_n\geq0$, with $\sum_{i=1}^n\epsilon_i=1$, such that for every $Q_N\in\mathcal P^n$ and every $a\in A$, $\widehat{\varphi}_a(Q_N) = \sum_{\{i\in N:\tau(Q_i)=a\}}\epsilon_i.$ Since $\tau(Q_i)=a$ if and only if $\tau(R_i)=\{a\}$, we get $\varphi_{\{a\}}(R_N) = \sum_{\{i\in N\mid\tau(R_i)=\{a\}\}}\epsilon_i.$  Finally, since $\varphi_X(R_N)=0$ for every non-singleton $X$, it follows that for all $R_N\in\du^n$ and all $X\in\mathcal A$, $ \varphi_X(R_N) = \sum_{\{i\in N\mid\tau(R_i)=X\}}\epsilon_i. $ Hence, $\varphi$ is random dictatorial.  \end{proof}

\section{Supplementary Material}\label{appen_supple}

The complete proofs for Theorem \ref{de_tops-only} and associated technical lemmas are hosted externally due to space and length considerations. 

The supplementary material is available as an ancillary file on the arXiv abstract page.

%     \section{Result on $\de$ domain for $m=3$ and $n=3$}\label{appen_6}
%     \begin{theorem}\label{de_m=3}
%         Let $m=n=3$. Then a PSCC $\varphi: \de^3 \to \Delta \mathcal{A}$ is unanimous and strategy-proof if and only if $\varphi$ is a random bi-dictatorial rule.
%     \end{theorem}

% \begin{proof}
% The if part of the theorem follows from Remark \ref{rem_star}. We prove the only-if part of the theorem. Let $A=\{a,b,c\}$ and let $\varphi:\de^3\to \Delta\mathcal{A}$ be a unanimous and strategy-proof PSCC. We go along the line of the proof of Theorem \ref{de_n}. The induction hypothesis considered at the beginning of the proof of  Theorem \ref{de_n} holds true by Theorem \ref{de_2}. Further, one may observe that in the proof of Theorem \ref{de_n}, the fact that $m\geq 4$ is used for the first time in Lemma \ref{de_lem_7}. Therefore, we assume that all the lemmas before Lemma \ref{de_lem_7} hold true in the current setting. \end{proof}

\end{appendix}

\end{document}